%% file: main.tex
\documentclass[11pt]{article}
\usepackage[utf8]{inputenc}
\usepackage{setspace}
\usepackage{comment}

\usepackage[left=1in,right=1in,top=1in,bottom=1in,bindingoffset=0cm]{geometry}
\usepackage{hyperref}
\hypersetup{
    colorlinks=true,
    linkcolor=blue,
    filecolor=magenta,      
    urlcolor=red,
    citecolor=gray,
    linktoc=page,
}
\usepackage{graphicx,wrapfig,lipsum,booktabs}
\graphicspath{ {images/} }
\usepackage{amsmath}
\usepackage{bm}
\usepackage{amsthm}
\usepackage{tikz}
\usepackage{pgfplots}
\pgfplotsset{compat = newest, width = 10cm}
\usepgfplotslibrary{fillbetween}
\usepackage{caption}
\usepackage{amssymb}
\usepackage{mathtools}

\usepackage{bbm}
\usepackage{enumitem}
\allowdisplaybreaks
\usepackage{breakurl}
\usepackage[backend=biber, style=alphabetic, sorting=nyt, maxnames=100,backref=true, firstinits=true, sortcites = true]{biblatex}
\DeclareFieldFormat[unpublished]{howpublished}{#1, }

\def\ER{Erd\H{o}s--R\'{e}nyi }

\def\E{\mathbb{E}}

\def\bin{{\rm Bin}}
\def\poiss{{\rm Poiss}}

\def\N{\mathbb{N}}

\def\E{\mathbb{E}}
\def\R{\mathbb{R}}
\def\P{\mathbb{P}}

\def\eps{\varepsilon}

\def\1{\mathbf{1}}

\def\tce{t_c + \eps}
\def\tce2{t_c + \frac{\eps}{2}}
\def\PGWLT{\mathbb{T}_{d}^{GW}}
\def\Hrnp{\mathcal{H}_{r}(n,p)}

\def\Hrrnp{\mathcal{H}(r,n,p)}

\def\Var{\mathsf{Var}}

\newtheorem{Theorem}{Theorem}[section]
\newtheorem{Proposition}[Theorem]{Proposition}
\newtheorem{Claim}[Theorem]{Claim}
\newtheorem{Corollary}[Theorem]{Corollary}
\newtheorem{lemma}[Theorem]{Lemma}
\newtheorem{Conjecture}{Conjecture}
\newtheorem*{Theorem*}{Theorem}

\newtheorem{Definition}[Theorem]{Definition}

\newcommand{\boldgamma}{\bm{\gamma}}
\newcommand{\set}[1]{\{#1\}}

\title{The Low-Degree Hardness of Finding Large Independent Sets in Sparse Random Hypergraphs}
\author{Abhishek Dhawan\footnote{Department of Mathematics, University of Illinois Urbana-Champaign; supported in part by the Georgia Tech ARC-ACO Fellowship, NSF grant DMS-2053333 (PI: Cheng Mao), NSF CAREER grant DMS-2239187 (PI: Anton Bernshteyn), and the NSF RTG grant DMS-1937241. Email: adhawan2@illinois.edu}
\and 
Yuzhou Wang\footnote{School of Mathematics, Georgia Institute of Technology; supported in part by NSF grant CCF-2309708 (PI: Will Perkins). Email: ywang3694@gatech.edu }}
\date{}

\begin{document}

\maketitle

\begin{abstract}
    We study the algorithmic task of finding large independent sets in sparse Erd\H{o}s--R\'enyi random $r$-uniform hypergraphs on $n$ vertices having average degree $d$. Krivelevich and Sudakov showed that the maximum independent set has density $\left(\frac{r}{r-1}\cdot\frac{\log d}{d}\right)^{1/(r-1)}$ in the double limit $n \to \infty$ followed by $d \to \infty$. We show that the class of \textit{low-degree polynomial algorithms} can find independent sets of density $\left(\frac{1}{r-1}\cdot\frac{\log d}{d}\right)^{1/(r-1)}$ but no larger. This extends and generalizes earlier results of Gamarnik and Sudan, Rahman and Vir\'ag, and Wein on graphs, and answers a question of Bal and Bennett. We conjecture that this statistical--computational gap of a multiplicative 
    of $r^{1/(r-1)}$ indeed holds for this problem.

    Additionally, we explore the universality of this gap by examining $r$-partite hypergraphs. A hypergraph $H = (V, E)$ is $r$-partite if there is a partition $V = V_1\cup \cdots \cup V_r$ such that each edge contains exactly one vertex from each set $V_i$. We consider the problem of finding large \textit{balanced independent sets} (independent sets containing the same number of vertices in each partition) in random $r$-uniform $r$-partite hypergraphs with $n$ vertices from each partition and average degree $d$. We prove that the maximum balanced independent set has density $\left(\frac{r}{r-1}\cdot\frac{\log d}{d}\right)^{1/(r-1)}$ asymptotically, matching that of independent sets in ordinary hypergraphs. Furthermore, we prove an analogous computational threshold of $\left(\frac{1}{r-1}\cdot\frac{\log d}{d}\right)^{1/(r-1)}$ for low-degree polynomial algorithms, answering a question of the first author. We prove more general statements regarding \textit{$\mathbf{\gamma}$-balanced independent sets} (where we specify the proportion of vertices of the independent set contained within each partition). Our results recover and generalize recent work of Perkins and the second author on bipartite graphs.

    Our results not only pin down the precise threshold for low-degree algorithms in two different settings, but also suggest that these gaps persist for larger uniformities as well as across many models. A somewhat surprising aspect of the gap for balanced independent sets is that the algorithm achieving the lower bound is a simple degree-$1$ polynomial.   
\end{abstract}

\sloppy

\newpage
\tableofcontents

\newpage
\input{intro}

\input{prelim}

\input{hypergraph}

\input{balanced}

\section*{Acknowledgments}

We thank Will Perkins and Cheng Mao for helpful comments.
We also thank the anonymous referee for carefully reading this manuscript and providing helpful suggestions.

\printbibliography

\appendix

\input{Local_Tree_Like}

\end{document}

%% file: intro.tex
\section{Introduction}\label{section: intro}

A \textit{hypergraph} $H$ is an ordered pair $(V, E)$ where $E$ is a subset of the power set of $V$.
The elements in $V$ and $E$ are referred to as the \textit{vertices} and \textit{edges} of $H$, respectively.
If every edge $e \in E$ is a set of size $r$, then we say $H$ is \textit{$r$-uniform} (for $r = 2$, $H$ is a graph).
An independent set $I$ of $H$ is a subset of $V$ such that for each $e \in E$, we have $e \not\subseteq I$, i.e., the hypergraph induced by the set $I$ contains no edges.
We say an independent set $I \subseteq V$ has \textit{density} $|I|/|V|$.

\subsubsection*{Independent Sets in Graphs and Hypergraphs}
The problem of finding a maximum size independent set in hypergraphs is of fundamental interest, both in practical and theoretical aspects.
It arises in various applications in data mining~\cite{gao2007minimum}, image processing~\cite{gu2021towards}, database design~\cite{brown1998determining}, and parallel computing~\cite{luby1993removing}, to name a few.
It is well-known that finding the maximum independent set in graphs is an NP-hard problem \cite{karp2010reducibility}.
In fact, it is hard to approximate the maximum independent set in a graph to within $n^{1-\eps}$ for any $\eps > 0$ \cite{hastad1996clique, khot2001improved, zuckerman2006linear}.
As hypergraphs are a generalization of graphs, these hardness results clearly extend to this setting.
In order to better understand the hardness of maximum independent set, a natural line of research is to analyze the structure of the ``hard'' instances.
We consider a somewhat related question: is maximum independent set tractable on ``typical'' instances?
In particular, we consider the $r$-uniform \ER hypergraphs $\Hrnp$, where each potential edge in an $n$-vertex $r$-uniform hypergraph is included independently with probability $p$ (see Definition~\ref{definition: models} for a formal description).

The study of finding large independent sets in random graphs dates back to Matula \cite{matula1976largest} who showed that an \ER graph does not contain an independent set of size at least $(2 + o(1))\log_{1/(1-p)}n$ with high probability for a large range of $p$.
A simple algorithm can construct an independent set of size $\sim\log_{1/(1-p)}n$, however, finding a larger one seems to be hard.
In 1976, Karp conjectured that no polynomial time algorithm can compute an independent set of size at least $(1+ \eps)\log_{2}n$ for any $\eps > 0$ with high probability for $p = 1/2$ \cite{karp1976probabilistic}.
This conjecture has spurred a large research effort in order to determine the so-called \textit{statistical--computational gap} for a number of ``hard'' problems on random graphs (see for example \cite{jerrum1992large, deshpande2015improved, gamarnik2020low, wein2022optimal, bresler2022algorithmic}).

In the \textit{sparse} regime, i.e., $p = d/n$, where $d$ is a constant, Frieze showed that with high probability as $n \to \infty$ the maximum independent set of $\mathcal{H}_2(n, p)$ has density $(2\pm o_d(1))\log d/d$ for sufficiently large constant $d$ \cite{frieze1990independence}.
It is conjectured that no polynomial time algorithm can find an independent set of density $(1+\eps)\log d/d$ with high probability for any $\varepsilon>0$, i.e., there is a statistical--computational gap of a multiplicative factor of $2$.
Proving this conjecture is equivalent to proving a statement stronger than $P = NP$ and so researchers have focused on providing evidence of intractability through a variety of methods.
One such method is to prove intractability for restricted classes of algorithms (see \S\ref{subsection: prior work} for examples). Gamarnik and Sudan~\cite{gamarnik2014limits} first proved a gap of a multiplicative factor of $2\sqrt{2}(1 + \sqrt{2})$ for \textit{local algorithms} and pioneered the framework of the \textit{Overlap Gap Property}, which subsequently led to proving the conjectured gap of a multiplicative factor of $2$ for both local algorithms~\cite{rahman2017local} and low-degree algorithms~\cite{wein2022optimal}.
In our work we focus on \textit{low-degree algorithms} (see \S\ref{subsection: low-deg} for an overview of the framework).

Moving beyond graphs to hypergraphs, we are interested in the \textit{sparse} regime, i.e., $p = d/\binom{n-1}{r-1}$, where $d$ is a constant.
Note that for our choice of $p$, the expected degree of a vertex is precisely $d$. Krivelevich and Sudakov showed that with high probability the maximum independent set of $\Hrnp$ has density $(1\pm o_d(1))\left(\frac{r}{r-1}\cdot\frac{\log d}{d}\right)^{1/(r-1)}$ \cite{krivelevich1998chromatic}, which matches Frieze's result when $r = 2$.
It is natural to ask whether this phenomenon of a statistical--computational gap for finding large independent sets extends to the hypergraph setting.
Additionally, determining the size of this gap is of interest as well.
To the best of our knowledge, this paper is the first to consider the computational threshold for independent sets in random hypergraphs.

\begin{Theorem*}[Informal version of Theorems~\ref{theorem: low-deg hypergraph achievability} and \ref{theorem: low-deg hypergraph impossibility}]
    Let $\eps > 0$ and let $r \geq 2$.
    In the double limit $n \to \infty$ followed by $d \to \infty$, the following hold for $p = d/\binom{n-1}{r-1}$.
    \begin{itemize}
        \item There is a low-degree algorithm that with high probability computes an independent set of density $(1-\eps)\left(\frac{1}{r-1}\cdot\frac{\log d}{d}\right)^{1/(r-1)}$ in $\Hrnp$.
        \item There is no low-degree algorithm that with suitably high probability finds an independent set of density $(1+\eps)\left(\frac{1}{r-1}\cdot\frac{\log d}{d}\right)^{1/(r-1)}$ in $\Hrnp$.
    \end{itemize}
\end{Theorem*}

As local algorithms are a subclass of low-degree algorithms, our results answer a question of Bal and Bennett \cite[\S6.2]{bal2023larger}.
In light of this result and the aforementioned statistical--computational gap conjecture of independent sets in random graphs, we make the following analogous conjecture for hypergraphs.

\begin{Conjecture}\label{conjecture: hypergraph SCG}
    For any fixed $\eps > 0$ and integer $r\geq 2$, there are $d,n \in \N$ sufficiently large such that there is no polynomial-time algorithm that finds an independent set in $\mathcal{H}_r\left(n, d/\binom{n-1}{r-1}\right)$ of density at least $(1+\eps)\left(\frac{1}{r-1}\cdot\frac{\log d}{d}\right)^{1/(r-1)}$ with high probability.
\end{Conjecture}

\subsubsection*{Balanced Independent Sets in Multipartite Hypergraphs}

In recent work, Perkins and the second author explored the universality of the statistical--computational gap of independent sets in graphs \cite{perkins2024hardness}.
They considered the following question: in what other random graph models (and for what kinds of independent sets) do we see such a gap?
Specifically, they studied independent sets of the \ER bipartite graph with $n$ vertices in each partition.
When considering bipartite graphs, there is a trivial independent set of density at least $1/2$, namely, the larger partition.
Furthermore, there exists a max-flow based polynomial time algorithm to find the maximum independent set in bipartite graphs.
As observed by Perkins and the second author, imposing global constraints on the independent set $I$ can introduce computational intractability.

When considering bipartite graphs, a natural global constraint to consider is \textit{balancedness}.
Let $G = (V_1\cup V_2,\, E)$ be a bipartite graph.
We say an independent set $I \subseteq V(G)$ is \textit{balanced} if $|I \cap V_1| = |I|/2$, i.e., $I$ contains an equal number of vertices from each partition.
There has been extensive research in determining the density of the maximum balanced independent set in both deterministic and random bipartite graphs \cite{favaron1993bipartite, axenovich2021bipartite, chakraborti2023extremal}.
Perkins and the second author showed that the \ER bipartite graph with $n$ vertices in each partition has a balanced independent set with density $(2 \pm o_d(1))\log d/d$ with high probability.
They also showed that the low-degree computational threshold of the density of the maximum balanced independent set is $\log d/ d$.
In particular, balanced independent sets in bipartite graphs exhibit the same behavior as independent sets in graphs.
We remark that Perkins and the second author prove a more general statement for so-called \textit{$\boldgamma$-balanced independent sets} (see Definition~\ref{definition: balanced independent sets}).

We extend these results to \textit{$r$-uniform $r$-partite hypergraphs}.
We say $H = (V, E)$ is \textit{$r$-partite} if there exists a partition $V = V_1\cup \cdots \cup V_r$ such that each edge $e \in E$ contains precisely one vertex from each set $V_i$.
Furthermore, if $|V_1| = \cdots = |V_r|$, we call such hypergraphs \textit{balanced}.
Multipartite hypergraphs find a wide array of applications in satisfiability problems, Steiner triple systems, and particle tracking in physics, to name a few.
Note that the union of any $r-1$ partitions forms an independent set in such a hypergraph, and so the problem of finding an independent set of density at least $1 - 1/r$ is trivial.
As this class of hypergraphs is a natural generalization of bipartite graphs, we may extend the definition of balanced independent sets to this setting.
In particular, we say an independent set $I$ in an $r$-uniform $r$-partite hypergraph $H = (V_1\cup \cdots \cup V_r,\, E)$ is \textit{balanced} if $|I \cap V_i| = |I|/r$ for each $1 \leq i \leq r$, i.e., $I$ contains an equal number of vertices from each partition.
Such a notion was introduced in recent work of the first author \cite{dhawan2023balanced}, who considered the size of the largest balanced independent set in deterministic $r$-partite hypergraphs.

We consider the \ER $r$-uniform $r$-partite balanced hypergraph $\Hrrnp$ with $n$ vertices in each partition (note the notational distinction between $\Hrrnp$ and $\Hrnp$).
The second main result of this paper concerns the statistical--computational gap of finding large balanced independent sets in $\Hrrnp$ for $p = d/n^{r-1}$. This parameter choice ensures the expected vertex degree equals precisely d.
We first prove a high probability bound on the asymptotic density of the largest balanced independent set in $\Hrrnp$.
Moreover, we describe a simple, efficient algorithm that can find a balanced independent set within a multiplicative factor of $r^{1/(r-1)}$ of the statistical threshold.
This begs the following question: can one do better or is there a universality of the statistical–computational gap for finding large independent sets in hypergraphs as well?
We provide evidence toward the latter.

\begin{Theorem*}[Informal version of Theorems~\ref{theorem: stat thresh balanced} and \ref{theorem: low-deg thresh balanced}]
    Let $\eps > 0$ and let $r \geq 2$.
    In the double limit $n \to \infty$ followed by $d \to \infty$, the following hold for $p = d/n^{r-1}$.
    \begin{itemize}
        \item The largest balanced independent set in $\Hrrnp$ has density $(1\pm\eps)\left(\frac{r}{r-1}\cdot\frac{\log d}{d}\right)^{1/(r-1)}$ with high probability.
        \item There is a low-degree algorithm that with high probability finds a balanced independent set of density $(1-\eps)\left(\frac{1}{r-1}\cdot\frac{\log d}{d}\right)^{1/(r-1)}$ in $\Hrrnp$.
        \item There is no low-degree algorithm that with suitably high probability finds a balanced independent set of density $(1+\eps)\left(\frac{1}{r-1}\cdot\frac{\log d}{d}\right)^{1/(r-1)}$ in $\Hrrnp$.
    \end{itemize}
\end{Theorem*}

Our results answer a question of the first author, and for $r = 2$, we recover the results of Perkins and the second author.
We also conjecture that this statistical--computational gap persists for polynomial-time algorithms (a version of this conjecture for $r = 2$ appeared in \cite{perkins2024hardness}).

\begin{Conjecture}\label{conjecture: balanced hypergraph SCG}
    For any $\eps > 0$ and integer $r\geq 2$, in the double limit $n \to \infty$ followed by $d \to \infty$ there is no polynomial-time algorithm that finds a balanced independent set in $\mathcal{H}\left(r, n, d/n^{r-1}\right)$ of density at least $(1+\eps)\left(\frac{1}{r-1}\cdot \frac{\log d}{d}\right)^{1/(r-1)}$ with high probability.
\end{Conjecture}

We do in fact prove more general statements for \textit{$\boldgamma$-balanced independent sets}.
Here, rather than considering an independent set $I$ with an equal proportion of vertices in each partition, we specify the proportion of vertices in $I$ within each partition in a vector $\boldgamma \in (0, 1)^r$ (see Definition~\ref{definition: balanced independent sets} for a formal definition).
For $\boldgamma = (1/r, \ldots, 1/r)$, we recover the original notion of balanced independent sets.

We remark that while multipartite hypergraphs are a natural extension of bipartite graphs to larger uniformities, certain properties do not extend.
For example, bipartite graphs are triangle-free, while $r$-partite hypergraphs can contain triangles for $r \geq 3$ (the same holds for odd cycles).
This makes our results all the more interesting as we recover the results for $r = 2$.

\subsection{Relation to Prior Work}\label{subsection: prior work}

\paragraph{Statistical-computational gaps.}
Inference problems with conjectured statistical--computational gaps are ubiquitous throughout statistics and computer science.
Some classical examples include the planted clique problem \cite{jerrum1992large, deshpande2015improved, meka2015sum, barak2019nearly} and community detection in random graphs \cite{decelle2011asymptotic, arias2014community, hopkins2017bayesian}, structured principal component analysis of matrices \cite{berthet2013computational, lesieur2015phase} and tensors \cite{hopkins2015tensor, hopkins2017power}, and solving or refuting constraint satisfaction problems \cite{achlioptas2008algorithmic, kothari2017sum}.

While proving statistical--computational gaps of the above problems is considered a daunting task, it has spurred a large research effort toward providing rigorous evidence of hardness for average-case instances.
As mentioned earlier, this involves proving hardness for certain restricted classes of algorithms.
These techniques include Markov chain Monte Carlo methods \cite{jerrum1992large, dyer2002counting}, local algorithms \cite{gamarnik2014limits, rahman2017local, chen2019suboptimality}, belief propagation and approximate message passing algorithms \cite{decelle2011asymptotic, lesieur2015phase}, reductions from the presumed hard planted clique problem \cite{hajek2015computational, brennan2019optimal}, and statistical query models \cite{kearns1998efficient, feldman2015complexity}, to name a few.

Within the random graph framework, there are three primary problems of interest: optimization, testing, and estimation.
This work focuses on the optimization problem, i.e., finding a large (balanced) independent set in a random hypergraph.
One can analogously consider the testing and estimation problems by analyzing so-called ``planted'' models.
The above techniques have been applied to a wide variety of such problems including the densest $k$-subgraph problem \cite{feige1997densest, bhaskara2010detecting}, planted dense subgraph problem \cite{jones2023sum, PDC, PDC_IT}, and hypergraph versions \cite{chlamtac2018densest, corinzia2022statistical, dhawan2023detection}.

\paragraph{Low-degree algorithms.}
The study of low-degree algorithms has become a popular approach in recent years. 
The main focus of this work is to provide sharp thresholds for the tractability of \textit{low-degree polynomial algorithms} (to be formally defined in the next section).
Here, each vertex's membership in the (balanced) independent set is determined by a multivariate polynomial in the edge indicator variables of the hypergraph.
This is a powerful class of algorithms as it includes the class of local algorithms as well as the algorithmic paradigms of approximate message passing and power iteration (see the discussion in \cite[Appendix A]{gamarnik2024hardness}).
Given the current state of average-case complexity theory, it has been shown that low-degree algorithms are as powerful as the best known polynomial-time algorithms for a number of problems in high-dimensional statistics including those mentioned in the preceding paragraphs. Therefore, the failure of low-degree algorithms is a sign of concrete evidence for computational hardness of statistical problems. 

There is by now a standard method for proving low-degree optimization bounds based on the \textit{Overlap Gap Property}, which we describe in \S\ref{subsection: proof overview} (see the survey \cite{gamarnik2021overlap} for a more extensive overview).
This technique has been employed in a number of inference problems including independent sets in random graphs \cite{gamarnik2020low, wein2022optimal}, random $k$-SAT \cite{bresler2022algorithmic}, discrepancy of random matrices \cite{venkat2022efficient}, and random CSPs \cite{chen2023local}.
Our work is the among the first to apply this strategy to random hypergraphs.

\paragraph{Statistical inference on random hypergraphs.}
Random graph inference has been a central research area for decades, while the hypergraph extension is less well-studied.
In general, the hypergraph setting is believed to be considerably harder.
In recent years, however, there has been an increase in both theoretical and applied interest in this setting.
These include applying spectral methods to solve testing and estimation problems \cite{luo2022tensor, jones2023sum} and determining statistical and computational thresholds for planted problems \cite{yuan2021heterogeneous, yuan2021information, dhawan2023detection}.

While the maximum independent set problem is NP-hard, there has been progress for restricted classes of algorithms on structurally constrained hypergraphs.
In particular, \cite{halldorsson2009independent, guruswami2011complexity} consider bounded-degree hypergraphs, \cite{halldorsson2016streaming} consider streaming algorithms for sparse hypergraphs, \cite{khanna2021independent} consider semi-random hypergraphs, and \cite{halperin2002improved, agnarsson2013sdp} explore SDP based methods of solving maximum independent set.
To the best of our knowledge, this is the first paper to consider the statistical--computational gap of finding independent sets in \ER hypergraphs.

In contrast, multipartite hypergraphs have been less extensively studied from a computational standpoint, however, there are a number of theoretical results which may be of interest to the reader \cite{kamvcev2017bounded, dhawan2023list, bowtell2024matchings}.
In \cite{guruswami2015inapproximability}, the authors consider the minimum vertex cover problem on such hypergraphs.
A somewhat surprising result is that while the problem is tractable for $r = 2$, it is NP-hard for $r \geq 3$.
This makes our results all the more interesting as we show that the behavior of maximum balanced independent set does not exhibit such a distinction based on $r$.
A similar observation was made in \cite{botelho2012cores} who consider the appearance of a $k$-core (a subhypergraph of minimum degree $k$) in certain random $r$-partite hypergraph models.

\subsection{The Low-Degree Framework}\label{subsection: low-deg}

We now define the framework of \textit{low-degree polynomial algorithms}.
We say a function $f\,:\,\mathbb{R}^m \to \mathbb{R}^{n}$ is a \textit{polynomial of degree (at most) D} if it can be written as $f(A) = (f_{1}(A), \ldots, f_{n}(A))$, where each $f_i\,:\,\mathbb{R}^m \to \mathbb{R}$ is a multivariate polynomial of degree at most $D$. 
For a probability space $(\Omega, P_{\omega})$, we say $f\,:\,\mathbb{R}^m \times \Omega \to \mathbb{R}^{n}$ is a \textit{random polynomial} if $f(\cdot, \omega)$ is a degree $D$ polynomial for each $\omega \in \Omega$, i.e., the coefficients of $f$ are random but do not depend on $A$. 
For our purposes, the input of $f$ is an indicator vector $A \in \{0,1\}^{m}$ encoding the edges of an $n$-vertex $r$-uniform hypergraph with $m = \binom{n}{r}$ or an $(rn)$-vertex $r$-uniform balanced $r$-partite hypergraph with $m = n^r$.

We must formally define what it means for such a polynomial to find an independent set in a hypergraph.
We follow the notation and definitions from \cite{wein2022optimal}, in which a rounding procedure is used to produce an independent set of $\mathcal{H}_2(n, p)$ based on the output $f(A)$. 
For brevity, we just state the following definition in this section for ordinary hypergraphs (the multipartite analog can be inferred as it is nearly identical).

\begin{Definition} \label{Def:V_eta_f}
    Let $f\,:\,\mathbb{R}^m \to \mathbb{R}^{n}$ be a random polynomial, with $m = {n \choose r}$. For $A \in \{0,1\}^m$ indicating the edges of a hypergraph on $n$ vertices and $\eta >0$, let $V^{\eta}_f(A,\omega)$ be the independent set constructed as follows. 
    Let 
    \begin{align*}
        I &=\big\{ i \in [n]\,:\, f_i(A,\omega) \geq 1\big\} \\
        \Tilde{I} &= \big\{i \in I\,:\, \forall e\ni i,\, e \not\subseteq I\big\} \,  ,  \, \, \text{ and} \\
         J &= \left\{ i \in [n]\,:\, f_i(A,\omega) \in \left(\frac{1}{2},\,1\right)\right\} \,.
    \end{align*}
    Then define
    \begin{align*}
        V^{\eta}_f(A,\omega) =
        \begin{cases}
            \Tilde{I} & \text{if } |I \setminus  \Tilde{I}|+|J| \leq \eta n; \\
            \varnothing & \text{otherwise}.
        \end{cases}
    \end{align*}
\end{Definition}

In essence, a vertex $i$ is in the independent set if the output of the corresponding polynomial $f_i$ is at least $1$, and it is not in the independent set if the output of $f_i$ is at most $1/2$.
We allow up to $\eta\,n$ ``errors'', i.e., vertices $i$ which satisfy one of the following:
\begin{itemize}
    \item either $f_i(A) \in (1/2,\, 1)$, or
    \item there is some edge $e$ containing $i$ such that $e \subseteq I$ ($i$ violates the independence constraint).
\end{itemize}
The error tolerance of $(1/2, 1)$ is important for our impossibility results, as this ensures that a small change in $f(A, \omega)$ cannot cause a large change in the resulting independent set without encountering the failure event $\varnothing$. 
With the mapping $V_f^\eta$ in hand, let us formally define how a random polynomial finds an independent set of a certain size.

\begin{Definition} \label{Def:f_Optimize_ind}
    For parameters $k > 0$, $\delta \in [0,1]$, and $\xi \geq 1$, a random function $f \,:\, \mathbb{R}^m \to \mathbb{R}^{n}$ is said to $(k, \delta, \xi, \eta)$-optimize the independent set problem in $\Hrnp$ if the following are satisfied when $A \sim \Hrnp$:
    \begin{itemize}
        \item $\mathbb{E}_{A,\omega}\left[\|f(A,\omega)\|^2\right] \leq \xi k$ and
        \item $\mathbb{P}_{A,\omega}\left[|V^\eta_f(A,\omega)| \geq k \right] \geq 1-\delta$.
    \end{itemize}
\end{Definition}

Here, $k$ is the desired size of the independent set, $\delta$ is the failure probability of the algorithm, $\xi$ is a normalization parameter, and $\eta$ is the error tolerance parameter of the rounding procedure.
Similarly, we define how a random polynomial finds a balanced independent set of a certain size.

\begin{Definition} \label{Def:f_Optimize_ind_bal}
    For parameters $k_{1}, \ldots, k_r > 0$, $\delta \in [0,1]$, and $\xi \geq 1$, a random function $f \,:\, \mathbb{R}^m \to \mathbb{R}^{rn}$ is said to $(k_{1}, \ldots, k_{r}, \delta, \xi, \eta)$-optimize the balanced independent set problem in $\Hrrnp$ if the following are satisfied when $A \sim \Hrrnp$:
    \begin{itemize}
        \item $\mathbb{E}_{A,\omega}\left[\|f(A,\omega)\|^2\right] \leq \xi(k_{1}+ \cdots + k_r)$ and
        \item $\mathbb{P}_{A,\omega}\left[\forall i \in [r],\,|V^\eta_f(A,\omega) \cap V_i| \geq k_i\right] \geq 1-\delta$,
    \end{itemize}
    where $V_1, \ldots, V_r$ are the vertex partitions of the $r$-partite hypergraph $\Hrrnp$.
\end{Definition}

Here, $k_1, \ldots, k_r$ are the desired sizes of the intersections of the independent set with each partition, while the other parameters are the same as in Definition~\ref{Def:f_Optimize_ind}.
We note that for balanced independent sets, we would have $k_1 = \cdots = k_r$, however, as we will see in the next section, we consider the more general \textit{$\gamma$-balanced independent sets} in which we allow different values of $k_i$.

\subsection{Main Results}\label{subsection: results}

We are now ready to state our main results.
First, we define the random models we consider.

\begin{Definition}[The Random Hypergraph Models]\label{definition: models}
    Let $n,\, r \in \N$ such that $n \geq r \geq 2$.
    \begin{itemize}
        \item We construct the hypergraph $H \sim \Hrnp$ on vertex set $[n]$ by including each $e \subseteq [n]$ of size $r$ in $E(H)$ independently with probability $p$.
        \item We construct the hypergraph $H \sim \Hrrnp$ on vertex set $[n]\times [r]$ by including each
        \[e \in V_1 \times \cdots \times V_r,\]
        in $E(H)$ independently with probability $p$.
        Here, $V_i = [n] \times \set{i}$.
    \end{itemize}
\end{Definition}

Our first results are on the low-degree computational threshold for constructing large independent sets in \ER hypergraphs.

\begin{Theorem}[Achievability result for $\Hrnp$]\label{theorem: low-deg hypergraph achievability}
    Let $\eps > 0$ and $r \in \N$ such that $r\geq 2$.
    There exists $d^* > 0$ such that for any $d \geq d^*$ and $\eta > 0$, there exist $n^*,\,C,\,D > 0$ and $\xi \geq 1$ such that the following holds for all $n \geq n^*$. For 
    \[k = n\,(1 - \eps)\left(\frac{1}{r-1}\cdot \frac{\log d}{d}\right)^{1/(r-1)}, \quad \text{and} \quad \delta = \exp\left(-Cn^{1/3}\right),\]
    there exists a degree-$D$ polynomial that $(k, \delta, \xi, \eta)$-optimizes the independent set problem in $\Hrnp$ for $p = d/\binom{n-1}{r-1}$.
\end{Theorem}

\begin{Theorem}[Impossibility result for $\Hrnp$]\label{theorem: low-deg hypergraph impossibility}
    Let $\eps > 0$ and $r \in \N$ such that $r\geq 2$.
    There exists $d^{*}>0$ such that for any $d \geq d^{*}$ there exist $\eta,\, C_1,\, C_2,\, n^* > 0$ such that the following holds for any $n \geq n^{*}$, $\xi \geq 1$, $1 \leq D \leq \frac{C_1 n}{\xi \log n}$, and $\delta \leq \exp(-C_2 \xi D \log n)$.
    If $k \geq n\,(1+\varepsilon) \left(\frac{1}{r-1}\cdot\frac{\log d}{d}\right)^{\frac{1}{r-1}}$, there is no degree-$D$ polynomial that $(k, \delta, \xi, \eta\big)$-optimizes the independent set problem in $\Hrnp$.
\end{Theorem}

Before we state our results for multipartite hypergraphs, we define $\boldgamma$-balanced independent sets.

\begin{Definition}[Balanced independent sets]\label{definition: balanced independent sets}
    Let $H = (V_1\cup \cdots \cup V_r,\, E)$ be an $r$-uniform $r$-partite hypergraph for $r \geq 2$, and let $\boldgamma = (\gamma_1, \ldots, \gamma_r)$ be such that $\gamma_i \in (0, 1)\cap \mathbb{Q}$ and $\sum_{i = 1}^r\gamma_i = 1$.
    An independent set $I \subseteq V(H)$ is \textit{$\boldgamma$-balanced} if $|I\cap V_i| = \gamma_i|I|$ for each $i \in [r]$.
    We let $\alpha_{\boldgamma}(H)$ denote the size of the largest $\boldgamma$-balanced independent set in $H$.
\end{Definition}

Our first result in this setting establishes the statistical threshold for $\alpha_{\boldgamma}(H)$ when $H \sim \Hrrnp$ (see \S\ref{subsection: IT bound} for the proof).

\begin{Theorem}[Statistical threshold for $\boldgamma$-balanced independent sets]\label{theorem: stat thresh balanced}
    Let $\eps > 0$ and $r \in \N$ such that $r\geq 2$, and let $\boldgamma = (\gamma_1, \ldots, \gamma_r)$ be such that $\gamma_i \in (0, 1) \cap \mathbb{Q}$ and $\sum_{i = 1}^r\gamma_i = 1$.
    There exists $d^* > 0$ such that for any $d \geq d^*$, there exists $n^* > 0$ such that for any $n \geq n^*$ and $p = d/n^{r-1}$, the hypergraph $H \sim \Hrrnp$ satisfies
    \[(1 - \eps)\left(\frac{\log d}{d(r-1)\prod_i\gamma_i}\right)^{\frac{1}{r-1}} \,\leq\, \frac{\alpha_{\boldgamma}(H)}{n} \,\leq\,  (1 + \eps)\left(\frac{\log d}{d(r-1)\prod_i\gamma_i}\right)^{\frac{1}{r-1}},\]
    with probability $1 - \exp\left(-\Omega(n)\right)$.
\end{Theorem}

Note that setting $\gamma_i = 1/r$ for each $1 \leq i \leq r$ implies the maximum balanced independent set has density $(1\pm\eps)\left(\frac{r}{r-1}\cdot \frac{\log d}{d}\right)^{1/(r-1)}$ with high probability (as we claimed in \S\ref{section: intro}).
Finally, we prove the low-degree computational threshold for finding large $\boldgamma$-balanced independent sets.

\begin{Theorem}[Low-degree threshold for $\boldgamma$-balanced independent sets]\label{theorem: low-deg thresh balanced}
    Let $\eps > 0$ and $r \in \N$ such that $r\geq 2$, and let $\boldgamma = (\gamma_1, \ldots, \gamma_r)$ be such that $\gamma_i \in (0, 1) \cap \mathbb{Q}$ and $\sum_{i = 1}^r\gamma_i = 1$.
    Additionally, let $i^\star = \arg\max_i \gamma_i$.
    There exists $d^* > 0$ such that for any $d \geq d^*$, the following hold.
    \begin{itemize}
        \item There exist $n^* > 0$ and $\xi \geq 1$ such that for any $n \geq n^*$
        and
        \[k_j \leq \gamma_j\,n(1 - \eps)\left(\frac{\gamma_{i^\star}\log d}{d(r - 1)\prod_{i = 1}^r\gamma_i}\right)^{1/(r-1)} \quad \text{for } \quad 1 \leq j \leq r,\]
        there is a degree-$1$ polynomial that $(k_1, \ldots, k_r, \exp\left(-\Omega(n)\right), \xi, 0)$-optimizes the balanced independent set problem in $\Hrrnp$ for $p = d/n^{r-1}$.

        \item There exist $\eta,\, C_1,\, C_2,\, n^* > 0$ such that the following holds for any $n \geq n^*$, $\xi \geq 1$, $1 \leq D \leq \frac{C_1n}{\xi\log n}$, and $\delta \leq \exp\left(-C_2\xi\,D\log n\right)$.
        If
        \[k_j \,\geq\, \gamma_j\,n(1+\eps)\left(\frac{\gamma_{i^\star}\log d}{d(r - 1)\prod_{i = 1}^r\gamma_i}\right)^{1/(r-1)} \quad \text{for } \quad 1 \leq j \leq r,\]
        then there is no random degree-$D$ polynomial that $(k_1, \ldots, k_r, \delta, \xi, \eta)$-optimizes the balanced independent set problem in $\Hrrnp$ for $p = d/n^{r-1}$.

    \end{itemize}
\end{Theorem}

Setting $\gamma_i = 1/r$ for each $i$ provides evidence toward Conjecture~\ref{conjecture: balanced hypergraph SCG}.
For arbitrary $\boldgamma$, we observe a statistical--computational gap of a multiplicative factor of $\gamma_{i^\star}^{-1/(r-1)}$ for low-degree polynomial algorithms, matching the result of \cite{perkins2024hardness} for $r = 2$.

\subsection{Proof Overview}\label{subsection: proof overview}

We will now provide an overview of our proof techniques.
We will consider each model separately.

\paragraph{Independent Sets in $\Hrnp$.}
Let us first discuss the achievability result (Theorem~\ref{theorem: low-deg hypergraph achievability}).
It is a well-known fact that sparse random graphs are locally tree-like.
In \cite{wein2022optimal}, Wein showed that the $s$-neighborhood of a vertex in $\mathcal{H}_2(n, d/n)$ can be well approximated by the \textit{Poisson Galton-Watson tree}, a random graph model for generating trees.
In \cite{pal2021community}, Pal and Zhu generalized this model to the hypergraph setting, describing the $r$-uniform \textit{Poisson Galton-Watson Hypertree $\mathbb{T}_d^{GW}$} (see \S\ref{subsection: local} for a formal description).
We observe that for $p = d/\binom{n-1}{r-1}$, the hypergraph $\Hrnp$ is locally $\mathbb{T}_d^{GW}$-like, which implies that the performance of a local algorithm on $\Hrnp$ is determined, up to the first order, by the expectation of the corresponding algorithm evaluated at the root of $\mathbb{T}_d^{GW}$.
Our proof now follows in two steps: first, we describe a local algorithm for finding large independent sets in $\mathbb{T}_d^{GW}$ (which can be adapted to one on $\Hrnp$ by the aforementioned observation); next, we show how a local algorithm on $\Hrnp$ can be well-approximated by a low-degree one, completing the proof.
For the first step, we adapt a local algorithm of Nie and Verstra{\"e}te on the \textit{$\Delta$-regular $r$-uniform hypertree $T_{\Delta}^r$} \cite{nie2021randomized} to one on $\mathbb{T}_d^{GW}$ for an appropriate $\Delta \coloneqq \Delta(d) > 0$.
The second step follows a similar approach to that of Wein on graphs.

The proof of our impossibility result (Theorem~\ref{theorem: low-deg hypergraph impossibility}) falls into a line of work initiated by Gamarnik and Sudan \cite{gamarnik2014limits}, who studied local algorithms for independent sets of $\mathcal{H}_2(n, p)$.
The proof of their impossibility result relies on the so-called \textit{Overlap Gap Property (OGP)}.
In particular, they show the following: if $I_1$ and $I_2$ are ``large'' independent sets, then either $I_1\cap I_2$ is a ``large'' independent set or it is ``small'', i.e., either $I_1$ and $I_2$ contain a lot of common vertices or are nearly disjoint.
They then show that if a local algorithm succeeds in finding a large independent set, it can be used to construct two independent sets violating the OGP condition.
Rahman and Vir\'ag \cite{rahman2017local} improved upon their results by considering a more intricate ``forbidden'' structure involving many independent sets as opposed to just two.
Their approach inspired further works in a number of different areas (see for example \cite{gamarnik2014performance, coja2017walksat, chen2019suboptimality}).
The hardness of random optimization problems within the low-degree polynomial framework was first studied in ~\cite{gamarnik2020low}.  Wein applied this ``ensemble'' variant of the OGP, first introduced in~\cite{chen2019suboptimality}, to establish his impossibility result for independent sets in $\mathcal{H}_2(n, p)$ \cite{wein2022optimal}.
In our proof, we employ the \textit{ensemble-OGP} in a similar fashion by considering a forbidden structure that involves many independent sets across many correlated random hypergraphs.
We show that with high probability this structure does not exist, and any low-degree algorithm that can construct a large independent set can also be used to construct an instance of this structure, leading to a contradiction.

\paragraph{Balanced Independent sets in $\Hrrnp$.}
Let us first discuss the statistical threshold (Theorem~\ref{theorem: stat thresh balanced}).
The proof of the upper bound follows by a standard application of the first moment method.
For the lower bound, our approach is inspired by that of Frieze in \cite{frieze1990independence} where he proves a high probability bound on the asymptotic size of the maximum independent set in $\mathcal{H}_2(n, p)$ (Krivelevich and Sudakov employed a similar approach in \cite{krivelevich1998chromatic} for $\Hrnp$).
The strategy follows a second moment argument which shows that the maximum independent set with an additional property has the desired size with positive probability.
The additional property ensures this random variable concentrates well, implying the lower bound with high probability through an application of Azuma's inequality.
A similar approach was employed by Perkins and the second author in \cite{perkins2024hardness} where they consider the case when $r = 2$.
For $r \geq 3$, however, certain combinatorial arguments no longer hold, requiring a more extensive analysis.
The details are provided in \S\ref{subsection: IT bound}.

The low-degree algorithm we describe for our achievability result in Theorem~\ref{theorem: low-deg thresh balanced} is inspired by recent work of the first author who considered balanced independent sets of deterministic $r$-partite hypergraphs \cite{dhawan2023balanced}.
While the primary interest of that work is on the existence of large balanced independent sets, the proof is constructive and yields a simple degree-$1$ algorithm.
The goal is to adapt this algorithm to construct $\boldgamma$-balanced independent sets in $\Hrrnp$.
We remark that for $r = 2$, a similar approach was used by Perkins and the second author \cite{perkins2024hardness}, where they describe an algorithm inspired by work of Chakraborty \cite{chakraborti2023extremal}.

The proof of our impossibility result follows a similar approach to that of Theorem~\ref{theorem: low-deg hypergraph impossibility}.
We once again employ the ensemble-OGP and are able to arrive at the desired contradiction.
There is an additional layer of complexity due to the structural constraints on the hypergraph and on the independent sets considered (being $r$-partite and $\boldgamma$-balanced, respectively).
Therefore, while the overall strategy of the proof is identical to that of Theorem~\ref{theorem: low-deg hypergraph impossibility}, the analysis differs greatly.
We remark that for $r = 2$, Perkins and the second author employed a similar approach.
They were able to reduce their argument to the case when $\gamma_1 = \gamma_2 = 1/2$, which substantially simplified their arguments.
Unfortunately, such a reduction is not always possible for larger uniformities (one can verify that this reduction is possible only when at least $r - 1$ of the $\gamma_i$'s satisfy $\gamma_i \geq 1/r$).
As a result, our analysis is much more involved.


\subsection{Concluding Remarks}\label{subsection: future}

In this work, we consider the computational hardness of finding large (balanced) independent sets in random hypergraphs.
We focus on low-degree polynomials, a powerful class of algorithms considered to be a useful proxy for computationally efficient algorithms.
For the \ER hypergraph $\Hrnp$, we determine the threshold for success of such algorithms in constructing independent sets, recovering and generalizing results of Wein \cite{wein2022optimal}.
For the \ER $r$-partite hypergraph $\Hrrnp$, we prove a matching threshold for balanced independent sets, generalizing results of Perkins and the second author \cite{perkins2024hardness}.
In both problems, our results indicate a statistical--computational gap of a multiplicative factor of $r^{1/(r-1)}$, that is, the low-degree threshold depends on $r$ for $r$-uniform hypergraphs. This is not entirely expected and cannot be easily deduced from the graph case.

A natural extension of our results would be to consider denser hypergraphs, i.e., larger values of the probability parameter $p$.
For $\mathcal{H}_2(n, p)$ where $p = \Theta(1)$, Matula and Karp's results \cite{matula1976largest, karp1976probabilistic, wein2022optimal} indicate a statistical--computational gap of a multiplicative factor of $2$ for independent sets.
It is not clear how to describe Karp's algorithm as a low-degree polynomial and so the question of the low-degree threshold in this regime still remains open.
Furthermore, to the best of our knowledge, hypergraph versions of these results have not yet appeared in the literature, which begs the following question: is there a statistical--computational gap of a multiplicative factor of $r^{1/(r-1)}$ for independent sets in $\Hrnp$ when $p = \Theta(1)$?

By now, OGP based methods have become standard in proving low-degree thresholds for optimization problems on graphs (see for example \cite{wein2022optimal, huang2022tight, du2023algorithmic, gamarnik2024hardness, perkins2024hardness}).
Our work is the among the first to apply the low-degree framework and the OGP based proof strategy to optimization problems on random hypergraphs.
A potential future line of inquiry would be to adapt this strategy to other optimization problems in this more general setting.

We conclude this section with an observation regarding Theorem~\ref{theorem: low-deg thresh balanced}.
What is most surprising about our result is that the algorithm achieving the low-degree threshold is a simple degree-$1$ polynomial.
While Conjecture~\ref{conjecture: hypergraph SCG} poses a long-standing challenge, the simplicity of our algorithm seems to indicate that finding large balanced independent sets may be easier (which would imply Conjecture~\ref{conjecture: balanced hypergraph SCG} is false).

%% file: prelim.tex
\section{Preliminaries}\label{section: prelim}

\subsection{Notation}\label{subsection: notation}

Let $\N$ denote the set of nonnegative integers.
For $n \in \N$ such that $n \geq 1$, we let $[n] \coloneqq \set{1, \ldots, n}$.
Throughout this work, we use the asymptotic notation $O_x(\cdot)$, $\Omega_x(\cdot)$, $o_x(\cdot)$, etc. as $x \to \infty$.
We ignore the subscript when $x$ is clear from context.
Note that our results are in the double limit $n \to \infty$ followed by $d \to \infty$.
In particular, $o_n(1)$ is a quantity depending arbitrarily on $d$ tending to $0$ as $n \to \infty$, while $o_d(1)$ is a quantity independent of $n$ which tends to $0$ as $d \to \infty$.

For any fixed $r \geq 2$, we let $K_n^r$ denote the $n$-vertex complete $r$-uniform hypergraph and $K_{r\times n}$ denote the $(rn)$-vertex complete balanced $r$-uniform $r$-partite hypergraph.
For a hypergraph $H$, we let $V(H)$ denote its vertex set and $E(H)$ denote its edge set. 
The degree of a vertex $v \in V(H)$ (denoted $\deg_H(v)$) is the number of edges containing $v$.  
We say $H$ is $d$-regular if all vertices have degree $d$. 
If $H$ is $r$-partite, we let $V_1(H), \ldots, V_r(H)$ denote each vertex partition. 
We say $H$ is \textit{linear} if for any $e,f \in E(H)$, $|e\cap f|\leq 1$. 

A \textit{path} in a hypergraph $H$ of length $k$ is a sequence $(v_1,e_1,v_2, e_2, \ldots, v_k, e_k, v_{k+1})$ of distinct vertices and edges such that $v_i, v_{i+1} \in e_i$ for each $i \in [k]$.
A \textit{cycle} in a hypergraph $H$ of length $k$ is a path of length $k$ with $v_{k+1} = v_1$. 
The \textit{girth} of a hypergraph $H$ is the length of the shortest cycle in $H$.
We say a hypergraph $H$ is \textit{connected} if every pair of vertices $u,\,v \in V(H)$ can be connected by a path, and \textit{acyclic} if it contains no cycles.
A \textit{hyperforest} is an acyclic graph and a \textit{hypertree} is a connected hyperforest (as analogously defined for graphs).
It is not too difficult to see that hyperforests are linear as well.

For a distribution $\mathcal{P}$, we let $\E_{\mathcal{P}}[\cdot]$ and $\Var_{\mathcal{P}}(\cdot)$ denote the expectation and variance, respectively.
We drop the subscript when $\mathcal{P}$ is clear.
We let $\mathsf{Unif}[0,1]$ denote the uniform distribution over the closed interval $[0,1]$, and  $\mathsf{Unif}[0,1]^n$ denote the $n$-dimensional vector whose entries are i.i.d $\mathsf{Unif}[0,1]$.
Let $\mathsf{Ber}(p)$ denote the Bernoulli distribution with parameter $p \in [0, 1]$, and $\mathsf{Pois}(d)$ denote the Poisson distribution with parameter $d>0$.

\subsection{Local Algorithms on Hypergraphs}\label{subsection: local}

In this section, we define some terminology pertaining to local algorithms on hypergraphs, which will be important in the proof of Theorem~\ref{theorem: low-deg hypergraph achievability}.
We will consider locally finite hypergraphs $H$, i.e. each vertex has a finite number of edges containing it. 

For any integer $s \geq 0$ and $v\in V(H)$, let $N_{s}(H,v)$ denote the the rooted hypergraph with root $v$ and its depth-$s$ neighborhood in $H$, and let $|N_{s}(H,v)|$ denote the number of hyperedges in $N_{s}(H,v)$. 
An \textit{$s$-local} algorithm for finding an independent set in a hypergraph $H$ is defined by a measurable $s$-local function $g = g(H,v, \mathbf{X})$.
Here, the input is a rooted hypergraph $H$ with root $v$ and depth at most $s$, and a vector $\mathbf{X} \in [0,1]^{|V(H)|}$.
The output of the function is $0$ (`out') or $1$ (`in') and is invariant under labeled isomorphisms of $N_{s}(H,v)$. 
We apply an $s$-local function $g$ to a hypergraph $H$ by assigning each vertex $v$ a label $\mathbf{X}_v \sim \textsf{Unif}[0,1]$ independently, and then evaluating $g(\cdot)$ on $N_{s}(H,v)$ and the restriction of the labels $\mathbf{X}$ to $V\left(N_{s}(H,v)\right)$.
We say an $s$-local function $g$ is valid if for any hypergraph $H$, and any label $\mathbf{X} \in [0,1]^{|V(H)|}$, the set $\set{v\,:\,g(H,v,\mathbf{X}) = 1}$ is an independent set in $H$.
A local algorithm is a valid $s$-local function for some constant $s$. 

We measure the performance of local algorithms for independent sets in random hypergraphs by the typical size or density of an independent set returned by the algorithm, with high probability over both the random graph and the random labels $\mathbf{X}$.

As mentioned in \S\ref{subsection: proof overview}, we will consider two kinds of hypertrees in our proof of Theorem~\ref{theorem: low-deg hypergraph achievability}.
First, we describe the $r$-uniform Poisson$(d)$ Galton-Watson Hypertree (denoted $\PGWLT$) introduced by Pal and Zhu in \cite{pal2021community}.
The process generates a rooted $r$-uniform hypertree $(T,o)$ as follows:
\begin{itemize}
    \item start with a root vertex $o$ at level $0$.
    \item For $k = 0,1,2, \ldots,$ each vertex at level $k$ independently spawns $\mathsf{Pois}(d)$ offspring edges that pairwise intersect only at the parent vertex.
    Note that for each vertex at level $k$, we have $(r-1)\mathsf{Pois}(d)$ children vertices at level $k+1$.
\end{itemize}

It is a well-known fact that sparse random graphs are locally tree-like.
The same holds for hypegraphs\footnote{To the best of our knowledge, no proof of this fact exists in the literature, however, the proof is nearly identical to the graph case. For completeness, we include such a proof in \S\ref{section: tree like}.} and so the performance of a local algorithm on $\Hrnp$ is determined, up to the first order, by the expectation of the corresponding local function $g$ evaluated at the root of $\PGWLT$. 
More precisely, we have the following lemma:

\begin{lemma}\label{Lemma:local_tree_np}
    Suppose $g$ is a valid $s$-local function for independent sets. Let $I$ be the random independent set obtained by applying $g$ to $H \sim \Hrnp$. Then for $T \sim \PGWLT$  with high probability over the randomness in $H$ and the random label $\mathbf{X} \sim \mathsf{Unif}[0,1]^n$,
    \[
        \frac{|I|}{n} = \E[g(T,o,\mathbf{X})]+o_n(1).
    \]
\end{lemma}

To assist with the proof of the above result, we will employ the following proposition.

\begin{Proposition}\label{prop: local alg moment inequality}
    Fix some integer $s$ and $r \geq 2$. 
    Let $H \sim \Hrnp$ for $p = d/{n-1 \choose r-1}$ where $d \geq 1$ and let $\mathbf{X} \sim \mathsf{Unif}[0,1]^n$.
    Let $g$ be an $s$-local function taking values in $\set{0, 1}$.
    We have
    \[\E\left[\bigg|\sum_{v \in [n]}g(H, v,\mathbf{X}) - \E\sum_{v\in [n]}g(H, v,\mathbf{X})\bigg|^2\right] = O(n)\]
\end{Proposition}

\begin{proof}
    Fix some order on the vertices of $H$.
    Let $E_x$ be set $\set{e\in K_n^{r}\,:\, x = \min e}$.
    Note that the random variables $(E_x)_{x\in [n]}$ are independent.
    Furthermore, there is a function $F$ such that
    \[\sum_{v \in [n]}g(H, v,\mathbf{X}) = F(E_1, \ldots, E_n).\]
    Let $Y$ denote this sum and let $Y_x = F(E_1, \ldots, E_{x-1}, \emptyset, E_{x+1}, \ldots, E_n)$.

    Equivalently, if $H_{x}$ indicates the subhypergraph of $H$ where all hyperedges containing $x$ as the minimal vertex have been deleted, we have
    \[
        Y_{x} = \sum_{v \in [n]}g(H_{x},v,\mathbf{X}).
    \]
    Since $(E_x)_{x \in [n]}$ are independent and $0 \le g(H,v,\mathbf{X}) \leq 1$ for all $v$, we can apply the moment inequality \cite[Theorem 15.5]{boucheron2013concentration}: there exists some constant $c_1 > 0$ such that
    \[           \mathbb{E}\left[\Bigg|\sum_{v} g(H, v,\mathbf{X}) - \mathbb{E} \sum_{v} g(H, v,\mathbf{X})\Bigg|^2\right] \leq c_1\mathbb{E}\left[\sum_{x} (Y-Y_x)^2\right].
    \]

    Now fix vertices $v$ and $x$ of $H$. 
    Since $g$ is $s$-local, we must have $g(H, v,\mathbf{X}) - g(H_x, v, \mathbf{X})=0$ unless $v$ is in the $s$-neighborhood of $x$, denoted $N_{s}(H,x)$.
    This implies
    \begin{align*}
        |Y-Y_{x}| 
        = \left|\sum_{v\in N_{s}(H,x)} g(H, v,\mathbf{X}) - g(H_x, v, \mathbf{X})\right|
        \leq \sum_{v\in N_{s}(H,x)} g(H, v,\mathbf{X}) + g(H_x, v, \mathbf{X}) \leq 2\left|N_{s}(G,x)\right|.
    \end{align*}
    It follows that
    \begin{align*}
        \mathbb{E}\left[\sum_{x} (Y-Y_{x})^2\right] \leq 4  \mathbb{E}\left[\sum_{x}\left|N_{s}(H,x)\right|^2\right] 
        = 4n\mathbb{E}\left[\left|N_{s}(H,x)\right|^2\right]
    \end{align*}
    Note that $|N_s(H, x)| \mid |N_{s-1}(H, x)| \sim (r-1)|N_{s-1}(H, x)|\mathsf{Pois}(d)$.
    By repeated applications of the tower property of expectation, one can show that there exists some constant $c_2 \coloneqq c_2(r,s)>0$ such that 
    \[
        \E\left[\left|N_{s}(H,x)\right|^2\right] \leq c_2 d^{2s}.
    \]
    Putting everything together, we have
    \begin{align*}
        \mathbb{E}\left[\Bigg|\sum_{v \in R} g_{r}(H_v,\mathbf{x}) - \mathbb{E} \sum_{v \in R} g_{r}(H_v,\mathbf{x})\Bigg|^2\right] \leq 4c_1c_2d^{2s} n,
    \end{align*}
    as desired.
\end{proof}

We note that Wein employed a version of the above result for graphs \cite[Corollary 3.5]{wein2022optimal}, which he obtained as a corollary to a special case of a more general statement on graph functionals \cite[Proposition 12.3]{bordenave2023detection}.
With this in hand, we are ready to prove Lemma~\ref{Lemma:local_tree_np}.

\begin{proof}[Proof of Lemma~\ref{Lemma:local_tree_np}]
    Utilizing Proposition~\ref{prop:ER_Local_convergence_PW_tree}, the local weak convergence of $\Hrnp$ to $\PGWLT$, it is easy to deduce that the expectation of $|I|$ is $\E[g(T,o,\mathbf{X})]n+o(n)$.
    Furthermore, Proposition~\ref{prop: local alg moment inequality} implies that the variance is $O(n)$.
    The claim now follows by Chebyshev's inequality.
\end{proof}
    Observe that when applying an $s$-local algorithm $g$ to $T\sim\PGWLT$ with vertex labels $\mathbf{X}$, the distribution of $g(T, v, \mathbf{X})$ does not depend on the vertex $v$, since the distribution of $N_{s}(T,v)$ does not depend on the choice of $v$ in $T$. 
    Hence we can define the density of the independent set $I$ obtained by applying $g$ to $T$ as $\mathsf{density}(I) = \P[h(T, o, \mathbf{X})=1] = \E[h(T, o, \mathbf{X})]$.
    To assist with our proofs, we define
    \begin{align}\label{Def:optimal__ind_density_GWtree}
    \alpha_{r}(d) \coloneqq \left(\frac{d}{\log d}\right)^{\frac{1}{r-1}}\,\sup\left\{\E[g(T,o,\mathbf{X})] \,:\, s \geq 0,\, g \text{ is a valid } s\text{-local function}\right\}.
\end{align}
By Lemma~\ref{Lemma:local_tree_np}, $\alpha_{r}(d)\left(\frac{\log d}{d}\right)^{\frac{1}{r-1}}$ is the optimal density (up to the first order) of an independent set that a local algorithm can find in $\Hrnp$ with high probability.

The second hypertree we consider is the the rooted $r$-uniform $\Delta$-regular hypertree (denoted $T_{\Delta}^r$) with root $o$. 
Again, observe that when applying an $s$-local algorithm $h$ to $T_{\Delta}^r$ with vertex labels $\mathbf{X}$, the distribution of $h(T_{\Delta}^r, v, \mathbf{X})$ does not depend on the vertex $v$. 
Hence we can define the density of the independent set $I$ obtained by applying $h$ to $T_{\Delta}^r$ as $\mathsf{density}(I) = \P[h(T_{\Delta}^r, o, \mathbf{X})=1] = \E[h(T_{\Delta}^r, o, \mathbf{X})]$. 
Therefore, we can define the following parameter similar to \eqref{Def:optimal__ind_density_GWtree}:
\begin{align} \label{Def:optimal_ind_density_regtree}
    \alpha(r,\Delta) \coloneqq \left(\frac{\Delta}{\log \Delta}\right)^{\frac{1}{r-1}}\,\sup\left\{\E[h(T_\Delta^r,o,\mathbf{X})] \,:\, s \geq 0,\, h \text{ is a valid } s\text{-local function}\right\}.
\end{align}

\subsection{Interpolating Paths and Stable Algorithms}

We write $A \sim \Hrnp$ (resp. $A \sim \Hrrnp$) to denote a random vector $A \in \{0,1\}^m$ for $m={n \choose r}$ (resp. $m = n^r$), with i.i.d $\mathsf{Ber}(p)$ entries.
As mentioned in \S\ref{subsection: proof overview}, the proof of our impossibility results will involve analyzing a sequence of correlated hypergraphs, which we will refer to as interpolation paths.
Let us define these for each of the random hypergraph models we consider.

\begin{Definition}[Interpolation Paths for $\Hrnp$ and $\Hrrnp$]\label{definition: interpolation path}
    Let $T \in \N$ and let $m = \binom{n}{r}$ (resp. $m = n^r$).
    We construct a length $T$ interpolation path $A^{(0)}, \ldots, A^{(T)}$ as follows:  $A^{(0)} \sim \Hrnp$ (resp. $\Hrrnp$); then for each $1 \leq t \leq T$, $A^{(t)}$ is obtained from $A^{(t-1)}$ by resampling the coordinate $\sigma(t) \in [m]$ from $Ber\left(d/{n-1 \choose r-1}\right)$ (resp. $Ber\left(d/n^{r-1}\right)$). 
    Here $\sigma(t) = t-k_t m$ where $k_t$ is the unique integer satisfying $1 \leq \sigma(t) \leq m$. 
\end{Definition}

Consider the hypercube $Q_m$ with vertex set $\{0,1\}^m$ and edges of the form $uv$ where $u$ and $v$ differ in exactly one coordinate.
Interpolation paths can therefore be viewed as random walks on $Q_m$ (for the appropriate value of $m$).
We define $(D, \Gamma, c)$-stable algorithms with respect to this random walk (the distribution $\mathcal{P}$ below is either $\Hrnp$ or $\Hrrnp$ with the edge probability  $p$ set to $d/{n-1 \choose r-1}$ or $d/n^{r-1}$, respectively).

\begin{Definition} \label{Def:c-bad}
    Let $f : \{0,1\}^m \to \R^{n}$ and $c >0$. 
    An edge $uv$ of the hypercube $Q_m$ is $c$-bad for $f$ if
    \[
        \|f(u)-f(v)\|^2 \geq c \cdot \mathbb{E}_{A\sim \mathcal{P}}\left[\|f(A)\|^2\right]
    \]
\end{Definition}

\begin{Definition}\label{definition: (D, G, c)-stable}
    Let $\Gamma \in \N$ and $c > 0$. 
    Let $A^{(0)}, \ldots, A^{(T)}$ be the interpolation path from Definition~\ref{definition: interpolation path} of length $T = \Gamma m$, with appropriate $p \leq 1/2$.
    A degree-$D$ polynomial $f : \{0,1\}^m \to \R^{n}$ is $(D, \Gamma, c)$-stable if
    \[
        \P\left[\text{no edge of } A^{(0)}, \ldots, A^{(T)} \text{ is $c$-bad for $f$}\right] \geq p^{4\Gamma D/c}.
    \]
\end{Definition}

The following lemma will be crucial to proving our impossibility results.
The proof follows an identical argument to that of \cite[Lemmas~2.7 and 2.8]{wein2022optimal} on graphs and so we omit it here.

\begin{lemma}\label{lemma: stable}
    For any constants $\Gamma \in \N$ and $c > 0$, any degree-$D$ polynomial $f$ on $A\sim \Hrnp$ or $A \sim \Hrrnp$ is $(D, \Gamma, c)$-stable.
\end{lemma}

%% file: hypergraph.tex
\section{Independent Sets in $\Hrnp$}

In this section, we will consider the random $r$-uniform hypergraph $\Hrnp$ for $p = d/{n-1 \choose r-1}$ (see Definition~\ref{definition: models}).
We split this section into three subsections.
First, we prove the achievability result (Theorem~\ref{theorem: low-deg hypergraph achievability}).
In the second subsection, we prove a version of the \textit{Overlap Gap Property} for independent sets in $\Hrnp$.
In the final subsection, we apply this result to prove intractability of low-degree algorithms for independent sets in the stated regime.

\subsection{Proof of the Achievability Result}

In this section we prove Theorem~\ref{theorem: low-deg hypergraph achievability}, which shows that there exist low-degree algorithms that can find independent sets of density $(1-\varepsilon)\left(\frac{1}{r-1}\cdot \frac{\log d}{d}\right)^{\frac{1}{r-1}}$ in $\Hrnp$ for $p = d/\binom{n-1}{r-1}$. 
We further split this subsection into two parts.
First, we show how local algorithms for finding independent sets in $T_\Delta^r$ can be adapted to find independent sets in $\PGWLT$.
Next, we show how to describe a low-degree algorithm for finding independent sets in $\Hrnp$ from a local algorithm for finding independent sets in $\PGWLT$.

\subsubsection{Local Algorithms: from $T_\Delta^r$ to $\PGWLT$}

In this section, we prove that there exists a local algorithm that outputs an independent set of the desired density on $\PGWLT$, which is crucial for the construction of our low-degree polynomial in the sequel.
First, we prove the following theorem.

\begin{Theorem}\label{Theorem:local_achievability}
    For any $\varepsilon>0$ and any sufficiently large $d = d(\varepsilon) > 0$, there exists $s \coloneqq s(\varepsilon,d)$ and a valid $s$-local function $g$ such that for $(T,o) \sim \PGWLT$ and i.i.d $\mathsf{Unif}[0,1]$ labels $\mathbf{X}$ over its vertices, $\E[g(T,o,\mathbf{X})] \geq (1-\varepsilon)\left(\frac{1}{r-1}\cdot \frac{\log d}{d}\right)^{\frac{1}{r-1}}$.
\end{Theorem}

We start by describing the following random greedy algorithm analyzed in \cite{nie2021randomized}, which generates an independent set in an $r$-uniform $\Delta$-regular hypergraph $H$ with large girth (which implies the case of $T_\Delta^r$):
\begin{itemize}
    \item Equip each vertex $v$ with i.i.d. labels $\mathbf{X}_v \sim \mathsf{Unif}[0,1]$.
    \item Iteratively select the vertex with the largest label of the remaining vertices of $H$, and add it to the independent set.
    Next, remove any remaining vertices that form an edge with the selected vertices, and repeat until no vertices remain.
\end{itemize}
We remark that while this algorithm as described is not local, it can be simulated by one that is (this can be inferred by the discussion in \cite[\S2]{nie2021randomized}).
Nie and Verstra\"ete~\cite[Theorem 4]{nie2021randomized} showed that the random greedy algorithm outputs an independent set of density $(1-\varepsilon)\left(\frac{1}{r-1}\cdot \frac{\log \Delta}{\Delta}\right)^{1/r-1}$ on $T_\Delta^r$.

Following an idea of Rahman and Vir\'ag in \cite{rahman2017local}, we will show that a local algorithm finding an independent set in $T_\Delta^r$ can be adapted to construct an independent set in $\mathbb{T}_d^{GW}$ as well. 
We remark that the parameter $d$ we choose will not be equal to $\Delta$, however, it will be chosen such that $\log d/d = (1\pm o_d(1))\log \Delta/\Delta$.
In particular, the independent set output by our local algorithm on $\PGWLT$ will have the same density as the one in $T_{\Delta}^r$ for $d$ sufficiently large.
Let $E(d,\Delta)$ denote the event that the root of $\mathbb{T}_d^{GW}$ and all of its neighbors have degree at most $\Delta$.
With this definition in hand, we prove the following proposition:

\begin{Proposition}\label{Prop:h_regular_to_g_GW}
    Given a local algorithm $h$ to find an independent set in $T_\Delta^{r}$, there exists a local algorithm $g$ satisfying the following:
    \[
        \E_{\mathbf{Y}}[h(T_\Delta^{r},o,\mathbf{Y})] \, \P[E(d,\Delta)] \,\leq\, \E[g(T,o,\mathbf{X})] \,\leq\,  \E_{\mathbf{Y}}[h(T_\Delta^{r},o,\mathbf{Y})].
    \]
    for any $T \sim \mathbb{T}_d^{GW}$ and i.i.d $\mathsf{Unif}[0,1]$ vertex labels $\mathbf{X},\, \mathbf{Y}$.
\end{Proposition}

\begin{proof}
    Let $I$ be an independent set generated by $h$ on $T_\Delta^r$ and some i.i.d $\mathsf{Unif}[0,1]$ vertex labels, and let $T \sim \mathbb{T}_d^{GW}$.
    The local algorithm $g$, that can generate the independent set $J$ of the desired density in $T$, is described explicitly through the following three steps.
    \begin{enumerate}[label=\textbf{Step~\arabic*}:, leftmargin=0.75in, wide]
        \item\label{step:removal} \textbf{Edge removal}. 
        We remove edges from $T$ to define a hyperforest $S$ where all vertices have degree at most $\Delta$.
        Begin with a random labelling $\mathbf{X}$ of the vertices of $T$.
        Consider a vertex $v$ satisfying $\deg_T(v) > \Delta$.
        For each edge $e \ni v$, let $X_e \coloneqq \max\set{\mathbf{X}_u\,:\, u \in e\setminus\set{v}}$.
        Order the edges containing $v$ by $X_e$ and remove the edges with the $\deg_T(v) - \Delta$ largest values (break ties arbitrarily).
        It is not difficult to see that $S$ is a disjoint collection of countably many $r$-uniform hypertrees.

        \item \textbf{Regularization}.
        We describe how to construct a $\Delta$-regular hyperforest $T'$ from $S$ such that $S \subseteq T'$.
        Define the $(\Delta - 1)$-ary hypertree to be the rooted hypertree in which every vertex has exactly $\Delta-1$ offspring edges.
        Consider $r-1$ such $(\Delta-1)$-ary hypertrees with roots $u_1, u_2, \ldots, u_{r-1}$ and create an edge $\set{v, u_1, \ldots, u_{r-1}}$. 
        Repeat this process $\Delta-\deg_{S}(v)$ many times for each $v \in V(S)$.
        Randomly label $T^{'}$ by a new labeling $\mathbf{X}'$ independent of $\mathbf{X}$.

        \item \textbf{Inclusion}.
        Since $T'$ is a disjoint collection of $r$-uniform $\Delta$-regular hypertrees, we can use $h$ with input $\mathbf{X}'$ to construct an independent set $I'$ of $T'$ with the same density as $I$. 
        However, due to the removal of edges from $T$, $I'$ may no longer be independent in $T$.
        We construct the desired set $J$ from $I'$ by including all vertices $v \in I'$ such that no edges containing $v$ were removed during \hyperref[step:removal]{Step 1}.
    \end{enumerate}

     By design, $g$ is a local algorithm on $\mathbb{T}_d^{GW}$. 
     We now show that $J$ is an independent set. Suppose for contradiction that $e \subseteq J$ for some  $e \in E(T)$. Then $e \subseteq I'$ as well.
     Since $I'$ is an independent set in $T'$, $e$ must have been removed during \hyperref[step:removal]{Step 1}.
     By the construction in the last step, none of the vertices in $e$ would have been added to $J$. 
     Therefore, $J$ is indeed an independent set in $T$.

     Let us now prove the desired bounds on the density of $J$. 
     Clearly, $\mathsf{density}(J) \leq \mathsf{density}(I')$ since $J \subseteq I'$. 
     Furthermore, for any $v \in I'$, if $v$ and all of its neighbors in $T$ have degree at most $\Delta$ then none of the edges containing $v$ will be removed, which implies that $v \in J$. 
     Hence the lower bound follows by the previous observation.
\end{proof}

\begin{lemma}\label{lemma:Pro_E_lambda_d}
    If $d= \Delta-\Delta^{u}$ for any $1/2 < u < 1$ then $\P[E(d,\Delta)] \to 1$ as $\Delta \to \infty$.
\end{lemma}

\begin{proof} 
    Let $p(d,\Delta) \coloneqq \P[\mathsf{Pois}(d)\geq \Delta]$.     
    A direct Chernoff bound using the moment generating function of $\mathsf{Pois}(d)$ gives
    \begin{align*}
       p(d,\Delta) \leq
       p_{err} \coloneqq \exp\left\{-\Big(1+O\left(\Delta^{u-1}\right)\Big)\frac{\Delta^{2u-1}}{2}\right\}.
    \end{align*}
    By a simply union bound, with probability at least $1-r\Delta \cdot p_{err}$, the event $E(d,\Delta)$ holds. 
\end{proof}

We are now ready to prove Theorem~\ref{Theorem:local_achievability}.

\begin{proof}[Proof of Theorem~\ref{Theorem:local_achievability}]
    Given $d$, let $\Delta \coloneqq \lceil d+d^{3/4}\rceil$. 
    Recall the definitions of $\alpha_{r}(d)$ and $\alpha(r,\Delta)$ in \eqref{Def:optimal__ind_density_GWtree} and \eqref{Def:optimal_ind_density_regtree}, respectively.
    By the conclusion of Proposition~\ref{Prop:h_regular_to_g_GW}, we have
    \[
        \alpha_{r}(d)\left(\frac{\log d}{d}\right)^{\frac{1}{r-1}} \geq \alpha(r,\Delta)\left(\frac{\log \Delta}{\Delta}\right)^{\frac{1}{r-1}} \cdot \P[E(d,\Delta)].
    \]
    The random greedy algorithm discussed at the beginning of this section implies that
    \[
        \liminf_{\Delta \to \infty} \alpha(r,\Delta) \geq \left(\frac{1}{r-1}\right)^{\frac{1}{r-1}}. 
    \]
    By the choice of $\Delta$ (as a function of $d$), we have $(\log \Delta/\Delta)/(\log d/d) \to 1$ as $d \to \infty$. 
    Combining the above with Lemma~\ref{lemma:Pro_E_lambda_d}, we may conclude
    \[
        \liminf_{d \to \infty} \alpha_r(d) \geq \left(\frac{1}{r-1}\right)^{\frac{1}{r-1}},
    \]
    which completes the proof.
\end{proof}

\subsubsection{From Local to Low-Degree}

Before we prove the main result (Theorem~\ref{theorem: low-deg hypergraph achievability}), we state two useful lemmas.

Let $\mathcal{T}$ be a set of rooted $r$-uniform hypertrees consisting of one representative from each isomorphism class of rooted hypertrees of depth at most $2s$, and let $\mathcal{T}_q \subseteq \mathcal{T}$ contain only those hypertrees with at most $q$ edges.
For $H \sim \mathcal{H}_r\left(n, d/{n-1 \choose r-1}\right)$ and for $T \in \mathcal{T}$, let $n_T$ denote the number of occurrences of $T$ in a local neighborhood of $H$, that is 
\[
    n_T = \big|\{v \in [n]:N_{2s}(H,v) \cong T\}\big|.
\] 
Here, $\cong$ denotes the isomorphism of rooted hypergraphs. 
The following lemma shows that $n_T$ is well-concentrated.

\begin{lemma}\label{lemma: concentration_n_T}
    There is a constant $c \coloneqq c(r) > 0$ such that for any $t \geq (2e)^{3/2}c\sqrt{n}(2rd)^{2s}$, we have
    \[
        \P\left[\big|n_T-\E[n_T]\big| \geq t \right] \leq \exp\left(-\frac{3t^{2/3}}{2ec^{2/3}n^{1/3}(2rd)^{2s/3}}\right).
    \]
\end{lemma}
\begin{proof}
    The proof follows by combining Markov's inequality with a hypergraph version of Proposition 12.3~\cite{barak2019nearly} (for which the proof is identical, \textit{mutatis mutandis}).
\end{proof}

Let $p_T$ denote the probability that $T$ appears as the $2s$-neighborhood of a rooted $\PGWLT$, i.e.,
\[
    p_T = \P_{(W,o)\sim \PGWLT}[N_{2s}(W,o) \cong T].
\]
The following lemma shows that the local neighborhood of $\Hrnp$ converges to $\PGWLT$.

\begin{lemma}\label{lemma:H_close_to_T_Eg}
    There is a constant $c\coloneqq c(r) > 0$ such that the following holds for $n$ sufficient large:
    \[
        \Big|\E[n_T] - p_T\cdot n\Big| \leq cn^{3/4}\log n.
    \]
\end{lemma}

\begin{proof}
    The proof follows by a hypergraph version of \cite[Lemma 12.4]{barak2019nearly} (for which the proof is identical, \textit{mutatis mutandis}), and linearity of expectation.
\end{proof}

\begin{Corollary}\label{corollary:concentration_nT_to_pTn}
    For any $\tau>0$ and $n$ sufficiently large in terms of $d,\,s,\,r,$ and $\tau$, there is a constant $C \coloneqq C(d,s,r,\tau) > 0$ such that
    \[
        \P\left[\big|n_T-p_T\cdot n\big| \geq \tau n\right] \leq \exp(-Cn^{1/3}).
    \]
\end{Corollary}
\begin{proof}
    Follows directly by combining Lemmas~\ref{lemma: concentration_n_T} and \ref{lemma:H_close_to_T_Eg}.
\end{proof}

Before we prove the main result, we make an analogous observation to the one in \cite[Lemma 2.11]{wein2022optimal} on graphs.
In particular, a random polynomial can be converted to a deterministic polynomial that works almost as well.

\begin{lemma}\label{lemma:ram_lowdegree_same_as Deter_lowdegree}
    Suppose $f$ is a random degree-$D$ polynomial that $(k,\delta,\xi, \eta)$-optimizes the independent set problem in $\Hrnp$. Then for any $c>2$, there exists a deterministic degree-$D$ polynomial that  $(k,c\delta,c\xi, \eta)$-optimizes the independent set problem in $\Hrnp$. 
\end{lemma}

\begin{proof}
    By Definition~\ref{Def:f_Optimize_ind}, we have
    \[
        \E_{A,\omega}[\|f(A,\omega)\|^2] \leq \xi k, \quad \text{and} \quad 
        \P_{A,\omega}[|V_f^{\eta}(A,\omega)|<k] \leq \delta.
    \]
    By Markov's inequality,
    \[
        \P_{\omega}\left[\E_{A}[\|f(A,\omega)\|^2] \geq c \xi k\right] \leq \frac{1}{c}<\frac{1}{2},
        \quad \text{and} \quad
         \P_{\omega}\left[ \P_{A}[|V_f^{\eta}(A,\omega)|<k] \geq c\delta\right] \leq \frac{1}{c}<\frac{1}{2}.
    \]
    Thus, there exists an $\omega^{*} \in \Omega$ such that the deterministic degree-$D$ polynomial $f^{*}(\cdot) = f(\cdot, \omega)$ satisfies
    \[
         \E_{A}[\|f^*(A)\|^2] \leq c \xi k \quad \text{and} \quad 
        \P_{A}[|V_f^{\eta}(A)|<k] \leq c \delta
    \]
    as desired.
\end{proof}

We are now ready to prove Theorem~\ref{theorem: low-deg hypergraph achievability}.

\begin{proof}[Proof of Theorem~\ref{theorem: low-deg hypergraph achievability}]
    Given $\varepsilon>0$, let $d$ and $s$ be sufficiently large in order to apply Theorem~\ref{Theorem:local_achievability}.
    As a result, we obtain a valid $s$-local function $g$ that outputs independent sets of density 
    \[\E[g(T,o,\mathbf{X})] \geq (1-\varepsilon/5)\left(\frac{1}{r-1}\frac{\log d}{d}\right)^{\frac{1}{r-1}},\] 
    when $(T,o)\sim \PGWLT$ and $\mathbf{X}$ encodes i.i.d $\mathsf{Unif}[0,1]$ labels on the vertices of $T$.
    By Lemma~\ref{lemma:ram_lowdegree_same_as Deter_lowdegree}, it is sufficient to prove the results for a random polynomial instead of deterministic one (up to a change in the constants $\xi$, $C$). 
    We will construct a random polynomial $f:\R^{n \choose r} \to \R^n$ where the input $A$ to $f$ represents a hypergraph on vertex set $[n]$, and the internal randomness of $f$ samples vertex labels $\mathbf{X}$ i.i.d from $\mathsf{Unif}[0,1]$. 
    Let $\mathbb{H}_{v,s,q}$ be the collection of hypergraphs $H$ on vertex set $[n]$ for which $|E(H)| \leq q$ and every non-isolated vertex is reachable from $v$ by a path of length at most $s$. 
    In other words, $\mathbb{H}_{v,s,q}$ consists of all possible $s$-neighborhoods of $v$ of size at most $q$. Let 
    \begin{align}\label{Def:f_v_based_on_local_g}
        f_{v}(A,\mathbf{X}) = \sum_{H \in  \mathbb{H}_{v,s,q}} \alpha(H,v,\mathbf{X})\prod_{e \in E(H)}A_e,
    \end{align}
    where $\alpha(H,v,\mathbf{X})$ is defined recursively by
    \begin{align}\label{Def:Coeff_of_f_v}
        \alpha(H,v,\mathbf{X}) = g(H,v,\mathbf{X}) - \sum_{\substack{H' \in \mathbb{H}_{v,s,q} \\ E(H') \subsetneq E(H)}} \alpha(H',v,\mathbf{X}).
    \end{align}
    Observe that by this construction, we have the following for any $v \in [n]$: if $N_{s}(A,v)$ is a hypertree with $|N_{s}(A,v)| \leq q$, then $f_v(A,\mathbf{X})= g(H,v,\mathbf{X})$, where $q = q(\varepsilon,d,\eta)$ will be chosen later.

    We first show that, under this construction, the rounding procedure $V_f^{\eta}(A, \mathbf{X})$ will not fail.
    Fix any vertex $v$, if $N_{2s}(A,v)$ is a hypertree with $|N_{2s}(A,v)|\leq q$, then for all $u \in V(N_1(A,v))$, $N_s(A,u)$ is a hypertree with with $|N_s(A,u)|\leq q$. 
    Hence $f_u(A,\mathbf{X}) = g(A,u,\mathbf{X})$ by construction.
    Since the output of $g$ is an independent set, $v$ will not be in the ``bad'' set $(I \setminus \Tilde{I}) \cup J$ (see Definition~\ref{Def:V_eta_f}). 
    Therefore, $(I \setminus \Tilde{I}) \cup J$ is disjoint from the following set:
    \[
        V_q := \bigcup_{T \in\mathcal{T}_q} \set{v \in [n] \,:\, N_{2s}(A,v)\cong T}.
    \]
    By Corollary~\ref{corollary:concentration_nT_to_pTn}, we have that $|n_T-p_T\cdot n|\leq \eta n/(2|\mathcal{T}_q|)$ with probability $1-\exp(-\Omega(n^{1/3}))$, where $\Omega(\cdot)$ hides a constant depending on $\varepsilon,d,r,s,q,\eta$. 
    Choose $q$ sufficiently large such that
    \[
        \sum_{T \in \mathcal{T}_q}p_T \geq 1-\eta/2.
    \]
    Then we have
    \[
        |I \setminus \Tilde{I}|+|J| = |(I \setminus \Tilde{I}) \cup J| \leq n-|V_q| \leq n- \sum_{T \in T_q}\left(p_T\cdot n-\frac{\eta n}{2|\mathcal{T}_q|}\right)
        = \left(1- \sum_{T \in \mathcal{T}_q}p_T\right)+\frac{\eta n}{2} \leq \eta n,
    \]
    which proves that the rounding procedure succeeds with probability at least $1-\exp(-\Omega(n^{1/3}))$.

    Next we show that the independent set $I_f :=V^{\eta}_f(A,\mathbf{X})$ has the desired density with high probability. 
    By our choice of the local function $g$, we have
    \[
        \left(1-\frac{\varepsilon}{5}\right)\left(\frac{1}{r-1}\cdot\frac{\log d}{d}\right)^{1/r-1} \leq \E_{(T,o)\sim\PGWLT}[g(T,o,\mathbf{X})] = \sum_{T\in \mathcal{T}}p_T\phi_T,
    \]
    where $\phi_T$ is defined as the probability over $\mathbf{X}$ that $g(A,v,\mathbf{X})=1$ conditioned on the event $N_{2s}(A,v) \cong T$.    
    Let $q$ be sufficiently large such that 
    \[\sum_{T \in \mathcal{T}_q}p_T \geq 1-\frac{\varepsilon}{5}\left(\frac{1}{r-1}\cdot\frac{\log d}{d}\right)^{1/r-1}.\] 
    Given that $\phi_T \in [0,1]$ for any $T$, we have
    \begin{align*}
        \frac{\varepsilon}{5}\left(\frac{1}{r-1}\cdot\frac{\log d}{d}\right)^{1/r-1} 
        &\geq 1-\sum_{T \in \mathcal{T}_q}p_T 
        = \sum_{T \in \mathcal{T}}p_T - \sum_{T \in \mathcal{T}_q}p_T \\
        &\geq  \sum_{T \in \mathcal{T}}p_T\phi_T - \sum_{T \in \mathcal{T}_q}p_T\phi_T \\
        &\geq   \left(1-\frac{\varepsilon}{5}\right)\left(\frac{1}{r-1}\cdot\frac{\log d}{d}\right)^{1/r-1} - \sum_{T \in \mathcal{T}_q}p_T\phi_T,
    \end{align*}
    which implies
    \[\sum_{T \in \mathcal{T}_q}p_T\phi_T \geq \left(1-\frac{2\varepsilon}{5}\right)\left(\frac{1}{r-1}\cdot\frac{\log d}{d}\right)^{1/r-1}.\]
    Using Corollary~\ref{corollary:concentration_nT_to_pTn} once again, with probability at least $1-\exp(-\Omega(n^{1/3}))$ over $A \sim \Hrnp$, we have
    \begin{align*}
        \sum_{T \in \mathcal{T}_q}n_T\phi_T &\geq  \sum_{T \in \mathcal{T}_q}\left(p_T n - \frac{\varepsilon}{5|\mathcal{T}_q|}\left(\frac{1}{r-1}\cdot\frac{\log d}{d}\right)^{1/r-1}n\right)\phi_T \\
        &= \left(\sum_{T \in \mathcal{T}_q}p_T\phi_T\right)n -  \frac{\varepsilon}{5}\left(\frac{1}{r-1}\cdot\frac{\log d}{d}\right)^{1/r-1}n
        \\
        &\geq 
        n\left(1-\frac{3\varepsilon}{5}\right)\left(\frac{1}{r-1}\cdot\frac{\log d}{d}\right)^{1/r-1}.
    \end{align*}
    Now fix $A$ satisfying the inequality above and consider the randomness over the vertex label $\mathbf{X}$. 
    As we showed above, $V_q$ is disjoint from the bad set $(I \setminus \Tilde{I}) \cup J$.
    Therefore, if $N_{2s}(A,v) \cong T$ for some $v \in V(H)$ and $T \in \mathcal{T}_q$, then $v \in I_f$ if and only if $g(A,v, \mathbf{X}) = 1$, which happens with probability $\phi_v = \phi_T$. 
    
    The idea now is to partition $V_q$ into subsets $U_1, \ldots, U_{(r-1)q+1}$ such that for each $U_i$, the vertices in $U_i$ have disjoint $s$-neighborhoods. 
    In particular, the random variables $\{\mathbbm{1}_{v \in I_f}\}_{v \in U_i}$ are independent conditioned on the outcome $A$.
    Each vertex $v \in V_q$ has at most $(r-1)q$ vertices in $N_{2s}(A,v)$, which implies there are at most $(r-1)q$ vertices $u \in V_q$ such that $u \neq v$ and $N_s(A,v) \cap N_s(A,u) \neq \varnothing$. 
    We can greedily form the partition $\set{U_i}$ while maintaining the disjointness property (given that there are $(r-1)q+1$ many such subsets available).
    For such a partition $\{U_i\}$, we apply Chernoff's bound for each $i$ to obtain the following:
    \begin{align}\label{Ineqaulity:Chernoff_v_in_I_f}
        \P_{\mathbf{X}}\left[\sum_{v\in U_i}\mathbbm{1}_{v \in I_f} \leq \left(1-\frac{\varepsilon}{5}\right)\mu_i\right] \leq \exp \left(-\frac{\varepsilon^2}{50}\mu_i\right)
    \end{align}
    where
    \[
        \mu_i = \E_{\mathbf{X}}\left[\sum_{v \in U_i}\mathbbm{1}_{v \in I_f}\right] = \sum_{v \in U_i}\phi_v.
    \]
    We say a subset $U_i$ is ``large'' if the corresponding value $\mu_i$ satisfies
    \[\mu_i \geq \frac{\varepsilon}{5((r-1)q+1)}\left(\frac{1}{r-1}\cdot\frac{\log d}{d}\right)^{1/r-1}n,\]
    and ``small'' otherwise.
    By \eqref{Ineqaulity:Chernoff_v_in_I_f} and a union bound over $i$, we have that every large subset $U_i$ satisfies $\sum_{v\in U_i}\mathbbm{1}_{v \in I_f} \geq \left(1-\frac{\varepsilon}{5}\right)\mu_i$ with probability $1-\exp(-\Omega(n))$. With this in hand, we have
    \begin{align*}
        |I_f| 
        &\geq \sum_{v \in V_q} \mathbbm{1}_{v \in I_f} 
        \geq \sum_{i: U_i\text{ large}}\sum_{v \in U_i}\mathbbm{1}_{v \in I_f}
        \geq \sum_{i: U_i\text{ large}}\left(1-\frac{\varepsilon}{5}\right)\mu_i
        \\
        &= \left(1-\frac{\varepsilon}{5}\right)\left[\sum_{i}\mu_i-\sum_{i: U_i\text{ small}}\mu_i \right] 
        \\
        &\geq  \left(1-\frac{\varepsilon}{5}\right)\left[\sum_{i}\mu_i- \sum_{i: U_i\text{ small}}\frac{\varepsilon}{5((r-1)q+1)}\left(\frac{1}{r-1}\cdot\frac{\log d}{d}\right)^{\frac{1}{r-1}}n\right]\\
        &\geq \left(1-\frac{\varepsilon}{5}\right)\left[\sum_{i}\mu_i- \frac{\varepsilon}{5}\left(\frac{1}{r-1}\cdot\frac{\log d}{d}\right)^{\frac{1}{r-1}}n\right] \\
        &=  \left(1-\frac{\varepsilon}{5}\right)\left[\sum_{v \in V_q}\phi_v- \frac{\varepsilon}{5}\left(\frac{1}{r-1}\cdot\frac{\log d}{d}\right)^{\frac{1}{r-1}}n\right] \\
        &=  \left(1-\frac{\varepsilon}{5}\right)\left[\sum_{T \in \mathcal{T}_{q}}n_T\phi_T- \frac{\varepsilon}{5}\left(\frac{1}{r-1}\cdot\frac{\log d}{d}\right)^{\frac{1}{r-1}}n\right]\\
        &\geq  \left(1-\frac{\varepsilon}{5}\right)\left[\left(1-\frac{3\varepsilon}{5}\right)\left(\frac{1}{r-1}\cdot\frac{\log d}{d}\right)^{\frac{1}{r-1}}n- \frac{\varepsilon}{5}\left(\frac{1}{r-1}\cdot\frac{\log d}{d}\right)^{\frac{1}{r-1}}n\right] \\
        &=  \left(1-\frac{\varepsilon}{5}\right) \left(1-\frac{4\varepsilon}{5}\right)\left(\frac{1}{r-1}\cdot\frac{\log d}{d}\right)^{\frac{1}{r-1}}n \\
        &\geq \left(1-\varepsilon\right) \left(\frac{1}{r-1}\frac{\log d}{d}\right)^{\frac{1}{r-1}}n.
    \end{align*}
    Thus, the random polynomial $f$ outputs an independent set of density at least $\left(1-\varepsilon\right) \left(\frac{1}{r-1}\cdot\frac{\log d}{d}\right)^{\frac{1}{r-1}}$ with probability $1-\exp(-\Omega(n^{1/3}))$ over both $A$ and $\mathbf{X}$.

    Finally, we need to verify the normalization condition in Definition~\ref{Def:f_Optimize_ind}:
    \[\E_{A,\mathbf{X}}[\|f(A,\mathbf{X})\|^2] \leq \xi \left(1-\varepsilon\right) \left(\frac{1}{r-1}\frac{\log d}{d}\right)^{\frac{1}{r-1}}n,\quad  \text{for some constant} \quad  \xi \geq 1.\]
    It suffices to show that $\E_{A,\mathbf{X}}[f_v(A,\mathbf{X})^2] = O(1)$ for all vertices $v$. 
    Fix a vertex $v$ and let $N = |N_s(A,v)|$. By \eqref{Def:f_v_based_on_local_g}, for each $H \in \mathbb{H}_{v,s,q}$ the corresponding term is not $0$ only if $H$ is a subgraph of $N_s(A,v)$. This implies that number of of nonzero terms is at most
    \[
        {N \choose \leq q} = \sum_{i=0}^q {N \choose i} \leq \sum_{i=1}^q N^i \leq (N+1)^{q}.
    \]
    Observe that by \eqref{Def:Coeff_of_f_v}, $|\alpha(H,v,\mathbf{X})| \leq \chi$ for some $\chi = \chi(r,s,q)$ independent of $v$ and $\mathbf{X}$. 
    Therefore, we may conclude that
    \begin{align} \label{Inequality: bound_on_Fv(X,Y)2}
        f_v(A, \mathbf{X})^2 \leq [\chi(N+1)^q]^2 = \chi^2(N+1)^{2q}.
    \end{align}
    In order to bound the $\E_{A,\mathbf{X}}[f_v(A,\mathbf{X})^2]$, it is sufficient to prove a tail bound for $N$.
    Starting from $m_0=1$, let $m_i$ be the number of vertices whose distance from $v$ in $A$ is exactly $i$. Conditioned on $m_i$, $m_{i+1}$ is stochastically dominated by $\mathsf{Bin}\left(r{n-1 \choose r-1}m_i, d/{n-1 \choose r-1}\right)$.
    By a Chernoff bound, for fixed $m_i \geq 1$ and any $\delta \geq 1$, we have
    \[
        \P[m_{i+1} \geq (1+\delta)rdm_i] \leq \exp\left(-\frac{\delta rdm_i}{3}\right) \leq \exp\left(-\frac{\delta rd}{3}\right).
    \]
    Note that this inequality holds trivially if $m_i = 0$.
    Taking a union bound over all $i \in [s]$,
    \[
        \P\big[m_i < [(1+\delta)rd]^i, \forall i \in [s]\big] \geq 1-s\cdot \exp\left(-\frac{\delta rd}{3}\right).
    \]
    In particular, on this event, 
    \[
        N<\sum_{i=1}^s [(1+\delta)rd]^i \leq [(1+\delta)rd+1]^s \leq (2\delta rd+1)^s \leq (3\delta rd)^s.
    \]
    We conclude that,
    \[
        \P[N \geq (3\delta rd)^s] \leq s\cdot \exp\left(-\frac{\delta rd}{3}\right), ~~~\forall \delta\geq 1.
    \]
    For $t=(3\delta rd)^s$, we have a subexponentially small tail bound on $N$:
    \[
        \P[N \geq t] \leq s\cdot \exp\left(-\frac{t^{1/s}}{9}\right), ~~~\forall t \geq (3rd)^s.
    \]
    Putting everything together, we have
    \begin{align*}
        \E_{A,\mathbf{X}}[(f_v(A,\mathbf{X})^2] 
        &\leq \sum_{i =0}^{\infty} \chi^2(t+1)^{2q} \P[N=t] 
        \\
        &\leq   \sum_{i =0}^{\lceil(3rd)^s\rceil} \chi^2(i+1)^{2q}+\sum_{i=\lceil(3rd)^s\rceil+1}^{\infty} \chi^{2}(i+1)^{2q}s\exp\left(-\frac{i^{1/s}}{9}\right),
    \end{align*}
   which is a finite constant independent of $n$.
\end{proof}

\subsection{Forbidden Structures}

In this section, we will define two ``forbidden structures'' in order to assist with our proof of Theorem~\ref{theorem: low-deg hypergraph impossibility}.
Before we define these structures and prove the corresponding results, we state the following theorem due to Krivelevich and Sudakov on the size of the maximum independent set in $\Hrnp$ (this result wasn't explicitly stated, but can inferred from Corollary~1 in their paper) \cite{krivelevich1998chromatic}.

\begin{Theorem}\label{theorem: hypergraph independence stat thresh}
    Let $\eps > 0$ and $r \in \N$ such that $r\geq 2$.
    There is $d_0 > 0$ such that for any $d \geq d_0$, there exists $n_0 > 0$ such that for any $n \geq n_0$ and $p = d/\binom{n-1}{r-1}$, the graph $H \sim \Hrnp$ satisfies
    \[(1 - \eps)\left(\frac{r}{r-1}\frac{\log d}{d}\right)^{\frac{1}{r-1}} \,\leq\, \frac{\alpha(H)}{n} \,\leq\,  (1 + \eps)\left(\frac{r}{r-1}\frac{\log d}{d}\right)^{\frac{1}{r-1}},\]
    with probability at least $1 - \exp\left(-\Omega(n)\right)$.
\end{Theorem}

First, we will define a sequence of correlated random graphs and a corresponding forbidden structure.
Recall that we represent an $r$-uniform hypergraph $H \sim \Hrnp$ by its edge indicator vector $A = \set{0, 1}^m$ where $m = \binom{n}{r}$.
Additionally, recall the interpolation path for $\Hrnp$ from Definition~\ref{definition: interpolation path}.
It is easy to see that the marginal distribution of $A^{(t)}$ is $\Hrnp$.
Furthermore, $A^{(t+m)}$ is independent of $A^{(0)}, \ldots, A^{(t)}$ as all edges are resampled by then.
We will now show that a certain structure of independent sets across the correlated sequence of random hypergraphs exists only with exponentially small probability.

\begin{Proposition}\label{Prop:OGP}
    Fix constants $\varepsilon>0$ and $K \in \N$ with $K \geq \left \lceil \frac{2^{r+1}}{\varepsilon^r}\right \rceil+ 1$. Consider the interpolation path $A^{(0)}, \cdots, A^{(T)}$ of any length $T = n^{O(r)}$. If $d=d(\varepsilon,K)>0$ is sufficiently large, then with probability $1-\exp\Big(-\Omega(n)\Big)$, there does not exist a sequence of sets $S_1, \cdots, S_K \subseteq [n]$ satisfying the following properties
    \begin{enumerate}[label=(H\arabic*)]
        \item\label{item:exist_indset} For each $k \in [K]$ there exist $0 \leq t_k \leq T$ such that $S_{k}$ is an independent set in $A^{(t_k)}$,
        \item\label{item:large_size} $|S_k| \geq (1+\varepsilon)\left(\frac{1}{r-1}\cdot \frac{\log d}{d}\right)^{\frac{1}{r-1}}n$ for all $k \in [K]$,
        \item\label{item:ogp} $|S_k \setminus (\cup_{\ell<k}S_\ell)| \in \left[\frac{\varepsilon}{4}\left(\frac{1}{(r-1)}\cdot \frac{\log d}{d}\right)^{\frac{1}{r-1}}n,\, \frac{\varepsilon}{2}\left(\frac{1}{(r-1)}\cdot \frac{\log d}{d}\right)^{\frac{1}{r-1}}n\right]$ for all $2 \le k \le K$.
    \end{enumerate}
\end{Proposition}

\begin{proof}
    
    Let $N$ be the number of the sequences $(S_1, \ldots, S_K)$ satisfying the properties \ref{item:exist_indset}--\ref{item:ogp}. 
    We will compute $\E[N]$ and show that it is exponentially small, then the assertion follows by Markov's inequality. 

    Define $\Phi := \left(\frac{\log d}{d}\right)^{\frac{1}{r-1}}n$ and let $a_k, b_k$ and $c$ be defined as follows:
    \[a_k = \frac{|S_k|}{\Phi}, \quad b_k = \frac{|S_k \setminus \cup_{l<k}S_\ell|}{\Phi}, \quad \text{and} \quad c = \frac{|\cup_{k}S_k|}{\Phi}.\] 
    By \ref{item:large_size} and \ref{item:ogp}, we have 
    \begin{align*}
        a_k &\in \left[(1+\varepsilon)\left(\frac{1}{r-1}\right)^{\frac{1}{r-1}},\,\left(1+\varepsilon\right)\left(\frac{r}{r-1}\right)^{\frac{1}{r-1}}\right], \\
        b_{k} &\in \left[\frac{\varepsilon}{4}\left(\frac{1}{r-1}\right)^{\frac{1}{r-1}},\, \frac{\varepsilon}{2}\left(\frac{1}{r-1}\right)^{\frac{1}{r-1}}\right]
    \end{align*}

    Given the even event in conclusion of Theorem~\ref{theorem: hypergraph independence stat thresh} holds with probability $1-e^{-\Omega(n)}$, we have 
    \[
        c \leq a_1+\frac{\varepsilon(K-1)}{2}\left(\frac{1}{(r-1)}\right)^{\frac{1}{r-1}},
    \]
    with probability $1-e^{-\Omega(n)}$.
    In other words, $c$ is bounded from above by a constant $C(\varepsilon, K)$ that does not depend on $d$.

    First, note that there are at most $n^{2K}$ choices for $\{a_k\}$ and $\{b_k\}$.
    Once $\{a_k\}$ and $\{b_k\}$ are fixed, the number of choices for $\{S_k\}$ is at most
    \begin{align*}
        &~~~~~{n \choose a_1\Phi} \prod_{k=2}^{K} {n \choose b_k \Phi} {c\Phi \choose (a_k-b_k) \Phi} 
        \\
        &\leq 
        \left(\frac{en}{a_1\Phi}\right)^{a_1\Phi} \prod_{k=2}^{K} \left(\frac{en}{b_k\Phi}\right)^{b_k\Phi} \left(\frac{ec}{(a_k-b_k)}\right)^{(a_k-b_k)\Phi} 
        \\
        & = 
        \exp \left(a_1\Phi \log \left(\frac{e}{a_1}\left(\frac{d}{\log d}\right)^{\frac{1}{r-1}}\right)+\sum_{k=2}^K b_k\Phi\log \left(\frac{e}{b_k}\left(\frac{d}{\log d}\right)^{\frac{1}{r-1}}\right)+(a_k-b_k)\Phi \log \left(\frac{ec}{a_k - b_k} \right)\right) \\
        &= \exp \left(\frac{1}{r-1}\Phi\log d \left(a_1+\sum_{k=2}^K b_k + o_d(1)\right)\right)
    \end{align*}
    where we used the bounds on $a_k,\, b_k$, and $c$ computed earlier. 

    Now for a fixed $\{S_k\}$ satisfying \ref{item:large_size} and \ref{item:ogp}, we take union bound over all $(T+1)^K $ possible choices of $\{t_k\}$ and compute an upper bound on the probability that \ref{item:exist_indset} is satisfied. 
    Let $E$ be the number of potential edges $e$ for which there exists $k$ such that all vertices of $e$ lie inside of $S_k$. 
    We compute a lower bound on $E$ as follows:
    \begin{align*}
        E
        &\geq {a_1\Phi \choose r}+ \sum_{k=2}^K {a_k\Phi \choose r} - {(a_k-b_k)\Phi \choose r} \\
        &= 
        \frac{\Phi^r}{r!}\left(a_1^r+\sum_{k=2}^K a_k^r-(a_k-b_k)^r+o_d(1)\right)\\
        &= \frac{\Phi^r}{r!}\left( a_1^r+\sum_{k=2}^K b_k\Big((a_k-b_k)^{r-1}+a_k(a_k-b_k)^{r-2}+ \cdots + a_k^{r-1}\Big)+o_d(1)\right)\\
        &\geq   \frac{\Phi^r}{r!}\left( a_1^r+\sum_{k=2}^K rb_k(a_k-b_k)^{r-1}+o_d(1)\right)
    \end{align*}
    In the first line, the first term counts potential edges within $S_1$ and the $k$th term counts potential edges crossing $S_k$ and $\cup_{\ell<k}S_\ell$. Notice that equality holds when all the $t_k$ are equal, and otherwise the number of non-edges that need to occur is only larger.
    For fixed $\{S_k\}$ and $\{t_k\}$, \ref{item:exist_indset} occurs when at least $E$ independent non-edges occur in the sampling of $\{A^{(t)}\}$; this happens with probability at most 
    \[\left(1-\frac{d}{{n-1 \choose r-1}}\right)^{E} \leq \exp \left(- \frac{Ed}{{n-1 \choose r-1}}\right).\]
    We compute the following supremum which will be used to upper bound $\E[N]$, 
    \begin{align*}
        &~~~~~~~~
        \sup_{\{a_k\},\{b_k\}} \exp \left(\frac{1}{r-1}\Phi\log d \left(a_1+\sum_{k=2}^K b_k + o_d(1)\right)\right) \cdot \exp\left(-\frac{d}{{n-1 \choose r-1}}E\right)
        \\
        &\leq 
        \sup_{\{a_k\},\{b_k\}} \exp \left(\frac{1}{r-1}\Phi\log d \left(a_1+\sum_{k=2}^K b_k + o_d(1)\right)-\frac{1}{r}\Phi \log d \left(a_1^r+\sum_{k=2}^K rb_k(a_k-b_k)^{r-1}+o_d(1)\right)\right)
        \\
        &\leq 
        \sup_{\{a_k\},\{b_k\}} \exp \left(\Phi\log d \left(\frac{a_1}{r-1} -\frac{a_1^r}{r}+\sum_{k=2}^K \frac{b_k}{r-1}- {b_k}(a_k-b_k)^{r-1}+o_d(1)\right)\right)
        \\
        &= 
        \sup_{\{a_k\},\{b_k\}} \exp \left(\Phi\log d \left(\frac{a_1}{r-1} -\frac{a_1^r}{r}-\sum_{k=2}^K {b_k}\left((a_k-b_k)^{r-1}-\frac{1}{r-1}\right)+o_d(1)\right)\right)
    \end{align*}

    From a basic calculus argument, we may conclude the following:
    \[\sup_{a \in \R} \left(\frac{a}{r-1} - \frac{a^r}{r}\right) = \frac{1}{r}\left(\frac{1}{r-1}\right)^{\frac{1}{r-1}}.\]
    In addition, by the bounds on $a_k$ and $b_k$,
    we can bound the other term as follows:
    \begin{align*}
        &~~~~~~{b_k}\left((a_k-b_k)^{r-1}-\frac{1}{r-1}\right) \\
        &\geq
        \frac{\varepsilon}{4}\left(\frac{1}{r-1}\right)^{\frac{1}{r-1}}
        \left[\left((1+\varepsilon)\left(\frac{1}{r-1}\right)^{\frac{1}{r-1}}-\frac{\varepsilon}{2}\left(\frac{1}{r-1}\right)^{\frac{1}{r-1}}\right)^{r-1}-\frac{1}{r-1}\right]
        \\
        &= 
        \frac{\varepsilon}{4} \left(\frac{1}{r-1}\right)^{\frac{r}{r-1}}\left[\left(1+\frac{\varepsilon}{2}\right)^{r-1}-1\right]
        \\
        &\geq
        \frac{\varepsilon}{4} \left(\frac{1}{r-1}\right)^{\frac{r}{r-1}}\left[1+\left(\frac{\varepsilon}{2}\right)^{r-1}-1\right] 
        \\
        &\geq
        \frac{\varepsilon^r}{2^{r+1}} \left(\frac{1}{r-1}\right)^{\frac{r}{r-1}}
    \end{align*}
    Combing the above with the assumption that $K \geq \left \lceil \frac{2^{r+1}}{\varepsilon^r}\right \rceil+ 1$, we have
    \begin{align*}
        &~~~~~\frac{a_1}{r-1} -\frac{a_1^r}{r}-\sum_{k=2}^K {b_k}\left((a_k-b_k)^{r-1}-\frac{1}{r-1}\right)
        \\
        &\leq \frac{1}{r}\left(\frac{1}{r-1}\right)^{\frac{1}{r-1}} - \sum_{k=2}^K \frac{\varepsilon^r}{2^{r+1}} \left(\frac{1}{r-1}\right)^{\frac{r}{r-1}}
        \\
        &= 
        \frac{1}{r}\left(\frac{1}{r-1}\right)^{\frac{1}{r-1}} - (K-1)\frac{\varepsilon^r}{2^{r+1}} \left(\frac{1}{r-1}\right)^{\frac{r}{r-1}}
        \\
        &\leq
        -\frac{1}{r(r-1)}\left(\frac{1}{r-1}\right)^{\frac{1}{r-1}}.
    \end{align*}

    Put everything together, we have 
    \begin{align*}
        \E[N] &= n^{2K}(T+1)^K \sup_{\{a_k\},\{b_k\}} \exp \left(\frac{1}{r-1}\Phi\log d \left(a_1+\sum_{k=2}^K b_k + o_d(1)\right)\right) \cdot \exp\left(-\frac{d}{{n-1 \choose r-1}}E\right) 
        \\
        &\leq
        n^{2K}(T+1)^K \exp \left(\Phi\log d \left(-\frac{1}{r(r-1)}\left(\frac{1}{r-1}\right)^{\frac{1}{r-1}}+o_d(1)\right)\right) 
        \\
        &= \exp\left(-\Omega(n)\right),
    \end{align*}
    for sufficiently large $d$, as desired.
\end{proof}

Next, we show that no independent set in $\Hrnp$ has a large intersection with some fixed set of vertices.

\begin{lemma}\label{Lemma:no_large_inter_fix}
    Let $\varepsilon>0$, $r\geq 2$, and $a >0$ be constants. 
    Fix $S \subseteq [n]$ with $|S| \leq a\left(\frac{\log d}{d}\right)^{\frac{1}{r-1}}n$. If $d$ is sufficiently large in terms of $\varepsilon$, $r$, and $a$, then with probability $1-\exp(\Omega(n))$ there is no independent set $T$ in $\Hrnp$ such that $|T \cap S| \geq \varepsilon\left(\frac{\log d}{d}\right)^{\frac{1}{r-1}}n.$
\end{lemma}

\begin{proof}
    As usual, we let $\Phi \coloneqq \left(\frac{\log d}{d}\right)^{\frac{1}{r-1}}n$. Let $N$ be the number of subsets $S^{'}$ of $S$ such that $|S'| = \left\lceil \varepsilon \Phi\right\rceil :=b\Phi$ and $S'$ is an independent set in $\Hrnp$. It remains to show that $N=0$ with high probability.
    To this end, we have
    \begin{align*}
        \E[N] &= {|S| \choose b\Phi} \left(1-\frac{d}{{n-1 \choose r-1}}\right) ^{b\Phi \choose r} 
        \\
        &\leq
        \left(\frac{e a}{b}\right)^{b\Phi}
        \exp\left(-\frac{d}{{n-1 \choose r-1}}\cdot {b\Phi \choose r} \right)
        \\
        &\leq
        \exp\left(b\Phi \log \left(\frac{e a}{b}\right)- \Phi\log d \cdot \frac{b^{r}}{r} +O(1)\right)
        \\
        &= 
        \exp\left(\Phi\log d \left(- \frac{b^{r}}{r}+o_d(1)\right)\right)
        \\
        &= \exp\left(-\Omega(n)\right),
    \end{align*}
    for sufficiently large $d$. The result follows by Markov’s inequality.
\end{proof}

\subsection{Proof of the Intractability Result}

We are now ready to prove Theorem~\ref{theorem: low-deg hypergraph impossibility}.
We will show that no $(D, \Gamma, c)$-stable algorithm can find a large independent set in $\Hrnp$ in the sense of Definition~\ref{Def:f_Optimize_ind}.
As a result of Lemma~\ref{lemma: stable}, this would complete the proof of the impossibility result.

\begin{Proposition}\label{proposition: stability}
    For any $r\geq 2$ and $\eps > 0$, there exist $K,\, d_0 > 0$ such that for any $d \geq d_0$ there exists $n_0, \eta, C_1, C_2 > 0$ such that for any $n \geq n_0$, $\xi \geq 1$, $1 \leq D \leq \dfrac{C_1n}{\xi\log n}$, and $\delta \leq \exp\left(-C_2\xi\,D\log n\right)$, if 
    \[k \,\geq\, n\,(1+\eps)\left(\frac{1}{r-1}\cdot\frac{\log d}{d}\right)^{1/(r-1)},\]
    there is no $\left(D, K-1, \frac{\varepsilon}{32 \xi(1+\varepsilon)}\right)$-stable function that $(k, \delta, \xi, \eta)$-optimizes the independent set problem in $\Hrnp$ for $p = d/\binom{n-1}{r-1}$.
\end{Proposition}

\begin{proof}
    Fix $\epsilon > 0$, and let $K = \left \lceil \frac{2^{r+1}}{\varepsilon^r}\right \rceil+ 1$, $T=(K-1)m$, $c = \frac{\varepsilon}{32 \xi(1+\varepsilon)}$ and $\eta = \frac{\varepsilon}{16} \left(\frac{\log d}{(r-1)d}\right)^{\frac{1}{r-1}}$. Let $d_0 = d_0(\varepsilon)$ large enough to apply Theorem~\ref{theorem: hypergraph independence stat thresh}, Proposition~\ref{Prop:OGP}, and Lemma~\ref{Lemma:no_large_inter_fix}.
    For the convenience of this proof, we let $\Phi := \left(\frac{\log d}{(r-1)d}\right)^{\frac{1}{r-1}}n$, which differs from the previous definition only by a multiplicative constant. 
    
    Assume for a contradiction that such a degree-$D$ polynomial $f$ indeed exists. Sample the interpolation path $A^{(0)}, \ldots, A^{(T)}$ as in Definition~\ref{definition: interpolation path}, and let $V_{t} = V^{\eta}_f\left(A^{(t)}\right)$ be the resulting independent sets.
    We construct a sequence of sets $S_1, \ldots, S_K$ based on $\{V_t\}$ as follows: let $S_1 = V_0$, and for $k \geq 2$, let $S_k$ be the first $V_t$ such that $|V_t \setminus \cup_{l < k}S_l| \geq \frac{\varepsilon}{4} \left(\frac{\log d}{d(r-1)}\right)^{\frac{1}{r-1}} n$; if no such $t$ exists then the process fails. We define the following three events:
    \begin{enumerate}[label=(E\arabic*)]
        \item\label{event:proc_succ} $|V_t| \geq (1+\varepsilon)\Phi$ for all $t \in [T]$, and the process of constructing $S_1, \ldots, S_K$ succeeds.

        \item\label{event:no_c_bad} No edge of the interpolation path is $c$-bad for $f$.

        \item\label{event:ogp} The forbidden structure of Proposition \ref{Prop:OGP} does not exist.
    \end{enumerate}

    In the following claim, we show that if \ref{event:proc_succ} and \ref{event:no_c_bad} occur, then the forbidden structure will be produced, which implies that \ref{event:proc_succ}, \ref{event:no_c_bad}, and \ref{event:ogp} cannot occur concurrently.

    \begin{Claim} \label{Lemma:1,2->not 3}
        If \ref{event:proc_succ} and \ref{event:no_c_bad} both occur, then the  sequence of sets $S_1, \ldots, S_k$ satisfies the properties of the forbidden structure of Proposition~\ref{Prop:OGP}.
    \end{Claim}

    \begin{proof}
        First we show that $|V_t \triangle V_{t-1}| \leq \frac{\varepsilon}{4}  \left(\frac{1}{r-1}\cdot\frac{\log d}{d}\right)^{\frac{1}{r-1}}n = \frac{\varepsilon}{4}  \Phi $. From \ref{event:proc_succ} we know that the failure event setting $V^{\eta}_f = \varnothing$ does not occur for all $t \in [T]$. By Definition~\ref{Def:V_eta_f}, $i \in V_t \triangle V_{t-1}$ if one of two events occur:
        \begin{itemize}
            \item $i \in (I \setminus \Tilde{I}) \cup J$ for either $V^{\eta}_f\left(A^{(t)}\right)$ or $V^{\eta}_f\left(A^{(t-1)}\right)$, or
            \item one of $f_i\left(A^{(t)}\right)$ and $f_i\left(A^{(t-1)}\right)$ is $\geq 1$ and the other is $\leq \frac{1}{2}$. 
        \end{itemize}

        Notice that the first case occurs for at most $2\eta n$ coordinates as $|(I \setminus \Tilde{I}) \cup J| \leq \eta n$ in Definition~\ref{Def:V_eta_f}. 
        Hence the second case occurs for at least $|V_t \triangle V_{t-1}|-2\eta n$ coordinates $i \in [n]$ and for those coordinates $i$ we have $\left|f_i\left(A^{(t)}\right)-f_i\left(A^{(t-1)}\right)\right| \geq \frac{1}{2}$. Together with event \ref{event:no_c_bad}, this implies
        \begin{align*}
            \frac{1}{4}\big(|V_t \triangle V_{t-1}|-2\eta n\big) \leq \big\|f(A^{(t)})-f(A^{(t-1)})\big\|^2 \leq c \cdot \mathbb{E}_{A \sim \Hrnp}\big[\|f(A)\|^2\big].
        \end{align*}
        By assumption, the independent set problem is $(k, \delta, \xi, \eta\big)$-optimized by the function $f$, and so we have
        \begin{align*}
            \mathbb{E}_{A \sim \Hrnp}\big[\|f(A)\|^2\big] \leq \xi(1+\varepsilon)\Phi.
        \end{align*}
        Putting the above two inequalities together, we have
        \begin{align*}
            |V_t \triangle V_{t-1}| \leq 4c\xi(1+\varepsilon)\Phi+2\eta n 
            = 4 \cdot \frac{\varepsilon}{32 \xi (1+\varepsilon)} \cdot \xi (1+\varepsilon)\Phi+ 2 \cdot \frac{\varepsilon}{16}\Phi 
            = \frac{\varepsilon}{4}\Phi.
        \end{align*}
        Recall that $S_k$ is the first $V_t$ for which $|V_t \setminus \cup_{l \leq k}S_l| \geq \frac{\varepsilon}{4}\Phi$. In particular, we must have $|V_{t-1} \setminus \cup_{l \leq k}S_l| < \frac{\varepsilon}{4}\Phi$.
        With this observation in hand, we have
        \begin{align*}
            |S_k \setminus \cup_{l \leq k}S_l| 
            &= |V_t \setminus \cup_{l \leq k}S_l| 
            = |(V_t \cap V_{t-1}) \setminus \cup_{l \leq k}S_l|
            + |(V_t \setminus V_{t-1}) \setminus \cup_{l \leq k}S_l| \\
            &\leq |V_{t-1} \setminus \cup_{l \leq k}S_l| + |V_t \triangle V_{t-1}| \leq \frac{\varepsilon}{2}\Phi.
        \end{align*}
        Thus, $|S_k \setminus \cup_{l \leq k}S_l| \in [\frac{\varepsilon}{4}\Phi, \frac{\varepsilon}{2}\Phi]$. for all $2 \leq k \leq K$. By construction, $S_k$ is an independent set in $A^{(t)}$ for some $t$, and so we conclude that $S_1, \cdots, S_k$ satisfies the properties of the forbidden structure in Proposition~\ref{Prop:OGP}.
    \end{proof}

    Next, we will show that with positive probability the events \ref{event:proc_succ}, \ref{event:no_c_bad}, and \ref{event:ogp} occur simultaneously.
    This would complete the proof as by Claim~\ref{Lemma:1,2->not 3}, this implies a contradiction and so no such $f$ can exist.

    We first  bound  the probability that event \ref{event:proc_succ} occurs. Since the independent set problem is $(k, \delta, \xi, \eta)$-optimized by $f$, we have
    \[\mathbb{P}_{A\sim\Hrnp}\left[\left|V^\eta_f(A) \right| \geq (1+\varepsilon)\Phi\right] \geq 1-\delta.\]
    Combining with Theorem \ref{theorem: hypergraph independence stat thresh}, we have $(1+\varepsilon)\Phi \leq |V_t| \leq \left(r^{\frac{1}{r-1}}+\varepsilon\right)\Phi$ with probability at least $1-\delta-\exp(-\Omega(n))$. 
    Now suppose that for some $0 \leq T' \leq T-m$, we have sampled $A^{(0)}, \ldots, A^{(T')}$ and $S_1 = V_0, \ldots, S_{K'} = V_{t_{K'}}$ have been successfully selected for some $K' < K$ and $t_{K'} \leq T'$.
    As mentioned previously, $A^{(T'+m)}$ is independent from $A^{(0)}, \ldots, A^{(T')}$. 
    So provided $|S_k| \leq \left(r^{\frac{1}{r-1}}+\varepsilon\right)\Phi$ for $1 \leq k \leq K'$, we may apply Lemma~\ref{Lemma:no_large_inter_fix} with $S = \cup_{k \leq K'}S_k$ and $a =  \left(r^{\frac{1}{r-1}}+\varepsilon\right)\left(\frac{1}{r-1}\right)^{\frac{1}{r-1}} K'$.
    As a result, $|V_{T'+m} \cap (\cup_{k \leq K'}S_k) | \leq \varepsilon \Phi$ with probability $1-\exp(-\Omega(n))$.
    Therefore, we have
    \[
    |V_{T'+m} \setminus \cup_{k \leq K'}S_k| = |V_{T'+m}| - |V_{T'+m} \cap (\cup_{k \leq K'}S_k)| \geq \Phi \geq \frac{\varepsilon}{4}\Phi.
    \]
    This implies that we can find $S_{K'+1} =  V_t$ for some $t \leq T'+m$ and thus by induction the process succeeds by timestep $T = (K-1)m$. By a union bound over $t$, \ref{event:proc_succ} holds with probability at least $1-\delta(T+1)-\exp(-\Omega(n))$.

    By Definition~\ref{definition: (D, G, c)-stable} and the stability of $f$, \ref{event:no_c_bad} holds with probability at least $\left(d/{n-1 \choose r-1}\right)^{4(K-1)D/c}$. By Proposition \ref{Prop:OGP}, \ref{event:ogp} holds with probability $1-\exp(-\Omega(n))$. We claim that it suffices to have 
    \begin{align} \label{Ineq:C_1andC_2}
         \left(\frac{d}{{n-1 \choose r-1}}\right)^{4(K-1)D/c} \geq 2\exp(-Cn), \quad
        \text{ and } \quad
         \left(\frac{d}{{n-1 \choose r-1}}\right)^{4(K-1)D/c} \geq 2\delta Kn^r
    \end{align}
    for some constant $C=C(\varepsilon,d)$ to conclude that all three events happen simultaneously with non-zero probability, since \eqref{Ineq:C_1andC_2} implies
    \begin{align*}
        \left(\frac{d}{{n-1 \choose r-1}}\right)^{4(K-1)D/c} &\geq  \exp(-Cn)+\delta Kn^r \\
        &> \exp(-\Omega(n))+\delta Km \\
        &> \exp(-\Omega(n))+\delta(T+1).
    \end{align*}
    For $d>1$, the first inequality in \eqref{Ineq:C_1andC_2} is equivalent to $D \leq \frac{c(Cn -\log 2)}{4(K-1)\left(\log{n-1 \choose r-1}-\log d\right)}$
    which is implied by 
    \[D \leq \frac{c(Cn - \log 2)}{4(K-1)(r-1)\log n} = \frac{\varepsilon(Cn - \log 2)}{128(1+\varepsilon)(K-1)(r-1)\xi\log n}.\]
    And this is implied by $D < \frac{C_1n}{\xi \log n}$, for large enough $n$ and some constant $C_1 = C_1(\varepsilon, d)>0$. 
    The second inequality is implied by 
    \begin{align}\label{Ineq:delta}
        \delta \leq \exp \left(-\frac{4(K-1)(r-1)D}{c}\log n-r\log n-\log (2K)\right).
    \end{align}
    For large enough $n$, given $\xi,D \geq 1$, there exists another constant $C_2 = C_2(\varepsilon, d)>0$ such that (\ref{Ineq:delta}) is implied by $\delta \leq \exp\left(-C_2\xi D \log n\right)$.
\end{proof}

%% file: balanced.tex
\section{Balanced Independent Sets in $\Hrrnp$}\label{section: balanced}

In this section, we will consider the random $r$-uniform $r$-partite hypergraph $\Hrrnp$ for $p = d/n^{r-1}$ (see Definition~\ref{definition: models}).
We split this section into four subsections.
First, we prove the statistical threshold stated in Theorem~\ref{theorem: stat thresh balanced}.
In the second subsection, we describe a degree-$1$ algorithm in order to prove the tractability result in Theorem~\ref{theorem: low-deg thresh balanced}.
In the final two subsections, we prove the impossibility result of Theorem~\ref{theorem: low-deg thresh balanced} by first proving a version of the \textit{Overlap Gap Property} for balanced independent sets in $\Hrrnp$ and then applying this result to prove intractability of low-degree algorithms for balanced independent sets in the stated regime.

\subsection{Statistical Threshold}\label{subsection: IT bound}

In this section, we will prove Theorem~\ref{theorem: stat thresh balanced}.
Let $n,\, d,\, p,\, \boldgamma,\, H$ be as in the statement of the theorem.
To assist with our proof, we define the following parameters:
\[c_\gamma = \frac{1}{(r-1)\prod_i\gamma_i}, \quad \text{and} \quad f \coloneqq \left(\frac{c_{\boldgamma}\log d}{d}\right)^{1/(r-1)}.\]
Furthermore, we let $q \in \N$ be the smallest integer such that $q\gamma_i \in \N$ for each $i \in [r]$.
We will assume $d_0$ is sufficiently large in terms of $q$.
To prove the upper bound, we will show that
\begin{align}\label{eq: high prob ub for balanced}
    \P[\alpha_{\boldgamma}(H) \leq (1+\eps)nf] = 1 - o(1), \quad \text{as } n \to \infty,
\end{align}
for any $\eps > 0$.
To this end, we note the following by a union bound over all possible $\boldgamma$-balanced independent sets of size $(1+\eps)rnf$:
\begin{align*}
    &~\P[\alpha_{\boldgamma}(H) \geq (1+\eps)nf] \\
    &\leq \left(\prod_{i = 1}^r\binom{n}{\gamma_i\,(1+\eps)\,n\,f}\right)\left(1 - \frac{d}{n^{r-1}}\right)^{((1+\eps)\,n\,f)^r\prod_{i = 1}^r\gamma_i} \\
    &\leq \left(\prod_{i = 1}^r\left(\frac{en}{\gamma_i\,(1+\eps)\,n\,f}\right)^{\gamma_i\,(1+\eps)\,n\,f}\right)\,\exp\left(- (1+\eps)^r\,n\,f^rd\prod_{i = 1}^r\gamma_i\right) \\
    &\leq \exp\left((1+\eps)\,n\,f\left(1 - \log f - \sum_{i = 1}^r\gamma_i\log\gamma_i\right) - (1+\eps)^r\,n\,f^rd\prod_{i = 1}^r\gamma_i\right) \\
    &= \exp\left((1+\eps)\,n\,f\left(1 - \log f - \sum_{i = 1}^r\gamma_i\log\gamma_i - (1+\eps)^{r-1}c_{\boldgamma}\,\log d\prod_{i = 1}^r\gamma_i\right)\right).
\end{align*}
Note that $-\sum_{i = 1}^r\gamma_i\log\gamma_i$ subject to the constraint $\sum_i\gamma_i = 1$ is maximized at $\gamma_i = 1/r$ for all $i$.
With this in hand, plugging in the value for $c_{\boldgamma}$, we may conclude that
\begin{align*}
    \P[\alpha_{\boldgamma}(H) \geq (1+\eps)nf] &\leq \exp\left((1+\eps)\,n\,f\left(1 + \log r - \log f - \frac{(1+\eps)^{r-1}\log d}{r-1}\right)\right) \\
    &= \exp\left((1+\eps)\,n\,f\left(1 + \log r - \frac{\log c_{\boldgamma}}{r-1} - \frac{((1+\eps)^{r-1} - 1)\log d}{r-1} - \frac{\log\log d}{r-1}\right)\right) \\
    &\leq \exp\left((1+\eps)\,n\,f\left(1 + \log r - \frac{\log c_{\boldgamma}}{r-1} - \eps\log d - \frac{\log\log d}{r-1}\right)\right) \\
    &= \exp\left(-\Omega(n)\right),
\end{align*}
for $d$ large enough, completing the proof of \eqref{eq: high prob ub for balanced}.

Now, to prove the lower bound, we would like to show that
\begin{align}\label{eq: high prob lb for balanced}
    \P[\alpha_{\boldgamma}(H) \geq (1-\eps)nf] = 1 - o(1), \quad \text{as } n \to \infty.
\end{align}
To this end, we will apply the approach of \cite{perkins2024hardness} for bipartite graphs, which was inspired by the proof of \cite{frieze1990independence} for ordinary graphs.
For each $i \in [r]$, denote the vertices of $V_i$ as $\set{1_i, \ldots, n_i}$.
Let $m = \left(\dfrac{d}{(\log d)^r}\right)^{\frac{1}{r-1}}$, $n' = n/m$, and let $P_{i, j} = \set{((j - 1)m + 1)_i, \ldots, (jm)_i}$ for each $i \in [r]$ and $j \in [n']$.
We say $I$ is a $\boldgamma$-balanced $P$-independent set if $I$ is $\boldgamma$-balanced and $|I \cap P_{i, j}| \leq 1$ for each $i \in [r]$ and $j \in [n']$.
Let $\beta(H)$ denote the size of the largest $\boldgamma$-balanced $P$-independent set in a hypergraph $H \sim \mathcal{H}(r, n, p)$ and let $Z_k$ denote the number of such independent sets of size $k$.

Let us first show that $\beta(H)$ is concentrated about its mean.

\begin{lemma}\label{lemma: beta bounds}
    Let $H \sim \mathcal{H}(r, n, p)$ and let $\overline{\beta} = \E[\beta(H)]$. Then
    \[\P\left[|\beta(H) - \overline{\beta}| \geq \lambda\right] \leq 2\exp\left(-\frac{\lambda^2m\,q^2}{2rn}\right).\]
\end{lemma}

We will apply McDiarmid's inequality to prove this lemma.

\begin{Theorem}[McDiarmid's inequality]\label{lemma:mcd}
    Let $X = f(\vec{Z})$, where $\vec Z = (Z_1, \ldots, Z_t)$ and the $Z_i$ are independent random variables.
    Assume the function $f$ has the property that whenever $\vec z,\, \vec w$ differ in only one coordinate we have $|f(\vec z) - f(\vec w)| \leq c$. Then, for all $\lambda > 0$ we have
    \[\P[|X - \E[X]| \geq \lambda] \leq 2\exp\left(-\frac{\lambda^2}{2c^2t}\right).\]
\end{Theorem}

\begin{proof}[Proof of Lemma~\ref{lemma: beta bounds}]
    We will employ the so-called ``vertex-based'' approach toward applying Theorem~\ref{lemma:mcd} to problems in random hypergraphs.
    First, we partition the edges of $K_{r\times n}$ into $rn'$ groups $\set{E_{ij}\,:\, 1 \leq i \leq r,\, 1 \leq j \leq n'}$. 
    We include $(v_1, \ldots, v_r) \in E_{ij}$ if and only if the following hold for $1 \leq l \leq r$:
    \begin{itemize}
        \item $v_l \in P_{l,j_l}$ for some $j_l > j$ if $l < i$, and
        \item $v_l \in P_{l,j_l}$ for some $j_l \geq j$ if $l \geq i$.
    \end{itemize}
    In other words, an edge $e=(v_1, \ldots, v_r)$ is assigned to $E_{ij}$ where $j$ is the minimum ``partition number'' of a vertex $v \in e$, and $i$ is the minimum value $l$ such that $v_l$ is in $P_{l,j}$.

    For each $1 \leq i \leq r$ and $1 \leq j \leq n'$, we let $Z_{ij}$ denote the outcomes of the edges in $E_{ij}$.
    Clearly, the variables $Z_{ij}$ are independent.
    Furthermore, we claim that changing the outcomes of the edges in a single $E_{ij}$ can change the value $\beta(\cdot)$ by at most $q$.
    We will prove this claim by considering a few cases.
    To this end, let $\vec Z$ denote an arbitrary outcome of the $Z_{ij}$'s, let $H$ denote the hypergraph determined by these outcomes, and let $I$ be the $\boldgamma$-balanced $P$-independent set in $H$ satisfying $|I| = \beta(H)$.
    Let $\vec Z'$ be obtained from $\vec Z$ by changing the outcomes of the edges in $E_{ij}$ for some $i,j$, let $H'$ be the corresponding hypergraph and let $I'$ be the largest $\boldgamma$-balanced $P$-independent set in $H'$.
    \begin{enumerate}[label=(\textbf{Case \arabic*}),leftmargin = \leftmargin + 1\parindent, wide]
        \item $I \cap P_{i, j} = \emptyset$.
        Note that $I$ is still a $\boldgamma$-balanced $P$-independent set in $H'$ and so $|I'| \geq |I|$.
        If $I' \cap P_{i,j} = \emptyset$, then $I'$ is a $\boldgamma$-balanced $P$-independent set in $H$ and so $|I'| = |I|$.
        Suppose $|I' \cap P_{i,j}| = 1$.
        Then, there exists a $\boldgamma$-balanced $P$-independent set $J' \subseteq I'$ in $H'$ such that $|J'| \geq |I'| - q$ and $J' \cap P_{i, j} = \emptyset$.
        As $J'$ is also a $\boldgamma$-balanced $P$-independent set in $H$, it follows that $|J'| \leq |I|$.
        In particular, $|I'| \leq |I| + q$, as desired.

        \item $|I \cap P_{i,j}| = 1$.
        Note that there exists a $\boldgamma$-balanced $P$-independent set $J \subseteq I$ in $H$ such that $|J| \geq |I| - q$ and $J \cap P_{i, j} = \emptyset$.
        As $J$ is a $\boldgamma$-balanced $P$-independent set in $H'$, it follows that $|J| \leq |I'|$.
        In particular, $|I'| \geq |I| - q$.
        If $I' \cap P_{i,j} = \emptyset$, then $I'$ is a $\boldgamma$-balanced $P$-independent set in $H$, implying $|I'| \leq |I|$.
        If not, an identical argument to the previous case shows that $|I'| \leq |I| + q$, as desired.
        
    \end{enumerate}
    
    Therefore, we may apply Lemma~\ref{lemma:mcd} with $c = q$ and $t = rn'$ to get
    \[\P\left[|\beta(H) - \overline{\beta}| \geq \lambda\right] \leq 2\exp\left(-\frac{\lambda^2\,q^2}{2rn'}\right),\]
    as claimed.
\end{proof}

Next, we show there exists a large $\boldgamma$-balanced $P$-independent set in $H \sim \mathcal{H}(r, n, p)$ with positive probability.

\begin{lemma}\label{lemma: large P-ISET}
    Let $k = (1- \rho)nf$ for some $\rho > 0$ sufficiently small.
    Then,
    \[\P[Z_k > 0] \geq \exp\left(\frac{50n(c_\gamma\log d)^{\frac{1}{2(r-1)}}}{d^{\frac{1 - \rho/4}{r - 1}}}\right).\]
\end{lemma}

\begin{proof}
    We will use the inequality $\P[Z_k > 0] \geq \E[Z_k]^2 / \E[Z_k^2]$.
    To this end, let us compute the mean of $Z_k$.
    We have
    \begin{align*}
        \E[Z_k] = \left(\prod_{i = 1}^r\binom{n'}{\gamma_ik}m^{\gamma_ik}\right)\left(1 - \frac{d}{n^{r-1}}\right)^{k^r\prod_{i = 1}^r\gamma_i}
        = \left(\prod_{i = 1}^r\binom{n'}{\gamma_ik}\right)m^{k}\left(1 - \frac{d}{n^{r-1}}\right)^{k^r\prod_{i = 1}^r\gamma_i}.
    \end{align*}
    To compute $\E[Z_k^2]$, we define $\mathcal{I}_k$ to be the set of all possible $\boldgamma$-balanced $P$-independent sets of size $k$ in $\mathcal{H}(r, n, p)$.
    For brevity, we let $A_I$ denote the event that $I$ is independent in $H$, where $I \in \mathcal{I}_k$.
    With this in hand, we have
    \begin{align*}
        \E[Z_k^2] &= \sum_{I, I' \in \mathcal{I}_k}\P[A_I,\, A_{I'}] \\
        &= \sum_{I\in \mathcal{I}_k}\P[A_{I}]\sum_{I' \in \mathcal{I}_k} \P[A_{I'} \mid A_{I} ]\\
        &= \sum_{I\in \mathcal{I}_k}\P[A_{I}]\sum_{I' \in \mathcal{I}_k}\left(1 - \frac{d}{n^{r-1}}\right)^{k^r\prod_{i = 1}^r\gamma_i - \prod_{i = 1}^r|I\cap I' \cap V_i|} \\
        &= \sum_{I\in \mathcal{I}_k}\P[A_{I}]\sum_{\ell_1= 0}^{\gamma_1 k}\cdots\sum_{\ell_r = 0}^{\gamma_r k}\sum_{\substack{I' \in \mathcal{I}_k, \\ |I \cap I' \cap V_i| = \ell_i}}\left(1 - \frac{d}{n^{r-1}}\right)^{k^r\prod_{i = 1}^r\gamma_i - \prod_{i = 1}^r\ell_i} \\
        &\leq \sum_{I\in \mathcal{I}_k}\P[A_{I}]\sum_{\ell_1= 0}^{\gamma_1 k}\cdots\sum_{\ell_r = 0}^{\gamma_r k}\left(\prod_{i = 1}^r\binom{\gamma_i k}{\ell_i}\binom{n' - \ell_i}{\gamma_i k - \ell_i}m^{\gamma_ik - \ell_i}\right)\left(1 - \frac{d}{n^{r-1}}\right)^{k^r\prod_{i = 1}^r\gamma_i - \prod_{i = 1}^r\ell_i} \\
        &= \E[Z_k]\,m^k\sum_{\ell_1= 0}^{\gamma_1 k}\cdots\sum_{\ell_r = 0}^{\gamma_r k}m^{-\sum_{i = 1}^r\ell_i}\left(\prod_{i = 1}^r\binom{\gamma_i k}{\ell_i}\binom{n' - \ell_i}{\gamma_i k - \ell_i}\right)\left(1 - \frac{d}{n^{r-1}}\right)^{k^r\prod_{i = 1}^r\gamma_i - \prod_{i = 1}^r\ell_i}.
    \end{align*}
    From here, we may simplify further as follows:
    \begin{align*}
        \frac{\E[Z_k^2]}{\E[Z_k]^2} &\leq \sum_{\ell_1= 0}^{\gamma_1 k}\cdots\sum_{\ell_r = 0}^{\gamma_r k} \frac{m^{-\sum_{i = 1}^r\ell_i}\left(\prod_{i = 1}^r\binom{\gamma_i k}{\ell_i}\binom{n' - \ell_i}{\gamma_i k - \ell_i}\right)\left(1 - \frac{d}{n^{r-1}}\right)^{k^r\prod_{i = 1}^r\gamma_i - \prod_{i = 1}^r\ell_i}}{\left(\prod_{i = 1}^r\binom{n'}{\gamma_ik}\right)\left(1 - \frac{d}{n^{r-1}}\right)^{k^r\prod_{i = 1}^r\gamma_i}} \\
        &= \sum_{\ell_1= 0}^{\gamma_1 k}\cdots\sum_{\ell_r = 0}^{\gamma_r k} m^{-\sum_{i = 1}^r\ell_i}\left(1 - \frac{d}{n^{r-1}}\right)^{- \prod_{i = 1}^r\ell_i}\left(\prod_{i = 1}^r\frac{\binom{\gamma_i k}{\ell_i}\binom{n' - \ell_i}{\gamma_i k - \ell_i}}{\binom{n'}{\gamma_ik}}\right).
    \end{align*}
    Note the following for any $i \in [r]$:
    \[\frac{\binom{n' - \ell_i}{\gamma_ik - \ell_i}}{\binom{n'}{\gamma_i k}} = \frac{\frac{(n' - \ell_i)!}{(\gamma_ik - \ell_i)!(n' - \gamma_ik)!}}{\frac{(n')!}{(n' - \gamma_ik)!(\gamma_ik)!}} = \frac{(n' - \ell_i)!}{(n')!}\,\frac{(\gamma_ik)!}{(\gamma_ik - \ell_i)!} \leq \left(\frac{\gamma_i k}{n'}\right)^{\ell_i}.\]
    Plugging this in above, we get
    \begin{align*}
        \frac{\E[Z_k^2]}{\E[Z_k]^2} &\leq \sum_{\ell_1= 0}^{\gamma_1 k}\cdots\sum_{\ell_r = 0}^{\gamma_r k} \left(\frac{k}{n'm}\right)^{\sum_{i = 1}^r\ell_i}\left(1 - \frac{d}{n^{r-1}}\right)^{- \prod_{i = 1}^r\ell_i}\left(\prod_{i = 1}^r\binom{\gamma_i k}{\ell_i}\gamma_i^{\ell_i}\right) \\
        &\leq \sum_{\ell = 0}^k \left(\frac{k}{n}\right)^{\ell}\left(1 - \frac{d}{n^{r-1}}\right)^{-m(\ell)}\sum_{\substack{\ell_1, \ldots, \ell_r \\ \sum_{i = 1}^r\ell_i = \ell}}\prod_{i = 1}^r\binom{\gamma_i k}{\ell_i},
    \end{align*}
    where we use the fact that $\gamma_i \leq 1$, $n'm = n$ and define $m(\ell)$ to be the maximum value of $\prod_{i = 1}^r\ell_i$ subject to the constraints $\sum_{i = 1}^r\ell_i = \ell$ and $\ell_i \leq \gamma_i\,k$ for all $i$. 
    Note that since $\sum_{i = 1}^r\gamma_i = 1$, the final term is just $\binom{k}{\ell}$ and so, we have
    \[\frac{\E[Z_k^2]}{\E[Z_k]^2} \leq \sum_{\ell = 0}^k \binom{k}{\ell}\,\left(\frac{k}{n}\right)^{\ell}\left(1 - \frac{d}{n^{r-1}}\right)^{-m(\ell)}.\]
    For $n$ large enough in terms of $d$, we conclude the following for $\delta = 1 - \frac{(1 - \rho)^{r-1}}{1 - \rho/2} \approx r\rho/2$:
    \[\frac{\E[Z_k^2]}{\E[Z_k]^2} \leq \sum_{\ell = 0}^k \binom{k}{\ell}\,\left(\frac{k}{n}\right)^{\ell}\exp\left(\frac{d\,m(\ell)}{(1 - \delta)n^{r-1}}\right).\]
    Let $u_\ell$ denote the parameter inside the sum. 
    The goal now is to provide an upper bound on $\sum_{\ell = 0}^k u_\ell$.
    To this end, we define the following parameter:
    \[\beta \coloneqq \left(r^r\prod_{i = 1}^r\gamma_i\right)^{\frac{1}{r-1}}.\]
    Note that $\beta \leq 1$ as $\prod_{i = 1}^r\gamma_i$ is largest when $\gamma_i = 1/r$ for all $i$.
    For $A > 0$, the following fact will be useful for our computations and can be verified by a simple calculus argument:
    \begin{align}\label{eq: calculus hw}
        \max_{x > 0}\left(\frac{A}{x^s}\right)^x = \exp\left(\frac{sA^{1/s}}{e}\right).
    \end{align}
    We will split into cases depending on the value of $\ell$.
    \begin{enumerate}[label=(\textbf{Case \arabic*}),leftmargin = \leftmargin + 1\parindent, wide] 
        \item $0 \leq \ell \leq \beta\,k$. 
        When considering $m(\ell)$, by relaxing the constraints $\ell_i \leq \gamma_ik$ we may conclude that $m(\ell) \leq (\ell/r)^r$.        
        With this in hand, we have:
        \begin{align*}
            \exp\left(\frac{d\,m(\ell)}{(1 - \delta)n^{r-1}}\right) &\leq \exp\left(\frac{d\,\ell^r}{(1 - \delta)n^{r-1}r^r}\right) \\
            &\leq \exp\left(\frac{d\,(\beta\,k)^{r-1}\ell}{(1 - \delta)n^{r-1}r^r}\right) \\
            &= \exp\left(\left(\frac{1 - \rho/2}{r-1}\right)\ell\,\log d\right)
        \end{align*}
        In particular, applying \eqref{eq: calculus hw} for $s = 1$ we have
        \begin{align*}
            u_\ell \leq \left(\frac{k\,e}{\ell}\,\frac{k}{n}\,d^{\frac{1 - \rho/2}{r-1}}\right)^\ell \leq \exp\left(f\,k\,d^{\frac{1 - \rho/2}{r-1}}\right) \leq \exp\left(\frac{n\,\left(c_{\boldgamma}\log d\right)^{2/(r-1)}}{d^{\frac{1 + \rho/2}{r-1}}}\right).
        \end{align*}

        \item $\beta k < \ell \leq k$. Consider $\ell < k$. We have the following:
        \begin{align*}
            \frac{u_\ell}{u_{\ell + 1}} &= \frac{(\ell + 1)}{(k - \ell)}\,\frac{n}{k}\,\exp\left(-\frac{d}{n^{r-1}(1-\delta)}\left(m(\ell + 1) - m(\ell)\right)\right) \\
            &\leq \frac{kn}{(k - \ell)\ell}\,\exp\left(-\frac{d}{n^{r-1}(1-\delta)}\left(m(\ell + 1) - m(\ell)\right)\right).
        \end{align*}        
        From here, we may conclude that
        \begin{align*}
            u_\ell &\leq \frac{(kn)^{k - \ell}}{(k - \ell)!\,\ell(\ell + 1)\cdots (k - 1)}\exp\left(-\frac{d}{n^{r-1}(1-\delta)}\left(m(k) - m(\ell)\right)\right)u_k \\
            &= \frac{(\ell - 1)!(kn)^{k - \ell}}{(k - \ell)!(k - 1)!}\exp\left(-\frac{d}{n^{r-1}(1-\delta)}\left(m(k) - m(\ell)\right)\right)u_k.
        \end{align*}
        Note that
        \[\binom{k - 1}{\ell - 1} \geq 1 \implies \frac{(k - 1)!}{(\ell - 1)!} \geq (k - \ell)!.\]
        In particular, we have
        \begin{align*}
            u_\ell &\leq \frac{(kn)^{k - \ell}}{(k - \ell)!^2}\exp\left(-\frac{d}{n^{r-1}(1-\delta)}\left(m(k) - m(\ell)\right)\right)u_k \\
            &\leq \left(\frac{e^2kn}{(k - \ell)^2}\right)^{k - \ell}\exp\left(-\frac{d}{n^{r-1}(1-\delta)}\left(m(k) - m(\ell)\right)\right)u_k.
        \end{align*}
        Let $\ell_1, \ldots, \ell_r$ be the parameters that maximize $m(\ell)$.
        There exists $i \in [r]$ such that $\ell_i \leq \gamma_i\,\ell$ (if not, we would violate the constraint $\sum_{i = 1}^r\ell_i = \ell$).
        For $j \neq i$, we also have $\ell_j \leq \gamma_jk$.
        Furthermore, it is easy to see that $m(k) = k^r\prod_{i = 1}^r\gamma_i$.
        Therefore, we have
        \[m(k) - m(\ell) = k^r\prod_{i = 1}^r\gamma_i - \prod_{i = 1}^r\ell_i \geq k^r\prod_{i = 1}^r\gamma_i - k^{r-1}\ell\prod_{i = 1}^r\gamma_i = (k - \ell)k^{r-1}\prod_{i = 1}^r\gamma_i.\]
        In particular,
        \begin{align*}
            \exp\left(-\frac{d}{n^{r-1}(1-\delta)}\left(m(k) - m(\ell)\right)\right) &\leq \exp\left(-\frac{d}{n^{r-1}(1-\delta)}\left((k - \ell)k^{r-1}\prod_{i = 1}^r\gamma_i\right)\right) \\
            &= \exp\left(-\left(\frac{1 - \rho/2}{r - 1}\right)(k - \ell)\log d\right).
        \end{align*}
        Finally, we may bound $u_k$ as follows:
        \begin{align*}
            u_k &= \left(\frac{k}{n}\right)^k\exp\left(\frac{dk^r\prod_{i = 1}^r\gamma_i}{(1- \delta)n^{r-1}}\right) \\
            &\leq \left(f\exp\left(\left(\frac{1 - \rho/2}{r - 1}\right)\log d\right)\right)^k \\
            &= \left( \frac{(c_\gamma \log d)^{\frac{1}{r-1}}}{d^{\frac{\rho/2}{r-1}}}\right)^k \leq 1,
        \end{align*}
        where the last step follows for $d$ large enough.
        Putting together all the pieces, we may conclude the following by applying \eqref{eq: calculus hw} with $s = 2$:
        \begin{align*}
            u_\ell \leq \left(\frac{e^2kn}{(k - \ell)^2\,d^{\frac{1- \rho/2}{r - 1}}}\right)^{k - \ell} \leq \exp\left(2\sqrt{\frac{kn}{d^{\frac{1- \rho/2}{r - 1}}}}\right) = \exp\left(\frac{2n(c_\gamma\log d)^{\frac{1}{2(r-1)}}}{d^{\frac{1 - \rho/4}{r - 1}}}\right).
        \end{align*}
    \end{enumerate}

    Putting together both cases, we have:
    \begin{align*}
        \sum_{\ell = 0}^k u_\ell &\leq k\left(\exp\left(\frac{n\,\left(c_{\boldgamma}\log d\right)^{2/(r-1)}}{d^{\frac{1 + \rho/2}{r-1}}}\right) + \exp\left(\frac{2n(c_\gamma\log d)^{\frac{1}{2(r-1)}}}{d^{\frac{1 - \rho/4}{r - 1}}}\right)\right) \\
        &\leq \exp\left(\frac{50n(c_\gamma\log d)^{\frac{1}{2(r-1)}}}{d^{\frac{1 - \rho/4}{r - 1}}}\right),
    \end{align*}
    for $d$ large enough, completing the proof.
\end{proof}

We are now ready to prove \eqref{eq: high prob lb for balanced}.
Setting $\lambda = \lambda_0 \coloneqq \dfrac{\eps\,n}{20d^{1/(r-1)}}$ in Lemma~\ref{lemma: beta bounds} and comparing with Lemma~\ref{lemma: large P-ISET} for $\rho = \eps/5$ we see that $\overline{\beta} \geq k - \lambda_0$.
Applying Lemma~\ref{lemma: beta bounds} with $\lambda = \lambda_0$ again provides the desired result.

\subsection{Proof of the Achievability Result}\label{subsection: balanced achievability}

In this section, we will construct a degree-$1$ polynomial to prove the first part of Theorem~\ref{theorem: low-deg thresh balanced}.
In particular, we will construct a polynomial $f\,:\,\set{0, 1}^{n^r} \to \R^{rn}$ that $(k_1, \ldots, k_r, \delta, \xi, 0)$-optimizes the maximum $\boldgamma$-balanced independent set problem in $H$ for some $\xi \coloneqq \xi(d, r, \boldgamma) \geq 1$ and $\delta = \exp\left(-\Omega(n)\right)$.
As mentioned in \S\ref{subsection: proof overview}, the algorithm we describe is adapted from that of \cite{dhawan2023balanced} for deterministic hypergraphs.

Without loss of generality, assume
\[k_i = \gamma_i\,n\,(1-\eps)\left(\frac{\gamma_{i^\star}\log d}{d(r-1)\prod_{i = 1}^r\gamma_i}\right)^{1/(r-1)}, \quad \text{for each } 1 \leq i \leq r.\]
If some $k_i$ is smaller, we can zero out some coordinates to achieve this.
Furthermore, assume $\eps \leq 1/(100r)$ and $i^\star = r$ (recall that $i^{\star} = \arg\max_i \gamma_i$).

For each $1 \leq i < r$, fix $L_i \subseteq V_i$ such that $|L_i| = k_i$ and for vertices $v \in V_i$, let $f_v(A) = \1\set{v \in L_i}$.
In particular, $f_v$ is a degree-$0$ polynomial for each $v \in V_1\cup \cdots \cup V_{r-1}$.
For $v \in V_r$, we define
\[f_v(A) = 1 - \sum_{\substack{v_1, \ldots, v_{r-1}, \\ v_i \in L_i}}A_{(v_1, \ldots, v_{r-1}, v)},\]
which is a degree-$1$ polynomial.
Let $L_r = \set{v\in V_r\,:\, f(v, A) = 1}$.
By definition of $f_v(A)$, the set $L_1\cup \cdots \cup L_{r}$ is an independent set in $H$.
Additionally, if $f_v(A) \neq 1$, then $f_v(A) \leq 0$ and so we may set $\eta = 0$, as desired.

It remains to show the conditions of Definition~\ref{Def:f_Optimize_ind_bal} are satisfied for some $\xi \geq 1$ and $\delta = o_n(1)$.
To this end, consider a vertex $v \in V_r$.
For $n$ large enough and $\eps$ small enough, we have the following:
\begin{align*}
    \P[f_v(A) = 1] &= \left(1 - \frac{d}{n^{r-1}}\right)^{\prod_{i = 1}^{r-1}k_i} \\
    &\geq \exp\left(-\frac{(1-\eps)^{r-1}}{(1 - r\eps/2)}\cdot \frac{1}{r-1}\cdot \log d\right) \\
    &\geq d^{-\left(\frac{1 - \tilde \eps}{r - 1}\right)},
\end{align*}
where we let $\tilde \eps = r\eps/4$.
Let $\mu \coloneqq \E\left[\sum_{v \in V_r}f_v(A)\right]$.
In particular, we may conclude that
\[\frac{k_r}{\mu} \leq \gamma_r\,\left(\frac{\gamma_{i^\star}\log d}{d^{\tilde \eps}\,(r-1)\,\prod_{i = 1}^r\gamma_i}\right)^{1/(r-1)} \leq \frac{1}{2},\]
for $d$ large enough.
As the events $\set{f_v(A) = 1}$ and $\set{f_u(A) = 1}$ are independent for $u, v \in V_r$ and $u \neq v$, by a simple Chernoff bound, we may conclude that
\[\P\left[\sum_{v \in V_r}f_v(A) < k_r\right] \leq \P\left[\sum_{v \in V_r}f_v(A) \leq \mu/2\right] = \exp\left(-\Omega\left(n\right)\right),\]
as desired.
As $|V_f^0(A) \cap V_i|$ is defined deterministically for $1 \leq i < r$, this completes the proof of the second condition of Definition~\ref{Def:f_Optimize_ind_bal}.

Let us now consider the first condition.
We have
\begin{align*}
    \E[f_v(A)^2] &= \sum_{N = 1 - \prod_{i = 1}^{r-1}k_i}^1N^2\P[f_v(A) = N] \\
    &= \P[f_v(A) = 1] + \sum_{N = 2}^{\prod_{i = 1}^{r-1}k_i}(N-1)^2\P[\deg_{L_1\times \cdots \times L_{r-1}}(v) = N],
\end{align*}
where $\deg_{L_1\times \cdots \times L_{r-1}}(v)$ is the number of edges incident on $v$ containing vertices from $L_1\cup \cdots \cup L_{r-1}$ (apart from $v$ itself, of course).
From here, we may conclude that
\begin{align*}
    \E[f_v(A)^2] &\leq 1 + \sum_{N = 2}^\infty (N-1)^2\P[\deg_{L_1\times \cdots \times L_{r-1}}(v) = N] \\
    &\leq 1 + \sum_{N = 0}^\infty N^2\P[\deg_{L_1\times \cdots \times L_{r-1}}(v) = N] \\
    &= 1 + \E[\deg_{L_1\times \cdots \times L_{r-1}}^2(v)] \\
    &\leq 1 + \E[\deg_{H}^2(v)] \,=\, O(d^2).
\end{align*}

It follows that
\begin{align*}
    \E[\|f\|^2] \,=\, \sum_{i = 1}^r\sum_{v \in V_i}\E[f_v(A)^2] \,\leq\, O(nd^2) + \sum_{i = 1}^{r-1}k_i \,\leq\, \xi\sum_{i = 1}^rk_i,
\end{align*}
for some $\xi \coloneqq \xi(d,r, \boldgamma)$ large enough.
This proves the first condition in Definition~\ref{Def:f_Optimize_ind_bal} and therefore, completes the proof of the achievability result in Theorem~\ref{theorem: low-deg thresh balanced}.

\subsection{Forbidden Structures}\label{subsection: OGP balanced}

Here we define a sequence of correlated random graphs and a corresponding forbidden structure.
Recall that we represent an $r$-uniform $r$-partite hypergraph $H \sim \Hrrnp$ by its edge indicator vector $A = \set{0, 1}^m$ where $m = n^r$.
Additionally, we let $V_1, \ldots, V_r$ denote the partition of the vertex set $V$ of $H$.

Recall the interpolation path for $\Hrrnp$ from Definition~\ref{definition: interpolation path}.
It is easy to see that the marginal distribution of $A^{(t)}$ is $\Hrrnp$.
Furthermore, $A^{(t+m)}$ is independent of $A^{(0)}, \ldots, A^{(t)}$ as all edges are resampled by then.
We will now show that a certain structure of independent sets across the correlated sequence of random hypergraphs exists only with exponentially small probability.

\begin{Proposition}\label{proposition: interpolating ogp balanced}
    Fix $r\geq 2$, $\eps > 0$, and $\boldgamma = (\gamma_1, \ldots, \gamma_r)$ such that $\gamma_i \in (0, 1) \cap \mathbb{Q}$ and $\sum_{i = 1}^r\gamma_i = 1$.
    Additionally, let $i^\star = \arg\max_i \gamma_i$.
    Let $K \in \N$ be a sufficiently large constant in terms of $r$, $\eps$, and $\boldgamma$.
    Consider an interpolation path $A^{(0)}, \ldots, A^{(T)}$ of length $T = n^{O(r)}$.
    If $d \coloneqq d(\eps, r, \boldgamma, K)$ is sufficiently large, then with probability $1 - \exp\left(-\Omega(n)\right)$ there does not exist a sequence of subsets $S_1, \ldots, S_K$ of $V$ satisfying the following:
    \begin{enumerate}[label=(B\arabic*)]
        \item\label{item: exists ISET} for each $k \in [K]$, there exists $t_k \in [T]$ such that $S_k$ is a $\boldgamma$-balanced independent set in $A^{(t_k)}$,
        \item\label{item: S large balanced} for each $k \in [K]$ and $i \in [r]$, we have $|S_k \cap V_i| \geq (1+\eps)\,\gamma_i\,n\left(\frac{\gamma_{i^\star}\log d}{d\,(r-1)\prod_{i = 1}^r\gamma_i}\right)^{1/(r-1)}$.
        \item\label{item: big change} for $2 \leq k \leq K$ and $\tilde \eps = \eps\,\min_{i \in [r]}\gamma_i$, we have 
        \[|S_k \setminus\bigcup_{j < k}S_j| \in \left[\frac{\tilde\eps}{4r}\,n\left(\frac{\gamma_{i^\star}\log d}{d\,(r-1)\prod_{i = 1}^r\gamma_i}\right)^{1/(r-1)},\, \frac{\tilde\eps}{2r}\,n\left(\frac{\gamma_{i^\star}\log d}{d\,(r-1)\prod_{i = 1}^r\gamma_i}\right)^{1/(r-1)}\right].\]
    \end{enumerate}
\end{Proposition}

\begin{proof}
    Without loss of generality, let $i^\star = r$.
    Let $N$ denote the number of tuples $(S_1, \ldots, S_K)$ satisfying conditions \ref{item: exists ISET}--\ref{item: big change}.
    We will show that $\E[N]$ is exponentially small (the result then follows by Markov's inequality).
    To this end, let
    \[\Phi \coloneqq \frac{n}{r}\left(\frac{\log d}{d\,(r-1)\prod_{i = 1}^{r-1}\gamma_i}\right)^{1/(r-1)},\]
    and for each $i \in [r]$ and $k \in [K]$, we define:
    \[a_{i, k} = \frac{|S_k \bigcap V_i|}{\Phi}, \qquad b_{i, k} = \frac{|(S_k \setminus\bigcup_{j < k}S_j) \cap V_i|}{\Phi}, \qquad \text{and} \qquad b_k = \sum_{i = 1}^r b_{i, k}.\]
    Conditions \ref{item: S large balanced} and \ref{item: big change} imply that $a_{i, k} \geq (1+\eps)\gamma_ir$ and $b_k \in [\tilde \eps/4,\, \tilde \eps/2]$.
    Then for $c = |\cup_{k}S_k|/\Phi$, we have $c \leq \sum_{i}a_{i, 1}(or a_{i,1}) + (K - 1)\tilde \eps/2$.
    As a result of the upper bound in Theorem~\ref{theorem: stat thresh balanced}, we may assume that $\sum_{i}a_{i, 1} \leq (1+\eps)r\,\gamma_r^{-1/(r-1)}$.
    In particular, $c = C(r, \eps, K)$ is a bounded constant independent of $d$.
    There are at most $n^{2rK}$ choices for $\set{a_{i, k}}$ and $\set{b_{i, k}}$ ($\set{b_k}$ can be determined from $\set{b_{i, k}}$).
    Once these are fixed, the number of tuples $(S_1, \ldots, S_K)$ is at most
    \begin{align*}
        &~\left(\prod_{i = 1}^r\binom{n}{a_{i, 1}\Phi}\right)\prod_{k = 2}^K\prod_{i = 1}^r\binom{n}{b_{i, k}\Phi}\binom{c\Phi}{(a_{i, k} - b_{i, k})\Phi} \\
        &\leq \left(\prod_{i = 1}^r\left(\frac{en}{a_{i, 1}\Phi}\right)^{a_{i, 1}\Phi}\right)\prod_{k = 2}^K\prod_{i = 1}^r\left(\frac{en}{b_{i, k}\Phi}\right)^{b_{i, k}\Phi}\left(\frac{ec}{(a_{i, k} - b_{i, k})}\right)^{(a_{i, k} - b_{i, k})\Phi} \\
        &= \exp\left(\sum_{i = 1}^ra_{i, 1}\Phi\log \left(\frac{en}{a_{i, 1}\Phi}\right) + \sum_{k = 2}^K\sum_{i = 1}^r\left(b_{i, k}\Phi\log \left(\frac{en}{b_{i, k}\Phi}\right) + (a_{i, k} - b_{i, k})\Phi\log \left(\frac{ec}{(a_{i, k} - b_{i, k})}\right)\right)\right).
    \end{align*}
    Plugging in the value for $\Phi$ wherever $n/\Phi$ appears, the above is at most
    \begin{align*}
        \exp\left(\frac{\Phi\,\log d}{r - 1}\left(\sum_{i = 1}^ra_{i, 1} + \sum_{k = 2}^Kb_k + o_d(1)\right)\right).
    \end{align*}
    
    Now, for a fixed $\set{S_k}$, there are at most $(T+ 1)^k$ choices for $\set{t_k}$.
    Let $E$ be the edges in $K_{r\times n}$ for which there is some $k \in [K]$ such that $e \subseteq S_k$.
    For fixed $\set{S_k}$ and $\set{t_k}$, condition \ref{item: exists ISET} is satisfied if and only if a certain collection of at least $|E|$ independent non-edges occur in the sampling of $\set{A^{(t)}}$.
    This occurs with probability at most 
    \[\left(1 - \frac{d}{n^{r-1}}\right)^{|E|} \leq \exp\left(-\frac{|E|\,d}{n^{r-1}}\right).\]
    Furthermore, we have
    \begin{align*}
        |E| &\geq \Phi^r\prod_{i = 1}^ra_{i, 1} + \Phi^r\sum_{k = 2}^K\left(\prod_{i  =1}^ra_{i, k} - \prod_{i = 1}^r(a_{i, k} - b_{i, k})\right) \\
        &= \frac{\Phi\,r\gamma_r\,n^{r-1}\log d}{d(r-1)}\left(\prod_{i = 1}^r\frac{a_{i, 1}}{r\gamma_i} + \sum_{k = 2}^K\left(\prod_{i  =1}^r\frac{a_{i, k}}{r\gamma_i} - \prod_{i = 1}^r\frac{(a_{i, k} - b_{i, k})}{r\gamma_i}\right)\right).
    \end{align*}
    Let $\tilde a_{i,k} = \frac{a_{i, k}}{r\gamma_i}$ and $\tilde b_{i,k} = \frac{b_{i, k}}{r\gamma_i}$.
    Note that $\tilde a_{i, k} > 1$.
    It can be verified by a simple calculus argument that
    \[\sup_{x_i > 1}\,\sum_{i = 1}^rx_i - \prod_{i = 1}^rx_i \leq r - 1.\]
    In particular, we have
    \[\prod_{i = 1}^r\tilde a_{i,k} \,\geq\, \sum_{i = 1}^r\tilde a_{i,k} - r + 1.\]
    The following simple claim will assist with our proof.

    \begin{Claim}\label{claim: induction stuff}
        For $r \geq 2$ and $\eps > 0$, let $x, y \in \R^r$ be such that for each $i \in [r]$, we have $x_i \geq x_i - y_i \geq 1 + \eps$.
        Then, 
        \[\prod_{i = 1}^rx_i - \prod_{i = 1}^r(x_i - y_i) \geq (1+\eps)^{r-1}\sum_{j = 1}^ry_j.\]
    \end{Claim}

    \begin{proof}
        As a base case, let $r = 2$.
        We have
        \begin{align*}
            x_1x_2 - (x_1 - y_1)(x_2 - y_2) &= y_1(x_2 - y_2) + y_2x_1 \\
            &\geq (1+\eps)(y_1 + y_2),
        \end{align*}
        as desired.
        Assume the claim holds for $r < n$ and consider $r = n$.
        We have
        \begin{align*}
            \prod_{i = 1}^nx_i - \prod_{i = 1}^n(x_i - y_i) &= (x_1 - y_1)\prod_{i = 2}^nx_i - (x_1 - y_1)\prod_{i = 2}^n(x_i - y_i) + y_1\prod_{i = 2}^nx_i \\
            &= (x_1 - y_1)\left(\prod_{i = 2}^nx_i - \prod_{i = 2}^n(x_i - y_i)\right) + y_1\prod_{i = 2}^nx_i \\
            &\geq (1+\eps)\,(1+\eps)^{n-2}\sum_{i = 2}^ny_i + (1+\eps)^{n-1}\,y_1 \\
            &= (1+\eps)^{n-1}\sum_{i = 1}^ny_i,
        \end{align*}
        as desired.
    \end{proof}

    The goal is to apply the above claim with $x_i = \tilde a_{i,k}$ and $y_i = \tilde b_{i, k}$.
    Note that $\tilde b_{i, k} \leq \dfrac{b_k}{r\min_i\gamma_i} \leq \eps/(2r)$ and $\tilde a_{i, k} \geq 1+\varepsilon$. So we may apply Claim~\ref{claim: induction stuff} with $\eps/2$ to get
    \begin{align*}
        |E| &\geq \frac{\Phi\,r\gamma_r\,n^{r-1}\log d}{d(r-1)}\left(\sum_{i = 1}^r\tilde a_{i,k} - r + 1 + \sum_{k = 2}^K(1+\eps/2)^{r-1}\sum_{i = 1}^r\tilde b_{i,k}\right) \\
        &\geq \frac{\Phi\,n^{r-1}\log d}{d(r-1)}\left(\sum_{i = 1}^ra_{i,k} - r\gamma_r(r - 1) + \sum_{k = 2}^K(1+\eps/2)^{r-1}\sum_{i = 1}^rb_{i,k}\right)
    \end{align*}
    where we use the fact that $\gamma_r \geq \gamma_i$ for each $i$ by definition of $i^\star$.

    Putting everything together, we have
    \begin{align*}
        \E[N] &\leq n^{2rK}(T + 1)^k\sup_{\set{a_{i, k}}, \set{b_{i, k}}}\exp\left(\frac{\Phi\,\log d}{r - 1}\left(\sum_{i = 1}^ra_{i, 1} + \sum_{k = 2}^Kb_k + o_{d}(1)\right)\right) \\
        &\qquad \exp\left(-\frac{\Phi\log d}{r-1}\left(\sum_{i = 1}^ra_{i,k} - r\gamma_r(r - 1) + \sum_{k = 2}^K(1+\eps/2)^{r-1}b_{k}\right)\right) \\
        &\leq n^{2rK}(T + 1)^k\sup_{\set{a_{i, k}}, \set{b_{i, k}}} \exp\left(\frac{\Phi\log d}{r-1}\left(r\gamma_r(r - 1) - \frac{r\eps}{4}\sum_{k = 2}^Kb_k + o_{d}(1)\right)\right).
    \end{align*}
    As $b_k \geq \tilde \eps/4$, for $K$ sufficiently large in terms of $r$ and $\tilde{\eps}$, we have
    \[\E[N] \leq n^{2rK}(T + 1)^k\exp\left(-\frac{\Phi\,\log d}{r - 1}\right) = \exp\left(-\Omega(n)\right),\]
    as desired.
\end{proof}

Next, we show that no $\boldgamma$-balanced independent set in $\Hrrnp$ has a large intersection with some fixed set of vertices in all partitions.

\begin{lemma}\label{lemma: no large intersection balanced}
    Fix $\eps, s > 0$ and $\boldgamma = (\gamma_1, \ldots, \gamma_r)$ such that $\gamma_i \in (0, 1)$ and $\sum_{i = 1}^r\gamma_i = 1$. 
    Additionally, let $i^\star = \arg\max_i \gamma_i$.
    Let $d = d(\eps, s, \boldgamma) > 0$ be sufficiently large, and for each $i \in [r]$ fix any $S_i\subseteq V_i$ such that 
    \[|S_i| \,\leq\, s\,\gamma_i\,n\left(\frac{\gamma_{i^\star}\log d}{d\,(r-1)\prod_{i = 1}^r\gamma_i}\right)^{1/(r-1)}.\]
    Then for $p = d/n^{r-1}$ there is no $\boldgamma$-balanced independent set $S'$ in $\Hrrnp$ satisfying
    \[|S'\cap S_i| \,\geq\, \eps\,\gamma_i\,n\left(\frac{\gamma_{i^\star}\log d}{d\,(r-1)\prod_{i = 1}^r\gamma_i}\right)^{1/(r-1)}, \quad \text{for each } 1\leq i \leq r,\]
    with probability $1 - \exp\left(-\Omega(n)\right)$.
\end{lemma}

\begin{proof}
    Without loss of generality, assume $i^\star = r$.
    Let $\Phi$ be as defined in the proof of Proposition~\ref{proposition: interpolating ogp balanced}.
    The probability that such a set $S'$ exists is at most
    \begin{align*}
        \left(\prod_{i = 1}^r\binom{|S_i|}{\eps\gamma_ir\Phi}\right)\left(1 - \frac{d}{n^{r-1}}\right)^{\eps^r\,r^r\,\Phi^r\prod_{i = 1}^r\gamma_i} 
        &\leq \left(\prod_{i = 1}^r\binom{s\,\gamma_i\,r\Phi}{\eps\gamma_ir\Phi}\right)\exp\left(- \frac{d}{n^{r-1}}\,\eps^r\,r^r\,\Phi^r\prod_{i = 1}^r\gamma_i\right) \\
        &\leq \left(\prod_{i = 1}^r\left(\frac{es}{\eps}\right)^{\eps\gamma_ir\Phi}\right)\exp\left(-\frac{\eps^r\,\gamma_rr\,\Phi\log d}{r-1}\right) \\
        &= \exp\left(-r\Phi\left(\frac{\eps^r\,\gamma_r\,\log d}{r-1} - \eps\,\log \left(\frac{es}{\eps}\right)\right)\right) \\
        &= \exp\left(-\Omega(n)\right),
    \end{align*}
    as claimed.
\end{proof}

\subsection{Proof of the Intractability Result}\label{subsection: low-deg balanced}

We will show that no $(D, \Gamma, c)$-stable algorithm can find a large $\boldgamma$-balanced independent set in $\Hrrnp$ in the sense of Definition~\ref{Def:f_Optimize_ind_bal}.
As a result of Lemma~\ref{lemma: stable}, this would complete the proof of the impossibility result.

\begin{Proposition}\label{proposition: stability balanced}
    For any $r\geq 2$, $\eps > 0$, and $\boldgamma = (\gamma_1, \ldots, \gamma_r)$ such that $\gamma_i \in (0, 1) \cap \mathbb{Q}$, $\sum_{i = 1}^r\gamma_i = 1$, and $i^\star = \arg \max_i \gamma_i$, there exist $K,\, d_0 > 0$ such that for any $d \geq d_0$ there exists $n_0, \eta, C_1, C_2 > 0$ such that for any $n \geq n_0$, $\xi \geq 1$, $1 \leq D \leq \dfrac{C_1n}{\xi\log n}$, and $\delta \leq \exp\left(-C_2\xi\,D\log n\right)$, if 
    \[k_i \,\geq\, (1+\eps)\,\gamma_i\,n\left(\frac{\gamma_{i^\star}\log d}{d\,(r-1)\prod_{i = 1}^r\gamma_i}\right)^{1/(r-1)}, \quad \text{for each } 1\leq i \leq r,\]
    there is no $\left(D, K-1, \frac{\tilde \eps}{32r\xi(1 + \eps)}\right)$-stable function that $(k_1, \ldots, k_r, \delta, \xi, \eta)$-optimizes the $\boldgamma$-independent set problem in $\Hrrnp$ for $p = d/n^{r-1}$.
\end{Proposition}

\begin{proof}
    Fix $r$, $\eps$, and $\boldgamma$ as described in the statement of the proposition.
    Additionally, assume without loss of generality that $i^\star = r$.
    We define the following parameters:
    \[\Phi = \frac{1}{r}\left(\frac{\gamma_r\,\log d}{d\,(r-1)\prod_{i = 1}^r\gamma_i}\right)^{1/(r-1)}, \qquad \tilde \eps = \eps\min_{i \in [r]}\gamma_i.\]
    Let $K$ be sufficiently large in terms of $r, \eps, \boldgamma$ as described in Proposition~\ref{proposition: interpolating ogp balanced} and let 
    \[T = (K-1)m, \quad c = \dfrac{\tilde \eps}{32r\xi(1 + \eps)},\quad \text{and} \quad \eta = \dfrac{\tilde\eps\,\Phi}{20r}.\]
    Let $d_0$ and $n_0$ be large enough such that Theorem~\ref{theorem: stat thresh balanced}, Proposition~\ref{proposition: interpolating ogp balanced}, and Lemma~\ref{lemma: no large intersection balanced} can be applied.
    
    Suppose for contradiction there exists a $(D, K-1, c)$-stable algorithm $f$.
    Sample the interpolation path $A^{(0)}, \ldots, A^{(T)}$ as in Definition~\ref{definition: interpolation path} and let $I_t = V_f^\eta(A^{(t)})$ be the resulting independent sets for $0 \leq t \leq T$.
    We construct a sequence of sets $S_1, \ldots, S_K$ as follows: let $S_1 = I_0$, and for $k \geq 2$, let $S_k$ be the first $I_t$ such that $|I_t \setminus \bigcup_{j < k}S_j| \geq \tilde \eps\,n\,\Phi / 4$ (if no such $t$ exists, we say the process fails).
    We define the following events:
    \begin{enumerate}[label=(EB\arabic*), leftmargin=0.5in]
        \item\label{event: large sets and process termination} For each $i \in [r]$ and $t \in [T]$, we have $|I_t \cap V_i| \geq (1+\eps)\gamma_i\,r\,n\,\Phi$, and the process of constructing $S_1, \ldots, S_K$ succeeds.
        \item\label{event: no c bad balanced} No edge of the interpolation path is $c$-bad for $f$.
        \item\label{event: ogp balanced} The forbidden structure of Proposition~\ref{proposition: interpolating ogp balanced} does not exist.
    \end{enumerate}
    
    The following claim shows that if the events \ref{event: large sets and process termination} and \ref{event: no c bad balanced} occur, then \ref{event: ogp balanced} does not.

    \begin{Claim}\label{claim: ogp balanced}
        If the events \ref{event: large sets and process termination} and \ref{event: no c bad balanced} occur, then the sequence $S_1, \ldots, S_K$ constructed satisfy the properties forbidden by Proposition~\ref{proposition: interpolating ogp balanced}.
    \end{Claim}
    \begin{proof}
        Let us show that $|I_t\,\triangle\, I_{t-1}| \leq \tilde\eps\,n\,\Phi/4$.
        As a result of \ref{event: large sets and process termination}, we know that the failure event setting $V_f^\eta(A^{(t)}) = \varnothing$ does not occur.
        Therefore, if $i \in I_t\,\triangle\, I_{t-1}$, we must have one of the following:
        \begin{itemize}
            \item $i \in (I \setminus \tilde I) \cup J$ for either $V_f^\eta(A^{(t)})$ or $V_f^\eta(A^{(t-1)})$, or
            \item one of $f_i(A^{(t)})$ or $f_i(A^{(t-1)})$ is $\geq 1$ and the other is $\leq 1/2$.
        \end{itemize}
        By Definition~\ref{Def:V_eta_f}, the first case above occurs for at most $2\eta\,rn$ coordinates.
        It follows that for at least $|I_t\,\triangle\, I_{t-1}| - 2\eta\,rn$ coordinates $i$, we have $|f_i(A^{(t)}) - f_i(A^{(t-1)})| \geq 1/2$.
        Therefore, as a result of event \ref{event: no c bad balanced}, we may conclude the following:
        \[\frac{1}{4}(|I_t\,\triangle\, I_{t-1}| - 2\eta\,rn) \,\leq\, \|f(A^{(t)}) - f(A^{(t-1)})\|_2^2 \,\leq\, c\,\E[\|f(A)\|^2].\]
        As $f$ is assumed to $(k_1, \ldots, k_r, \delta, \xi, \eta)$-optimize the $\boldgamma$-independent set problem in $\Hrrnp$, we have
        \[\E[\|f(A)\|^2] \,\leq\, r\xi(1+\eps)n\Phi.\]
        Putting the above inequalities together, we have
        \[|I_t\,\triangle\, I_{t-1}| \,\leq\, 2\eta\,rn + 4\,c\,r\,\xi(1+\eps)n\Phi \,\leq\, \frac{\tilde \eps \,n\,\Phi}{4},\]
        as desired.
        Let $2 \leq k \leq K$ and let $t \in [T]$ be such that $S_k = I_t$.
        By definition of $S_k$, we have
        \[|I_t\setminus\bigcup_{j <k}S_j| \geq \frac{\tilde \eps \,n\,\Phi}{4}, \quad \text{and} \quad |I_{t-1}\setminus\bigcup_{j <k}S_j| \leq \frac{\tilde \eps \,n\,\Phi}{4}.\]
        From here, we may conclude that
        \begin{align*}
            |S_k\setminus\bigcup_{j <k}S_j| = |I_t\setminus\bigcup_{j <k}S_j| &= |(I_t\cap I_{t-1})\setminus\bigcup_{j <k}S_j| + |(I_t\setminus I_{t-1})\setminus\bigcup_{j <k}S_j| \\
            &\leq |I_{t-1}\setminus\bigcup_{j <k}S_j| +|I_t\,\triangle\, I_{t-1}| \\
            &\leq \frac{\tilde \eps \,n\,\Phi}{2}.
        \end{align*}
        It follows that the sequence of sets $S_1, \ldots, S_K$ satisfy the conditions in Proposition~\ref{proposition: interpolating ogp balanced}, i.e., such a forbidden structure exists as claimed.
    \end{proof}
    
    We will now show that with positive probability all three events \ref{event: large sets and process termination}, \ref{event: no c bad balanced}, and \ref{event: ogp balanced} occur.
    As this contradicts Claim~\ref{claim: ogp balanced}, this would imply no such $f$ can exist, completing the proof.

    By Theorem~\ref{theorem: stat thresh balanced} and Definition~\ref{Def:f_Optimize_ind_bal}, we may conclude that 
    \[(1+\eps)\gamma_i\,r\,n\,\Phi \,\leq\, |I_t \cap V_i| \,\leq\, (1+\eps)\,r\gamma_r^{-1/(r-1)}\,n\,\Phi, \qquad \text{for all } 0 \leq t \leq T,\,\, 1 \leq i \leq r,\]
    with probability at least
    \[1 - (T+1)(\delta + \exp\left(-\Omega(n)\right)).\]
    On the event that this occurs, we will show that the procedure to define $S_k$ succeeds with high probability.
    Suppose for some $0 \leq T' \leq T - m$, we have sampled $A^{(0)}, \ldots, A^{(T')}$ and successfully selected $S_1 = I_0, \ldots, S_{K'} = I_{t_{K'}}$ for some $1 \leq K' < K$.
    Let $S = \cup_{k \leq K'}S_k$ and note that $A^{(T' + m)}$ is independent of $A^{(T')}$.
    Applying Lemma~\ref{lemma: no large intersection balanced} with
    \[S_i = S\cap V_i, \quad \text{and} \quad s = \frac{(1+\eps)\,r\gamma_r^{-1/(r-1)}K'}{\min_i\gamma_i},\]
    we may conclude that
    \[\exists j \in [r], \quad |I_{T'+m} \cap S \cap V_j| \leq \eps\,r\,\gamma_j\,n\Phi,\]
    with probability at least $1 - \exp\left(-\Omega(n)\right)$.
    Fix such a $j$.
    We have:
    \begin{align*}
        |I_{T'+m} \setminus S| \geq |(I_{T'+m} \setminus S) \cap V_j| &= |I_{T'+m} \cap V_j| - |I_{T'+m} \cap S \cap V_j| \\
        &\geq \frac{\tilde \eps\,n\, \Phi}{4}.
    \end{align*}
    In particular, $S_{K' + 1} = I_t$ for some $T' < t \leq T' + m$ and thus the process succeeds by timestep $T = (K-1)m$ with high probability.
    We may conclude by a union bound over $t$ that \ref{event: large sets and process termination} occurs with probability at least $1 - (T+1)\delta - \exp\left(-\Omega(n)\right)$.

    Since $f$ is assumed to be $(D, K-1, c)$-stable, \ref{event: no c bad balanced} occurs with probability at least $p^{4(K-1)D/c}$, where $p = d/n^{r-1}$.
    By Proposition~\ref{proposition: interpolating ogp balanced}, we may conclude that at least 1 of the events \ref{event: large sets and process termination}--\ref{event: ogp balanced} occurs with probability at most
    \begin{align}\label{eq: prob bound impossibility}
        1 - \left(\frac{d}{n^{r-1}}\right)^{4(K-1)D/c} + (T+1)\delta + \exp\left(-\Omega(n)\right).
    \end{align}
    It is now enough to show that the above is $< 1$, which is implied by the following for some constant $C > 0$ sufficiently large:
    \[\left(\frac{d}{n^{r-1}}\right)^{4(K-1)D/c} \geq 2\exp\left(-Cn\right), \quad \text{and} \quad \left(\frac{d}{n^{r-1}}\right)^{4(K-1)D/c} \geq 2\delta Km.\]
    The first inequality follows by the upper bound on $D$ for $C_1 \coloneqq C_1(\eps, K, d) > 0$ sufficiently large, and the second follows by the upper bound on $\delta$ for $C_2 \coloneqq C_2(\eps, K, d) > 0$ sufficiently large.
    Plugging this in above, we have that \eqref{eq: prob bound impossibility} is at most
    \begin{align*}
        &~ 1 - \exp\left(-Cn\right) - \delta Km + (T+1)\delta + \exp\left(-\Omega(n)\right) \\
        &\leq 1 - \exp\left(-\Omega(n)\right) + \delta(T + 1 - Km) \\
        &= 1 - \exp\left(-\Omega(n)\right) - \delta(m - 1) < 1,
    \end{align*}
    for $C$ sufficiently large, completing the proof.
\end{proof}

%% file: Local_Tree_Like.tex
\section{Sparse Random Hypergraphs are Locally Hypertree-Like}\label{section: tree like}

In this section, we give the proof of our claim that $\Hrnp$ is locally tree-like for $p = d/\binom{n-1}{r-1}$ (this fact was used in the proof of Lemma~\ref{Lemma:local_tree_np}).
Our proof is nearly identical to that of \cite[Theorem~3.12]{bordenave2016lecture} for graphs.
There are two key lemmas in their proof that differ in the hypergraph setting (albeit slightly).
We include the proofs of these lemmas in this section.

We begin by showing that $\Hrnp$ ``looks like'' $\PGWLT$ locally.
The essence of ``looks like'' in the context of local weak convergence is to show that the $s$-neighborhood of a random vertex in $\Hrnp$ has a distribution that converges to the $s$-level subtree of $ \PGWLT$ as $ n \to \infty$. 
More formally, we prove the following proposition:

\begin{Proposition}\label{prop:ER_Local_convergence_PW_tree}
    Fix any integers $d, r, s >0$ and $\varepsilon >0$.
    Then for sufficiently large $n$ and $p = d/{n-1 \choose r-1}$, we have
    \[
        \P_{H \sim \Hrnp} \left[\left|\{u \in V(H)\,:\, B_{H}(u,s) \ncong \PGWLT\}\right| \geq \varepsilon n\right] \leq \varepsilon.
    \]
\end{Proposition}

Consider a hypergraph $H = (V, E)$ and let $v \in V$ be smallest vertex under some arbitrary order.
The proof follows the breadth-first search (BFS) of the connected component $H(v) \subseteq H$ containing $v$. 

This exploration is iterative.
At integer step $t$, a vertex may belong to the active set $A_t$, to the unexplored set $U_t$, or to the connected set $C_t = V \setminus (A_t \cup U_t)$. 
We begin with $A_0 = \{v\}$, $C_0 = \emptyset$, and $U_0 = V \setminus \{v\}$. For integer $t \geq 0$, if $A_t \neq \emptyset$, we define $v_{t+1}$ as the vertex in $A_t$ with the minimum order. 
Let $E_{t+1} = \{e \in E \,:\, v_{t+1} \in e,\, e \setminus \{v_{t+1}\} \subseteq U_{t}\}$ be the edges incident to $v_{t+1}$ that have not yet been explored in the BFS, and let $I_t = \{v \in U_t \,:\, v \in e \text{ for some } e \in E_{t+1}\}$ be the neighbors of $v_{t+1}$ in $U_t$. We set
\[
    \begin{cases}
        A_{t+1} = (A_t \cup I_{t+1}) \setminus \{v_{t+1}\} \\
        U_{t+1} = U_t \setminus I_{t+1} \\
        C_{t+1} = C_t \cup \{v_{t+1}\}
    \end{cases}
\]
The process stops when $A_t = \emptyset$. 
It follows by construction that $|H(v)| = \inf \{ t \geq 1 : A_t = \emptyset \}$. 

For ease of notation, we set $X_{t+1} = |E_{t+1}|$ and $\tau = |H(v)|$ so that for $t < \tau$, we have
\[
    |A_t| \leq 1+\sum_{s=1}^{t}((r-1)X_s-1), \qquad |U_{t}| = n-t-|A_t|, \qquad |C_t| = t.
\]
Now, we consider the BFS with starting vertex $v$ and the graph $H \sim \Hrnp$. 
We define the filtration $\mathcal{F}_t = \sigma((A_0, U_0, C_0), \dots, (A_t, U_t, C_t))$. 
The hitting time $\tau \coloneqq \inf \{ t \geq 1 : A_t = \emptyset \}$ is a stopping time for this filtration. 
Notice that for any integer $t \geq 0$, given $\mathcal{F}_t$, if $\{ t < \tau \} \in \mathcal{F}_t$ then $X_{t+1}\sim \bin\left({|U_t| \choose r-1}, d/{n-1 \choose r-1}\right)$.

With this set-up in hand, we first show the convergence of the exploration.

\begin{lemma}\label{lemma:convergence of the exploration}
    On an enlarged probability space, there exists a sequence $(X_t')_{t \geq 1}$ of i.i.d.\ $\poiss(d)$ variables such that
    \[
        \mathbb{P} \left[\left(X_1, \dots, X_{t \wedge \tau}\right) \neq \left(X_1', \dots, X_{t \wedge \tau}'\right)\right] \leq \frac{rd(1+rd)t^2+td}{n-1}.
    \]
\end{lemma}

\begin{proof}

    The stopping property implies that $\{ t < \tau \}$ is $\mathcal{F}_t$-measurable. 
    Note that if $\{ t < \tau \}$ holds, then
    \[\mathbb{E}(X_{t+1} | \mathcal{F}_t) = \frac{d {|U_{t}| \choose r-1}}{{n-1 \choose r-1}} \leq d.\]
    In particular, we have 
    \begin{align} \label{ineq:Expectation_A_t}
        \E\left[ (|A_t|-1)\mathbbm{1}_{t<\tau} \right]
        \leq 
        \sum_{s=0}^{t-1} \E[((r-1)X_{s+1} - 1)\mathbbm{1}_{s<\tau}]  \leq rdt.
    \end{align}
    Now, on an enlarged probability space, let $\xi_{t+1}$ be, given $\mathcal{F}_t$, a binomial variable $\bin\left({n-1 \choose r-1} - {|U_t| \choose r-1}, d/{n-1 \choose r-1}\right)$ independent of $X_{t+1}$. Then $Y_{t+1} = X_{t+1} + \xi_{t+1}$ is a binomial variable $\bin\left({n-1 \choose r-1}, d/{n-1 \choose r-1}\right)$ and $(Y_t)_{t \geq 1}$ is an i.i.d. sequence. 
    By a union bound, we have
    \[
        \mathbb{P} \left[\left(X_1, \dots, X_{t \wedge \tau}\right) \neq  \left(Y_1, \dots, Y_{t \wedge \tau}\right) \right] \leq \sum_{s=1}^t \mathbb{E} \left[ \mathbbm{1}_{s < \tau} \mathbb{P}\left[X_s \neq Y_s | \mathcal{F}_s\right] \right] 
        \leq \sum_{s=1}^t \mathbb{E} \left[ \mathbbm{1}_{s < \tau} \mathbb{P}\left[\xi_s \neq 0 | \mathcal{F}_s\right] \right].
    \]
    It is straightforward to compute the probability that $\xi_s \neq 0$:
    \begin{align*}
        \mathbb{P}\left[\xi_s \neq 0 | \mathcal{F}_s\right] &= 1 - \left(1 - \frac{d}{{n-1 \choose r-1}}\right)^{{n-1 \choose r-1} - {|U_s| \choose r-1}} 
        \\
        &\leq 
        \frac{d}{{n-1 \choose r-1}}\left({n-1 \choose r-1} - {|U_s| \choose r-1}\right) \\
        &=
        \frac{d}{{n-1 \choose r-1}}\left({n-1 \choose r-1} - {n-t-|A_s| \choose r-1}\right) ~~~~~~~~~~(\text{since }|U_s| = n-s-|A_s|) \\
        &\leq
        \frac{d}{{n-1 \choose r-1}} \cdot (s+|A_s|-1){n-2 \choose r-2} ~~~~~~~~~~\left(\text{using } {n \choose r}-{n-k \choose r} \leq k{n-1 \choose r-1}\right) \\
        &\leq \frac{rd(t+|A_t|-1)}{n-1},
    \end{align*}
    where we use the fact that $s \leq t$ and $A_s \subseteq A_t$.
    By \eqref{ineq:Expectation_A_t}, we have
    \begin{align}
        \mathbb{P} \left[ \left(X_1, \dots, X_{t \wedge \tau}\right) \neq  \left(Y_1, \dots, Y_{t \wedge \tau}\right) \right] 
        \leq
        \sum_{s=1}^{t} \frac{rd(t+rdt)}{n-1}
        \leq \frac{rd(1+rd)t^2}{n-1}. \label{eq: x and y close}
    \end{align}

    We use the following simple bound on the total variation distance between binomial and Poisson random variables:
    \[
        d_{TV}\left(\bin\left({n-1 \choose r-1}, d/{n-1 \choose r-1}\right), \poiss(d)\right)\leq  \frac{d}{{n-1 \choose r-1}},
    \]
    which implies
    \begin{align}\label{eq: tvd}
        d_{TV}(\left(Y_1, \dots, Y_{t}\right),\poiss(d)^{\otimes t}) \leq \frac{td}{{n-1 \choose r-1}} \leq \frac{td}{n-1},
    \end{align}
    for $r \geq 2$.
    Combining \eqref{eq: x and y close} and \eqref{eq: tvd}, the claim follows by the maximal coupling inequality. 
\end{proof}

\begin{lemma}\label{lemma:ER_local_tree_like_whp}
    For integer $t \geq 0$, let $J_t = \{(u_1, \cdots, u_{r-1}) \subset A_t: (v_{t+1}, u_1, \cdots, u_{r-1}) \in E\}$. We have that
    \[
        \P[\exists 1\leq s \leq t\wedge\tau: |J_s|\neq 0] \leq \frac{rd^2t^2}{n-1}
    \]
\end{lemma}
\begin{proof}
    
    Given $\mathcal{F}_t$, if $t < \tau$, $|J_t|$ is a binomial random variable $\bin\left({|A_t| - 1 \choose r-1}, d/{n-1 \choose r-1}\right)$. The union bound yields
    \[
        \mathbb{P} \left[ \exists 1 \leq s \leq t \wedge \tau : |J_s| \neq 0 \right] \leq \sum_{s=1}^t \mathbb{E} \left[ \mathbbm{1}_{s < \tau} \mathbb{P}[|J_s| \neq 0 | \mathcal{F}_s] \right].
    \]
    Since
    \[
        \mathbb{P}[|J_s| \neq 0 | \mathcal{F}_s] = 1 - \left(1 - \frac{d}{{n-1 \choose r-1}}\right)^{{|A_s| - 1 \choose r-1}} \leq \frac{d (|A_s| - 1)}{n-1},
    \]
    we have
    \[
        \mathbb{P} \left[ \exists 1 \leq s \leq t \wedge \tau : |J_s| \neq 0 \right] \leq 
        \frac{d}{n-1} \sum_{s=1}^t \mathbb{E}\left[ (|A_s| - 1) \mathbbm{1}_{s < \tau} \right].
    \]
   Using the bound (\ref{ineq:Expectation_A_t}) again gives
    \[
        \mathbb{P} \left[ \exists 1 \leq s \leq t \wedge \tau : |J_s| \neq 0 \right] \leq \frac{d}{n-1} \cdot rd t^2 = \frac{rd^2 t^2}{n-1},
    \]
    as desired.
\end{proof}

With these lemmas in hand, the proof of Proposition~\ref{prop:ER_Local_convergence_PW_tree} is identical to the graph case, \textit{mutatis mutandis}.
One can, for example, follow the strategy of Theorem~3.12 in the lecture notes of Bordenave \cite{bordenave2016lecture} by replacing Lemmas~3.13 and 3.14 with Lemmas~\ref{lemma:convergence of the exploration} and \ref{lemma:ER_local_tree_like_whp}, respectively.

%% file: hyperOGP.bib
@article{rahman2017local,
  title={Local algorithms for independent sets are half-optimal},
  author={Rahman, Mustazee and Vir{\'a}g, B{\'a}lint},
  journal={Annals of Probability},
  volume={45},
  number={3},
  pages={1543--1577},
  year={2017},
  publisher={Institute of Mathematical Statistics}
}

@article{krivelevich1998chromatic,
  title={The chromatic numbers of random hypergraphs},
  author={Krivelevich, Michael and Sudakov, Benny},
  journal={Random Structures \& Algorithms},
  volume={12},
  number={4},
  pages={381--403},
  year={1998},
  publisher={Wiley Online Library}
}

@article{dhawan2023balanced,
    author = {Dhawan, Abhishek},
    title = {Balanced Independent Sets and Colorings of Hypergraphs},
    journal = {Journal of Graph Theory},
    volume={109},
    number={1},
    pages={43--51},
    year={2025}
}

@inproceedings{perkins2024hardness,
  title={On the hardness of finding balanced independent sets in random bipartite graphs},
  author={Perkins, Will and Wang, Yuzhou},
  booktitle={Proceedings of the 2024 Annual ACM-SIAM Symposium on Discrete Algorithms (SODA)},
  pages={2376--2397},
  year={2024},
  organization={SIAM}
}

@book{karp2010reducibility,
  title={Reducibility among combinatorial problems},
  author={Karp, Richard M},
  year={2010},
  publisher={Springer}
}

@inproceedings{hastad1996clique,
  title={Clique is hard to approximate within n/sup 1-/spl epsiv},
  author={Hastad, Johan},
  booktitle={Proceedings of 37th Conference on Foundations of Computer Science},
  pages={627--636},
  year={1996},
  organization={IEEE}
}

@inproceedings{khot2001improved,
  title={Improved inapproximability results for maxclique, chromatic number and approximate graph coloring},
  author={Khot, Subhash},
  booktitle={Proceedings 42nd IEEE Symposium on Foundations of Computer Science},
  pages={600--609},
  year={2001},
  organization={IEEE}
}

@inproceedings{zuckerman2006linear,
  title={Linear degree extractors and the inapproximability of max clique and chromatic number},
  author={Zuckerman, David},
  booktitle={Proceedings of the thirty-eighth annual ACM symposium on Theory of computing},
  pages={681--690},
  year={2006}
}

@techreport{matula1976largest,
  title={The largest clique size in a random graph. Southern Methodist University},
  author={Matula, D},
  year={1976},
  institution={Tech. Report, CS 7608}
}

@misc{karp1976probabilistic,
  title={Probabilistic analysis of some combinatorial search problems. Traub, JF (ed.): Algorithms and complexity: New Directions and Recent Results},
  author={Karp, R},
  year={1976},
  publisher={Academic Press}
}

@article{frieze1990independence,
  title={On the independence number of random graphs},
  author={Frieze, Alan M},
  journal={Discrete Mathematics},
  volume={81},
  number={2},
  pages={171--175},
  year={1990},
  publisher={Elsevier}
}

@article{wein2022optimal,
  title={Optimal low-degree hardness of maximum independent set},
  author={Wein, Alexander S},
  journal={Mathematical Statistics and Learning},
  volume={4},
  number={3},
  pages={221--251},
  year={2022}
}

@article{axenovich2021bipartite,
  title={Bipartite independence number in graphs with bounded maximum degree},
  author={Axenovich, Maria and Sereni, Jean-S{\'e}bastien and Snyder, Richard and Weber, Lea},
  journal={SIAM Journal on Discrete Mathematics},
  volume={35},
  number={2},
  pages={1136--1148},
  year={2021},
  publisher={SIAM}
}

@article{chakraborti2023extremal,
  title={Extremal bipartite independence number and balanced coloring},
  author={Chakraborti, Debsoumya},
  journal={European Journal of Combinatorics},
  volume={113},
  pages={103750},
  year={2023},
  publisher={Elsevier}
}

@article{favaron1993bipartite,
  title={On the bipartite independence number of a balanced bipartite graph},
  author={Favaron, Odile and Mago, Pedro and Ordaz, Oscar},
  journal={Discrete mathematics},
  volume={121},
  number={1-3},
  pages={55--63},
  year={1993},
  publisher={Elsevier}
}

@article{jerrum1992large,
  title={Large cliques elude the Metropolis process},
  author={Jerrum, Mark},
  journal={Random Structures \& Algorithms},
  volume={3},
  number={4},
  pages={347--359},
  year={1992},
  publisher={Wiley Online Library}
}

@inproceedings{deshpande2015improved,
  title={Improved sum-of-squares lower bounds for hidden clique and hidden submatrix problems},
  author={Deshpande, Yash and Montanari, Andrea},
  booktitle={Conference on Learning Theory},
  pages={523--562},
  year={2015},
  organization={PMLR}
}

@inproceedings{meka2015sum,
  title={Sum-of-squares lower bounds for planted clique},
  author={Meka, Raghu and Potechin, Aaron and Wigderson, Avi},
  booktitle={Proceedings of the forty-seventh annual ACM symposium on Theory of computing},
  pages={87--96},
  year={2015}
}

@article{barak2019nearly,
  title={A nearly tight sum-of-squares lower bound for the planted clique problem},
  author={Barak, Boaz and Hopkins, Samuel and Kelner, Jonathan and Kothari, Pravesh K and Moitra, Ankur and Potechin, Aaron},
  journal={SIAM Journal on Computing},
  volume={48},
  number={2},
  pages={687--735},
  year={2019},
  publisher={SIAM}
}

@article{decelle2011asymptotic,
  title={Asymptotic analysis of the stochastic block model for modular networks and its algorithmic applications},
  author={Decelle, Aurelien and Krzakala, Florent and Moore, Cristopher and Zdeborov{\'a}, Lenka},
  journal={Physical review E},
  volume={84},
  number={6},
  pages={066106},
  year={2011},
  publisher={APS}
}

@article{arias2014community,
  title={Community detection in dense random networks},
  author={Arias-Castro, Ery and Verzelen, Nicolas},
  journal={The Annals of Statistics},
  pages={940--969},
  year={2014},
  publisher={JSTOR}
}

@unpublished{hopkins2017bayesian,
  title={Bayesian estimation from few samples: community detection and related problems},
  author={Hopkins, Samuel B and Steurer, David},
  howpublished={\url{https://arxiv.org/abs/1710.00264}},
  year={2017}
}

@unpublished{berthet2013computational,
  title={Computational lower bounds for sparse PCA},
  author={Berthet, Quentin and Rigollet, Philippe},
  howpublished={\url{https://arxiv.org/abs/1304.0828}},
  year={2013}
}

@inproceedings{lesieur2015phase,
  title={Phase transitions in sparse PCA},
  author={Lesieur, Thibault and Krzakala, Florent and Zdeborov{\'a}, Lenka},
  booktitle={2015 IEEE International Symposium on Information Theory (ISIT)},
  pages={1635--1639},
  year={2015},
  organization={IEEE}
}

@inproceedings{hopkins2015tensor,
  title={Tensor principal component analysis via sum-of-square proofs},
  author={Hopkins, Samuel B and Shi, Jonathan and Steurer, David},
  booktitle={Conference on Learning Theory},
  pages={956--1006},
  year={2015},
  organization={PMLR}
}

@inproceedings{hopkins2017power,
  title={The power of sum-of-squares for detecting hidden structures},
  author={Hopkins, Samuel B and Kothari, Pravesh K and Potechin, Aaron and Raghavendra, Prasad and Schramm, Tselil and Steurer, David},
  booktitle={2017 IEEE 58th Annual Symposium on Foundations of Computer Science (FOCS)},
  pages={720--731},
  year={2017},
  organization={IEEE}
}

@inproceedings{achlioptas2008algorithmic,
  title={Algorithmic barriers from phase transitions},
  author={Achlioptas, Dimitris and Coja-Oghlan, Amin},
  booktitle={2008 49th Annual IEEE Symposium on Foundations of Computer Science},
  pages={793--802},
  year={2008},
  organization={IEEE}
}

@inproceedings{kothari2017sum,
  title={Sum of squares lower bounds for refuting any CSP},
  author={Kothari, Pravesh K and Mori, Ryuhei and O'Donnell, Ryan and Witmer, David},
  booktitle={Proceedings of the 49th Annual ACM SIGACT Symposium on Theory of Computing},
  pages={132--145},
  year={2017}
}

@article{dyer2002counting,
  title={On counting independent sets in sparse graphs},
  author={Dyer, Martin and Frieze, Alan and Jerrum, Mark},
  journal={SIAM Journal on Computing},
  volume={31},
  number={5},
  pages={1527--1541},
  year={2002},
  publisher={SIAM}
}

@inproceedings{gamarnik2014limits,
  title={Limits of local algorithms over sparse random graphs},
  author={Gamarnik, David and Sudan, Madhu},
  booktitle={Proceedings of the 5th conference on Innovations in theoretical computer science},
  pages={369--376},
  year={2014}
}

@article{chen2019suboptimality,
  title={Suboptimality of local algorithms for a class of max-cut problems},
  author={Chen, Wei Kuo and Gamarnik, David and Panchenko, Dmitry and Rahman, Mustazee},
  journal={Annals of Probability},
  volume={47},
  number={3},
  pages={1587--1618},
  year={2019},
  publisher={Institute of Mathematical Statistics}
}

@inproceedings{brennan2019optimal,
  title={Optimal average-case reductions to sparse pca: From weak assumptions to strong hardness},
  author={Brennan, Matthew and Bresler, Guy},
  booktitle={Conference on Learning Theory},
  pages={469--470},
  year={2019},
  organization={PMLR}
}

@inproceedings{hajek2015computational,
  title={Computational lower bounds for community detection on random graphs},
  author={Hajek, Bruce and Wu, Yihong and Xu, Jiaming},
  booktitle={Conference on Learning Theory},
  pages={899--928},
  year={2015},
  organization={PMLR}
}

@article{kearns1998efficient,
  title={Efficient noise-tolerant learning from statistical queries},
  author={Kearns, Michael},
  journal={Journal of the ACM (JACM)},
  volume={45},
  number={6},
  pages={983--1006},
  year={1998},
  publisher={ACM New York, NY, USA}
}

@inproceedings{feldman2015complexity,
  title={On the complexity of random satisfiability problems with planted solutions},
  author={Feldman, Vitaly and Perkins, Will and Vempala, Santosh},
  booktitle={Proceedings of the forty-seventh annual ACM symposium on Theory of Computing},
  pages={77--86},
  year={2015}
}

@article{dhawan2023detection,
  title={Detection of Dense Subhypergraphs by Low-Degree Polynomials},
  author={Dhawan, Abhishek and Mao, Cheng and Wein, Alexander S},
  journal={Random Structures \& Algorithms},
  volume={66},
  number={1},
  pages={e21279},
  year={2025},
  publisher={Wiley Online Library}
}

@inproceedings{bhaskara2010detecting,
  title={Detecting high log-densities: an O (n $1/4$) approximation for densest k-subgraph},
  author={Bhaskara, Aditya and Charikar, Moses and Chlamtac, Eden and Feige, Uriel and Vijayaraghavan, Aravindan},
  booktitle={Proceedings of the forty-second ACM symposium on Theory of computing},
  pages={201--210},
  year={2010}
}

@inproceedings{feige1997densest,
  title={On the Densest K-subgraph Problem},
  author={Uriel Feige and Michael Seltser},
  year={1997}
}

@inproceedings{jones2023sum,
  title={Sum-of-squares lower bounds for densest k-subgraph},
  author={Jones, Chris and Potechin, Aaron and Rajendran, Goutham and Xu, Jeff},
  booktitle={Proceedings of the 55th Annual ACM Symposium on Theory of Computing},
  pages={84--95},
  year={2023}
}

@article{PDC_IT,
  title={Information-theoretic thresholds for planted dense cycles},
  author={Mao, Cheng and Wein, Alexander S and Zhang, Shenduo},
  journal={IEEE Transactions on Information Theory},
  year={2024},
  publisher={IEEE}
}

@inproceedings{PDC,
  title={Detection-recovery gap for planted dense cycles},
  author={Mao, Cheng and Wein, Alexander S and Zhang, Shenduo},
  booktitle={The Thirty Sixth Annual Conference on Learning Theory},
  pages={2440--2481},
  year={2023},
  organization={PMLR}
}

@inproceedings{corinzia2022statistical,
  title={Statistical and computational thresholds for the planted k-densest sub-hypergraph problem},
  author={Corinzia, Luca and Penna, Paolo and Szpankowski, Wojciech and Buhmann, Joachim},
  booktitle={International Conference on Artificial Intelligence and Statistics},
  pages={11615--11640},
  year={2022},
  organization={PMLR}
}

@article{chlamtac2018densest,
  title={The densest k-subhypergraph problem},
  author={Chlamt{\'a}c, Eden and Dinitz, Michael and Konrad, Christian and Kortsarz, Guy and Rabanca, George},
  journal={SIAM Journal on Discrete Mathematics},
  volume={32},
  number={2},
  pages={1458--1477},
  year={2018},
  publisher={SIAM}
}

@inproceedings{gamarnik2020low,
  title={Low-degree hardness of random optimization problems},
  author={Gamarnik, David and Jagannath, Aukosh and Wein, Alexander S},
  booktitle={2020 IEEE 61st Annual Symposium on Foundations of Computer Science (FOCS)},
  pages={131--140},
  year={2020},
  organization={IEEE}
}

@article{gamarnik2021overlap,
  title={The overlap gap property: A topological barrier to optimizing over random structures},
  author={Gamarnik, David},
  journal={Proceedings of the National Academy of Sciences},
  volume={118},
  number={41},
  pages={e2108492118},
  year={2021},
  publisher={National Acad Sciences}
}

@inproceedings{bresler2022algorithmic,
  title={The algorithmic phase transition of random k-sat for low degree polynomials},
  author={Bresler, Guy and Huang, Brice},
  booktitle={2021 IEEE 62nd Annual Symposium on Foundations of Computer Science (FOCS)},
  pages={298--309},
  year={2022},
  organization={IEEE}
}

@article{chen2023local,
  title={Local algorithms and the failure of log-depth quantum advantage on sparse random CSPs},
  author={Chen, Antares and Huang, Neng and Marwaha, Kunal},
  journal={arXiv preprint arXiv:2310.01563},
  year={2023}
}

@unpublished{venkat2022efficient,
  title={Efficient algorithms for certifying lower bounds on the discrepancy of random matrices},
  author={Venkat, Prayaag},
  howpublished={\url{https://arxiv.org/abs/2211.07503}},
  year={2022}
}

@article{gamarnik2024hardness,
  title={Hardness of Random Optimization Problems for Boolean Circuits, Low-Degree Polynomials, and Langevin Dynamics},
  author={Gamarnik, David and Jagannath, Aukosh and Wein, Alexander S},
  journal={SIAM Journal on Computing},
  volume={53},
  number={1},
  pages={1--46},
  year={2024},
  publisher={SIAM}
}

@article{halldorsson2009independent,
  title={Independent sets in bounded-degree hypergraphs},
  author={Halld{\'o}rsson, Magn{\'u}s M and Losievskaja, Elena},
  journal={Discrete applied mathematics},
  volume={157},
  number={8},
  pages={1773--1786},
  year={2009},
  publisher={Elsevier}
}

@article{halldorsson2016streaming,
  title={Streaming algorithms for independent sets in sparse hypergraphs},
  author={Halld{\'o}rsson, Bjarni V and Halld{\'o}rsson, Magn{\'u}s M and Losievskaja, Elena and Szegedy, Mario},
  journal={Algorithmica},
  volume={76},
  pages={490--501},
  year={2016},
  publisher={Springer}
}

@inproceedings{guruswami2011complexity,
  title={The complexity of finding independent sets in bounded degree (hyper) graphs of low chromatic number},
  author={Guruswami, Venkatesan and Sinop, Ali Kemal},
  booktitle={Proceedings of the Twenty-Second Annual ACM-SIAM Symposium on Discrete Algorithms},
  pages={1615--1626},
  year={2011},
  organization={SIAM}
}

@inproceedings{khanna2021independent,
  title={Independent sets in semi-random hypergraphs},
  author={Khanna, Yash and Louis, Anand and Paul, Rameesh},
  booktitle={Algorithms and Data Structures: 17th International Symposium, WADS 2021, Virtual Event, August 9--11, 2021, Proceedings 17},
  pages={528--542},
  year={2021},
  organization={Springer}
}

@article{agnarsson2013sdp,
  title={SDP-based algorithms for maximum independent set problems on hypergraphs},
  author={Agnarsson, Geir and Halld{\'o}rsson, Magn{\'u}s M and Losievskaja, Elena},
  journal={Theoretical Computer Science},
  volume={470},
  pages={1--9},
  year={2013},
  publisher={Elsevier}
}

@article{halperin2002improved,
  title={Improved approximation algorithms for the vertex cover problem in graphs and hypergraphs},
  author={Halperin, Eran},
  journal={SIAM Journal on Computing},
  volume={31},
  number={5},
  pages={1608--1623},
  year={2002},
  publisher={SIAM}
}

@article{luo2022tensor,
  title={Tensor clustering with planted structures: Statistical optimality and computational limits},
  author={Luo, Yuetian and Zhang, Anru R},
  journal={The Annals of Statistics},
  volume={50},
  number={1},
  pages={584--613},
  year={2022},
  publisher={Institute of Mathematical Statistics}
}

@article{yuan2021heterogeneous,
  title={Heterogeneous dense subhypergraph detection},
  author={Yuan, Mingao and Shang, Zuofeng},
  journal={Statistica Neerlandica},
  volume={78},
  number={4},
  pages={759--775},
  year={2024},
  publisher={Wiley Online Library}
}

@article{yuan2021information,
  title={Information limits for detecting a subhypergraph},
  author={Yuan, Mingao and Shang, Zuofeng},
  journal={Stat},
  volume={10},
  number={1},
  pages={e407},
  year={2021},
  publisher={Wiley Online Library}
}

@article{guruswami2015inapproximability,
  title={Inapproximability of minimum vertex cover on k-uniform k-partite hypergraphs},
  author={Guruswami, Venkatesan and Sachdeva, Sushant and Saket, Rishi},
  journal={SIAM Journal on Discrete Mathematics},
  volume={29},
  number={1},
  pages={36--58},
  year={2015},
  publisher={SIAM}
}

@article{botelho2012cores,
  title={Cores of random r-partite hypergraphs},
  author={Botelho, Fabiano C and Wormald, Nicholas and Ziviani, Nivio},
  journal={Information Processing Letters},
  volume={112},
  number={8-9},
  pages={314--319},
  year={2012},
  publisher={Elsevier}
}

@unpublished{gamarnik2014performance,
  title={Performance of the survey propagation-guided decimation algorithm for the random NAE-K-SAT problem},
  author={Gamarnik, David and Sudan, Madhu},
  howpublished={\url{https://arxiv.org/abs/1402.0052}},
  year={2014}
}

@article{coja2017walksat,
  title={Walksat stalls well below satisfiability},
  author={Coja-Oghlan, Amin and Haqshenas, Amir and Hetterich, Samuel},
  journal={SIAM Journal on Discrete Mathematics},
  volume={31},
  number={2},
  pages={1160--1173},
  year={2017},
  publisher={SIAM}
}

@article{dhawan2023list,
  title={List colorings of $ k $-partite $ k $-graphs},
  author={Dhawan, Abhishek},
  journal = {Electronic Journal of Combinatorics},
  volume={32},
  number={2},
  pages={2.16},
  year={2025}
}

@article{bowtell2024matchings,
  title={Matchings in multipartite hypergraphs},
  author={Bowtell, Candida and Mycroft, Richard},
  howpublished={\url{https://arxiv.org/abs/2403.05219}},
  year={2024}
}

@article{kamvcev2017bounded,
  title={Bounded colorings of multipartite graphs and hypergraphs},
  author={Kam{\v{c}}ev, Nina and Sudakov, Benny and Volec, Jan},
  journal={European journal of combinatorics},
  volume={66},
  pages={235--249},
  year={2017},
  publisher={Elsevier}
}

@article{du2023algorithmic,
  title={The algorithmic phase transition of random graph alignment problem},
  author={Du, Hang and Gong, Shuyang and Huang, Rundong},
  journal={Probability Theory and Related Fields},
  pages={1--56},
  year={2025},
  publisher={Springer}
}

@article{nie2021randomized,
  title={Randomized greedy algorithm for independent sets in regular uniform hypergraphs with large girth},
  author={Nie, Jiaxi and Verstra{\"e}te, Jacques},
  journal={Random Structures \& Algorithms},
  volume={59},
  number={1},
  pages={79--95},
  year={2021},
  publisher={Wiley Online Library}
}

@unpublished{bal2023larger,
  title={Larger matchings and independent sets in regular uniform hypergraphs of high girth},
  author={Bal, Deepak and Bennett, Patrick},
  howpublished={\url{https://arxiv.org/abs/2307.15601}},
  year={2023}
}

@article{pal2021community,
  title={Community detection in the sparse hypergraph stochastic block model},
  author={Pal, Soumik and Zhu, Yizhe},
  journal={Random Structures \& Algorithms},
  volume={59},
  number={3},
  pages={407--463},
  year={2021},
  publisher={Wiley Online Library}
}

@inproceedings{huang2022tight,
  title={Tight lipschitz hardness for optimizing mean field spin glasses},
  author={Huang, Brice and Sellke, Mark},
  booktitle={2022 IEEE 63rd Annual Symposium on Foundations of Computer Science (FOCS)},
  pages={312--322},
  year={2022},
  organization={IEEE}
}

@inproceedings{gao2007minimum,
  title={The minimum consistent subset cover problem and its applications in data mining},
  author={Gao, Byron J and Ester, Martin and Cai, Jin-Yi and Schulte, Oliver and Xiong, Hui},
  booktitle={Proceedings of the 13th ACM SIGKDD international conference on Knowledge discovery and data mining},
  pages={310--319},
  year={2007}
}

@inproceedings{gu2021towards,
  title={Towards computing a near-maximum weighted independent set on massive graphs},
  author={Gu, Jiewei and Zheng, Weiguo and Cai, Yuzheng and Peng, Peng},
  booktitle={Proceedings of the 27th ACM SIGKDD Conference on Knowledge Discovery \& Data Mining},
  pages={467--477},
  year={2021}
}

@inproceedings{brown1998determining,
  title={Determining a Minimal and Independent Set of Image Processing Operations for a Multimedia Database System},
  author={Brown, Leonard and Gruenwald, Le},
  booktitle={Proceedings of the 1998 Energy Technology Conference and Exhibition},
  year={1998}
}

@article{luby1993removing,
  title={Removing randomness in parallel computation without a processor penalty},
  author={Luby, Michael},
  journal={Journal of Computer and System Sciences},
  volume={47},
  number={2},
  pages={250--286},
  year={1993},
  publisher={Elsevier}
}

@misc{bordenave2016lecture,
  title={Lecture notes on random graphs and probabilistic combinatorial optimization. Version April 8, 2016},
  author={Bordenave, C},
  year={2016}
}

@article{bordenave2023detection,
  title={Detection thresholds in very sparse matrix completion},
  author={Bordenave, Charles and Coste, Simon and Nadakuditi, Raj Rao},
  journal={Foundations of Computational Mathematics},
  volume={23},
  number={5},
  pages={1619--1743},
  year={2023},
  publisher={Springer}
}

@misc{boucheron2013concentration,
  title={Concentration Inequalities},
  author={Boucheron, S and Lugosi, G and Massart, P},
  year={2013},
  publisher={Oxford University Press}
}
